              \def\version{4 November 2022}	        	%

\documentclass[reqno,11pt]{amsart} 
\usepackage[T1]{fontenc}
\usepackage[utf8]{inputenc}
\usepackage{amsmath,amsthm}

\usepackage[mathscr]{eucal}
\usepackage{amssymb}
\usepackage{srcltx} 
\usepackage{dsfont}
\usepackage{hyperref}
\usepackage{xcolor}
\usepackage{enumerate}
\usepackage{tikz, pgfplots}
\usepackage{caption}
\usepackage{subcaption}
\usepackage{comment}
\usepackage{lscape}
\usepackage{graphicx}
\usepackage{tabularx}
\pgfplotsset{compat=1.11}
\usepgfplotslibrary{fillbetween}
\usetikzlibrary{intersections}
\pgfdeclarelayer{bg}
\pgfdeclarelayer{ft}
\pgfsetlayers{bg,main,ft}

\numberwithin{equation}{section}
 
\def\barn1a{\bar n_{1a}}
\def\emptyset{\varnothing} 

\def\d{{\rm d}} 
\def\e{\varepsilon} 
 
\newfam\Bbbfam 
\font\tenBbb=msbm10 
\font\sevenBbb=msbm7 
\font\fiveBbb=msbm5 
\textfont\Bbbfam=\tenBbb 
\scriptfont\Bbbfam=\sevenBbb 
\scriptscriptfont\Bbbfam=\fiveBbb

\def\2{\mathbf 2}

\newcommand{\R}     {\mathbb{R}} 
 
\newcommand{\N}     {\mathbb{N}} 
\renewcommand{\P}   {\mathbb{P}} 
 
\newcommand{\E}     {\mathbb{E}}

\newcommand{\smfrac}[2]{\textstyle{\frac {#1}{#2}}}

\def\1{{\mathchoice {1\mskip-4mu\mathrm l}      
{1\mskip-4mu\mathrm l} 
{1\mskip-4.5mu\mathrm l} {1\mskip-5mu\mathrm l}}} 
 
\def\comment#1{} 
\newtheoremstyle{thm}{2ex}{2ex}{\itshape\rmfamily}{} 
{\bfseries\rmfamily}{}{1.7ex}{} 
 
\newtheoremstyle{rem}{1.3ex}{1.3ex}{\rmfamily}{} 
{\itshape\rmfamily}{}{1.5ex}{}

\renewcommand{\theequation}{\thesection.\arabic{equation}} 
 
\newtheorem{theorem}{Theorem}[section] 
\newtheorem{lemma}[theorem]{Lemma} 
\newtheorem{prop}[theorem] {Proposition} 
\newtheorem{cor}[theorem]  {Corollary}

\theoremstyle{definition}
 
\newtheorem{example}[theorem] {Example}

\newtheorem{remark}[theorem]{Remark}

\renewcommand{\d}{{\rm d}} 
 
\newcommand{\eps}{\varepsilon}

\newcommand{\Bcal}  {{\mathcal B}}

\newcommand{\Lcal}   {{\mathcal L }}

\newcommand\numberthis{\addtocounter{equation}{1}\tag{\theequation}}
\renewcommand{\e}   {{\operatorname e }}

\definecolor{Red}{rgb}{1,0,0}

\begin{document} 
 
\title[Virus-induced dormancy]{Microbial virus epidemics in the presence of contact-mediated host dormancy}
\author[Jochen Blath and András Tóbiás]{}
\maketitle
\thispagestyle{empty}
\vspace{-0.5cm}

\centerline{\sc Jochen Blath{\footnote{Goethe-Universität Frankfurt, Robert-Mayer-Straße 10, 60325 Frankfurt am Main, Germany, {\tt blath@math.uni-frankfurt.de}}} and András Tóbiás{\footnote{Department of Computer Science and Information Theory, Budapest University of Technology and Economics,
Műegyetem rkp. 3., 1111 Budapest, Hungary, {\tt tobias@cs.bme.hu}}}}
\renewcommand{\thefootnote}{}

\bigskip

\centerline{\small(\version)} 
\vspace{.5cm} 
 
\begin{quote} 
{\small {\bf Abstract:}} We investigate a stochastic individual-based model for the population dynamics of host--virus systems where the microbial hosts may transition into a dormant state upon contact with virions, thus evading infection. Such a contact-mediated defence mechanism was described in Bautista et al (2015) for an archaeal host, while Jackson and Fineran (2019) and Meeske et al (2019) describe a related, CRISPR-Cas induced, dormancy defense of bacterial hosts to curb phage epidemics.
We first analyse the 
effect of the dormancy-related model parameters on  
the probability and time of invasion of a newly arriving virus into a resident host population.
Given successful invasion in the stochastic system, we then show that the emergence (with high probability) 
of  a persistent virus infection (`epidemic') in a large host population can be 
determined by the existence of a coexistence equilibrium for the dynamical system arising as the deterministic many-particle limit of our model. This is an extension of a dynamical system  considered by Beretta and Kuang (1998) that is known to exhibit a Hopf bifurcation, giving rise to a ‘paradox of enrichment’. In our system, we verify that the additional dormancy component can, at least for certain parameter ranges, prevent the associated loss of stability.
Finally, we show that the presence of contact-mediated dormancy enables the host population to attain higher equilibrium sizes  -- while still being able to avoid a persistent epidemic -- than host populations without this trait.



\end{quote}




\setcounter{tocdepth}{3}


\setcounter{section}{0}
\begin{comment}{
This is not visible.}
\end{comment}


\section{Introduction}\label{sec-introductionHGT}

{\bf Motivation and background.}
The abstract concept of `dormancy' describes the ability of an organism to switch into a reversible state of low to vanishing metabolic activity. This strategy to cope with adverse environmental conditions is wide-spread among many taxa, comes in many different forms, and is employed in particular by many microorganisms \cite{LJ11}, \cite{LdHWB21}. The resulting `seed banks' comprised of dormant individuals have profound effects on the evolutionary and ecological behaviour of populations, in particular increasing diversity and resilience against various forms of external stress.

The mathematical analysis of the effects of dormancy in ecology and evolution via dynamical systems,  and increasingly also via stochastic individual based models, has been an active field of research for several decades. One of the basic paradigms is that dormancy, and the resulting seed banks, can be highly beneficial in {\em fluctuating environments}, where they can often be understood as bet hedging strategies. This has been confirmed by abstract theory many times beginning with the delayed seed germination model of Cohen \cite{Cohen1966}. Modelling has grown significantly since then, incorporating e.g.\ the related concepts of `spontaneous vs.\ responsive transitioning' and `phenotypic plasticity', often in the context of microbial populations (e.g.\ \cite{Balaban2004, KusselLeibler2005, MS08}). Overall, random environmental fluctuations can give rise to an interesting panorama of optimal dormancy initiation and resuscitation strategies, see e.g.\ \cite{DMB11, BHHS19+}.


However, the assumption that external environmental fluctuations are necessary for dormancy to be an evolutionary successful strategy has also been challenged. Several mathematical models show that competitive pressure for resources \cite{Ellner1987, LR06, BT19} or certain predator-prey dynamics \cite{Tan2020} may also favor dormancy even in the absence of additional abiotic variation. 
For example, predator dormancy has  been shown to be able to prevent the occurrence of the `paradox of enrichment' \cite{KMO09}, thus stabilizing the coexistence regimes in predator-prey systems.


In this paper, we investigate a further scenario, in which dormancy enters as a defence strategy  of host cells against virus attacks. For example, it has been reported that infected bacteria can enter a dormant state as part of a CRISPR-Cas immune response, thereby curbing phage epidemics (cf.\ \cite{JF19} resp.\ \cite{MNM19}). Moreover, it has been suggested that dormancy of hosts may even be initiated upon mere contact of virus particles with their cell hull, so that the dormant host may entirely avoid infection, cf.\ Bautista et al \cite{B15}. Indeed, in experiments, Bautista et\ al\ observed that  {\em Sulfolobus islandicus} (an archeon) populations may switch almost entirely into dormancy within hours after being exposed to the {\em Sulfolobus spindle-shape virus SSV9}, even when the initial virus-to-host ratio is relatively small.
The authors argue that this highly sensitive anti-viral response 
should be taken into account in models for virus-host interactions so that its ecological consequences can be understood.

%
%
%
%


A first step in this direction was taken by 
Gulbudak and Weitz \cite{GW15} who  provide a biophysical model for the `early stages' (covering a few hours) of the above host--virus dynamics. Indeed, their deterministic model 
can reproduce the observed rapid switches into dormancy for relatively small virus-to-host ratios. 
However, their model is focussed on a relatively short `time-window' of  host--virus dynamics and neither allows for a stochastic invasion analysis involving low numbers of newly arriving virions (where random fluctuations play an important role), nor 
virus reproduction via host cells, which would be necessary for a `long-time' analysis of the system.

%


%
%
Here, we follow up on the suggestion by \cite{B15} to further investigate the consequences of virus-induced host dormancy 
in two directions: i) We include explicit individual-based stochasticity, which is relevant during the early phases of an emerging virus epidemic (that is, when only few or even single virus particles arrive in the host population), and ii) incorporate a mechanism for (lytic) virus reproduction which allows an analysis of the coexistence / extinction regimes of the virus population. Once the virus epidemic becomes `macroscopic' (that is, the number of virus particles is at least of the order of the resident population -- a `successful invasion'), the stochastic model can then be approximated by a deterministic dynamical system. 
This limiting system extends a model of \cite{BK98}, where a lytic virus infection against single-cell hosts was studied. Their model, which is dormancy-free and also excludes the recovery of infected individuals, also motivated our present study.

\medskip
{\
{\bf Main goals.} The aim of this paper is to give at least partial answers to the following questions about our model. 
\begin{enumerate}[Q 1)]
\item\label{first-questionsvirus} Under which conditions is an invasion, starting with the arrival of a single virion, into a (large) resident host population possible with high probability? 
Here, by invasion we mean that the virus population reaches a level `visible' on the scale of the carrying capacity of the host population, thus initiating a `macroscopic epidemic'. 
This refers to a first, stochastic phase of the infection, where random fluctuation play a crucial role for the establishment of an epidemic.
\item How long does a successful invasion typically  take?
\item Given successful invasion, 
what is the dynamical system that corresponds to the many-particle limit of the stochastic individual based model? This system can then be used to describe the dynamics of the model after the initial `stochastic phase', where now random fluctuations become negligible.
\item\label{lastnormal-questionsvirus} During this second deterministic phase,  under which conditions does the virus epidemic either become persistent (both host- and virus populations maintain macroscopic sizes over `long' time-intervals)  or break down (extinction of the virus population)?
\item\label{firstbifurcation-questions} In the case of persistence, what is 
the long-term behaviour of the 
system (i.e.\ stable coexistence vs.\  periodic/chaotic behaviour, emergence of a `paradox of enrichment' phenomenon)? What are the consequences of dormancy and other model extensions in comparison to earlier models?
\end{enumerate}
For all of these questions, we are particularly interested in the specific roles of the dormancy-related model parameters. 
While we aim for mathematical results whenever possible, some questions regarding the long-term behaviour will be attacked `only' via simulation, sometimes leading to conjectures that invite further theoretical work.

}
%

%

{\bf Organization of the paper.}
In Section~\ref{sec:model_def_main_results} we introduce our model, we provide some heuristics, and we present our main results. In particular, Section~\ref{sec-modeldefvirus} contains the definition of our model. Based on some preliminary results on the underlying dynamical system in Section~\ref{ssn-dynsyst} and on some crucial branching processes in Section~\ref{sec-phase1heuristics},  in Section~\ref{ssn:phases} we describe the two phases of a successful virus epidemic and set up the notation that is needed for the statement of our results in Section~\ref{sec-resultsvirus}. In Section~\ref{sec-discussionvirus} we discuss multiple aspects of our model and results.

In Section~\ref{sec-dynsyst2} we focus on questions regarding the long-term behaviour of the system that our main assertions keep unanswered, providing some conjectures and partial results in Section~\ref{sec-bifurcations}, numerical results in Section~\ref{sec-simulationsvirus}, and a comparison with predator--prey systems in Section~\ref{sec-paradoxofenrichment}. Finally, in Section~\ref{sec-proofsvirus} we carry out the proofs of our results.

\section{Model definition, heuristics, and main results}
\label{sec:model_def_main_results}

\subsection{A stochastic individual-based model for host dormancy}
\label{sec-modeldefvirus}

\tikzstyle{1a}=[circle,draw=blue!50,fill=blue!20,thick,minimum size=5mm]
\tikzstyle{1d}=[circle,draw=black!50,fill=black!20,thick,minimum size=5mm]
\tikzstyle{2}=[circle,draw=green!50,fill=green!20,thick,minimum size=5mm]
\tikzstyle{D}=[rectangle,draw=black!50,fill=white!20,thick,minimum size=5mm]
\tikzstyle{v}=[circle,draw=red!50,fill=red!20,thick,minimum size=2mm]
\tikzstyle{1i}=[circle,draw=blue!50,fill=purple!20,thick,minimum size=2mm]



We consider a  host--virus population 
consisting of individuals of four types: Active 
hosts (type 1a), dormant hosts (type 1d), infected hosts  (type 1i) and free virions  (type 2). 
The stochastic dynamics of the system is given as follows:
\begin{itemize}
	\item Active (1a) host cells reproduce via binary fission at rate $\lambda_1>0$ and die at rate $\mu_1 \in (0,\lambda_1)$. Dormant (1d) or infected (1i) cells do not reproduce (cf.\ Figure \ref{fig:repro}, left panel).
	\item Virions (2) do not reproduce individually (but instead indirectly via infection of a host cell, see below) and die/degrade at rate $\mu_2>0$ (cf.\  Figure \ref{fig:repro}, right panel).  
	
	 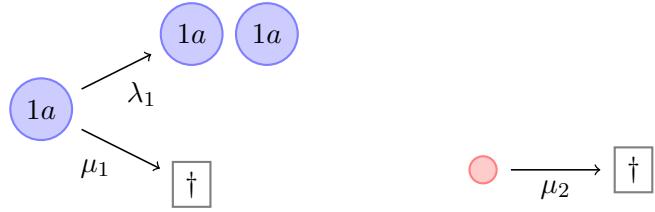
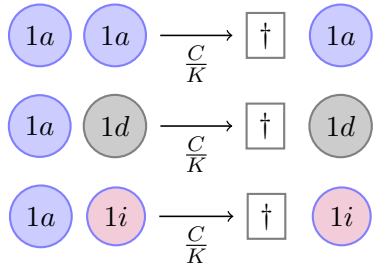
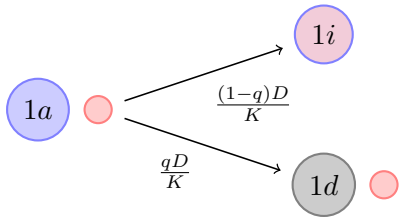
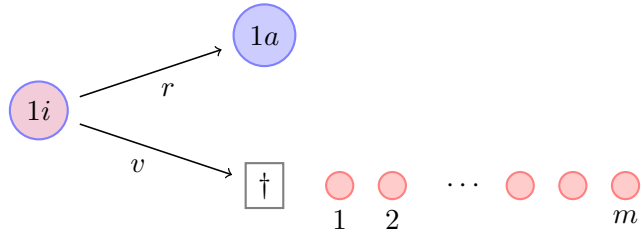
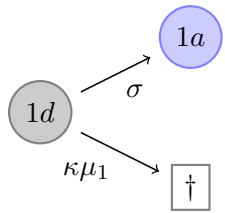
\begin{figure}
	 \vspace{1cm}
	 \begin{subfigure}{1.\textwidth}
	 \centering
 \begin{tikzpicture}[node distance=2cm, semithick, shorten >=5pt, shorten <=5pt]
  \node  [1a] 	(1a_rep)	 {$1a$};
  \node  [1a] 	(1a_child_1)	[right of=1a_rep, yshift=1cm]	 {$1a$};    
  \node  [1a] 	(1a_child_2) 	[right of=1a_child_1, xshift=-1cm] 	{$1a$};
  \node  [D] 		(death) 	[below of=1a_child_1] 	{$\dagger$}; 
  	 \draw[->] (1a_rep) to node[auto, swap] {$\lambda_1$} (1a_child_1);
 	 \draw[->] (1a_rep) to node[auto, swap] {$\mu_1$} (death);  
 \end{tikzpicture}
 \hspace{2cm}
  \begin{tikzpicture}[node distance=2cm, semithick, shorten >=5pt, shorten <=5pt]
  \node  [v] 	(virus)	{}; 
  \node  [D] 		(death) 	[right of=virus] 	{$\dagger$}; 
 	 \draw[->] (virus) to node[auto, swap] {$\mu_2$} (death);  
 \end{tikzpicture}
 \caption{Clonal reproduction of active hosts (type 1a) at rate $\lambda_1$ resp.\ death at rate $\mu_1$ (left), and virus degradation at rate $\mu_2$ (right). The symbol $\dagger$ represents the `death state'.  \vspace{0.5cm}} \label{fig:repro}
 \end{subfigure}
  \vspace{0.5cm}
 \begin{subfigure}{1.\textwidth}
   \centering
     \begin{tikzpicture}[node distance=2cm, semithick, shorten >=5pt, shorten <=5pt]
  \node  [1a] 	(2_comp_1)	 {$1a$};  
  \node  [1a] 	(2_comp_2)	 [right of=2_comp_1, xshift=-1cm]{$1a$};
  \node  [1a] 	(2_child_1)	[right of=death]	 {$1a$};    
  \node  [D] 		(death) 	[right of=2_comp_2] 	{$\dagger$}; 
  	 \draw[->] (2_comp_2) to node[auto, swap] {$\frac CK$} (death);
 \end{tikzpicture} \\
  \begin{tikzpicture}[node distance=2cm, semithick, shorten >=5pt, shorten <=5pt]
  \node  [1a] 	(2_comp_1)	 {$1a$};  
  \node  [1d] 	(2_comp_2)	 [right of=2_comp_1, xshift=-1cm]{$1d$};
  \node  [1d] 	(2_child_1)	[right of=death,xshift=-1cm]	 {$1d$};
  \node  [D] 		(death) 	[right of=2_comp_2] 	{$\dagger$}; 
  	 \draw[->] (2_comp_2) to node[auto, swap] {$\frac CK$} (death);
 \end{tikzpicture} \\
      \begin{tikzpicture}[node distance=2cm, semithick, shorten >=5pt, shorten <=5pt]
  \node  [1a] 	(2_comp_1)	 {$1a$};  
  \node  [1i] 	(2_comp_2)	 [right of=2_comp_1, xshift=-1cm]{$1i$};
  \node  [1i] 	(2_child_1)	[right of=death,xshift=-1cm]	 {$1i$};    
  \node  [D] 		(death) 	[right of=2_comp_2] 	{$\dagger$}; 
  	 \draw[->] (2_comp_2) to node[auto, swap] {$\frac CK$} (death);
 \end{tikzpicture}
     \caption{Competition events for pairs of host cells, involving the death of the first (type $1a$) individual at rate $C/K$.} \label{fig:comp1}
      \end{subfigure}
        \vspace{0.5cm}
	\begin{subfigure}{1.\textwidth}
	\centering
 \begin{tikzpicture}[node distance=2cm, semithick, shorten >=5pt, shorten <=5pt]
  \node  [1a] 	(1a_rep)	 {$1a$};
  \node  [v] 	(virus)	[right of=1a_rep, xshift=-1.2cm]	 {};    
  \node  [1i] 	(infected)	[right of=virus,yshift=1cm, xshift=1cm]	 {$1i$};      
  	 \draw[->] (virus) to node[auto, swap] {\small $\frac{(1-q)D}{K}$} (infected);
  \node  [1d] 	(virus_2)	[below of=infected]	 {$1d$};    
  \node  [v] 		(dormant) 	[right of=virus_2, xshift=-1.2cm] 	{}; 
  	 \draw[->] (virus) to node[auto, swap] {\small $\frac{qD}{K}$} (virus_2);	 
 \end{tikzpicture}
	\caption{Virus attack resulting in infection (rate $(1-q)D/K$) of host or dormancy (rate $qD/K$).} \label{fig:attack}
      \end{subfigure}
        \vspace{0.5cm}
      \begin{subfigure}{1.\textwidth}
      	\centering
  \begin{tikzpicture}[node distance=2cm, ,semithick, shorten >=5pt, shorten <=5pt]
  \node  [1i]  	(infected)	 {$1i$};  
  \node  [1a] 	(recovered)	[right of=infected,yshift=1cm, xshift=1cm]	 {$1a$};    
  	 \draw[->] (infected) to node[auto, swap] {$r$} (recovered);
  \node  [D] 	 (death)	[below of=recovered]	 {$\dagger$};    
  \node  [v] 		(virus_1) 	[right of=death, xshift=-1cm,  label=below:$1$] 	{}; 
  	 \draw[->] (infected) to node[auto, swap] {$v$} (death);	 
  \node  [v] 		(virus_2) 	[right of=virus_1, xshift=-1.3cm,  label=below:$2$, label=right:$\quad\cdots$] 	{}; 
  \node  [v] 		(virus_3) 	[right of=virus_2, xshift=-0.3cm] 	{}; 
  \node  [v] 		(virus_4) 	[right of=virus_3, xshift=-1.3cm] 	{}; 
  \node  [v] 		(virus_m) 	[right of=virus_4, xshift=-1.3cm,  label=below:$m$] 	{}; 
 \end{tikzpicture}
	\caption{ Infected cells recover (at rate $r>0$) or release $m$ virions after lysis, at rate $v>0$.}   \label{fig:inf} 
\end{subfigure}
   \begin{subfigure}{1.\textwidth}
   	\centering
  \begin{tikzpicture}[node distance=2cm, semithick, shorten >=5pt, shorten <=5pt]
  \node  [1d] (dormant)	 {$1d$};  
  \node  [1a] (active)	 [right of=dormant, yshift=1cm]{$1a$};
  \node  [D] 	(death)	[below of=active]	 {$\dagger$};    
  	 \draw[->] (dormant) to node[auto, swap] {$\sigma$} (active);
  	 \draw[->] (dormant) to node[auto, swap] {$\kappa \mu_1$} (death);	 
 \end{tikzpicture} 
\caption{Leaving the dormant state $1d$ by resuscitation (rate $\sigma$) or death (rate $\kappa \mu_1$).}  \label{fig:dorm} 
      \end{subfigure}
      \caption{Overview of transitions of the host--virus model.} \label{fig:model}
      \end{figure}	
	\item Competition: Fix parameters $K>0$ called \emph{carrying capacity} and  $C>0$ called {\em competition strength}. For any ordered pair consisting of one active (1a) host cell and one other host cell (of either type 1a or 1i or 1d), at rate $C/K$, a death due to competition/overcrowding happens, affecting the active individual, which is removed from the population (cf.\ Figure \ref{fig:comp1}).

	\item Virus attack: Fix a parameter  $q \in (0,1)$ called {\em dormancy initiation probability}.  For any ordered pair of individuals containing one active (1a) cell and one virion (2),  a virus attack happens at rate $D/K$ for some $D>0$. In this case, with probability $q$, the attacked host (1a) senses the virion (e.g.\ upon contact with its cell-hull) and is able to switch into dormancy (from 1a to 1d) before infection, and with probability $1-q$, the host cell gets infected (i.e.\ switches from 1a to 1i) and the free virus (2) is `removed' (in the sense that it enters the cell), see Figure \ref{fig:attack}.

	\item Infected cells (1i) either recover (at rate $r>0$) or produce $m \in \N$ (where typically $m$ is large) new virions (2) and then dissolve (lysis), at rate $v>0$ (see Figure~\ref{fig:inf}).

	\item Dormant cells (1d) resuscitate into (1a) at rate $\sigma>0$ and die at rate $\kappa\mu_1$ for some $\kappa \geq 0$ (Figure \ref{fig:dorm}).

\end{itemize}


Note that in this model, there is no classical meaning of `fitness' for the virions (type 2), since the reproduction of this type of individual rests entirely on host availability. As indicated in the introduction, such a 
reproduction mechanism reflects lysis, 
and we refer to $m$ as average \emph{burst size} (cf.~e.g.\ \cite{B09} for burst sizes in archea). 


The corresponding population process can be formally defined as a continuous time Markov chain $\mathbf N=(\mathbf N_t)_{t\geq 0}$ on $ \N_0^4$, where 
$$
(\mathbf N_t)_{t \geq 0} = (N_{1a,t},N_{1d,t},N_{1i,t}, N_{2,t})_{t \geq 0}
$$
will be interpreted as
$$
 N_{x,t} = \# \{\mbox{individuals of type $x$ alive at time $t$} \},
$$
for $x \in \{ 1a, 1d, 1i, 2\}$.

According to the above description, $\mathbf N$ is then the unique time continuous Markov process with transitions
\begin{align*}
	(n_{1a},n_{1d},n_{1i}, n_2) \to
	\begin{cases}
		& (n_{1a}+1,n_{1d}, n_{1i}, n_{2}) \text{ at rate }  n_{1a}\lambda_1, \\
		& (n_{1a}-1,n_{1d}, n_{1i}, n_{2}) \text{ at rate }  n_{1a}(\mu_{1} +  C \frac{n_{1a}+n_{1d}+n_{1i}}{K}), \\
		& (n_{1a},n_{1d}, n_{1i}, n_{2}-1) \text{ at rate }  n_{2}\mu_{2}, \\
		& (n_{1a}-1,n_{1d}, n_{1i}+1, n_{2}-1) \text{ at rate } \frac{(1-q) D n_{1a}n_{2}}{K}, \\
				& (n_{1a}-1,n_{1d}+1, n_{1i}, n_{2}) \text{ at rate } \frac{q D n_{1a} n_{2}}{K}, \\
			& (n_{1a}+1,n_{1d}, n_{1i}-1, n_{2}) \text{ at rate } r n_{1i}, \\
		& (n_{1a},n_{1d}, n_{1i}-1, n_{2}+m) \text{ at rate } v n_{1i}, \\
				& (n_{1a}+1,n_{1d}-1, n_{1i}, n_{2}) \text{ at rate } n_{1d}\sigma,\\
				& (n_{1a},n_{1d}-1, n_{1i}, n_{2}) \text{ at rate }  n_{1d} \kappa \mu_{1}. \\
	\end{cases}
\end{align*}
Let us mention some elementary properties of this Markov chain. Its only absorbing state is $(0,0,0,0)$, which corresponds to the extinction of all the four types. Moreover, if $\mathbf N_0 \in [0,\infty) \times \{ 0\}^3$, then $\mathbf N_t \in [0,\infty) \times \{ 0 \}^3$ for all $t>0$, and if  $\mathbf N_0 \in [0,\infty)^2 \times \{ 0\}^2$, then  $\mathbf N_t \in [0,\infty)^2 \times \{ 0 \}^2$ for all $t>0$. Finally, in this specific case, $t \mapsto N_{1d,t}$ is monotonically decreasing and the expected time until it reaches 0 is finite. In other words, if there are initially neither infected individuals nor viruses, then this will also be the case for all positive times, and the dormant population will vanish rapidly.

Our goal is to analyze a situation where $K$ is large (even the limit as $K \to \infty$, corresponding to a many-particles limit), and the initial size of the (scaled) host population $ N_{1a,0}^K$ is close to its (virus-free) equilibrium. 
We will thus consider the rescaled process
$$
(\mathbf N_t^K)_{t \geq 0} = (N_{1a,t}^K,N_{1d,t}^K,N_{1i,t}^K, N_{2,t}^K)_{t \geq 0},
$$
which is defined via
$$
 N_{x,t}^K = \frac{1}{K} \# \{\mbox{individuals of type $x$ alive at time $t$} \},
$$
so that $\mathbf N_t^K = \mathbf N_t/K$ for all $K,t>0$, recalling that $K>0$ is the carrying capacity of the system.
We also write
$$ 
N_{1,t}^K = N_{1a,t}^K + N_{1d,t}^K + N_{1i,t}^K
$$
for the total population size of the host individuals (scaled by $K$).

\subsection{The limiting dynamical system for large populations: equilibria and stability}\label{ssn-dynsyst}
Once a macroscopic epidemic has emerged, that is, 
the population size of the virus particles and the host cells are both of order $K$ (for large $K$), then the Markov chain $(\mathbf N_t^K)_{t \geq 0}$ satisfies a `functional law of large numbers' and  can be  approximated by a limiting deterministic dynamical system. We now introduce this scaling limit. Assume that the initial population size of $\mathbf N_0^K$ satisfies
$$
\frac 1K \mathbf N_0 = \mathbf N_0^K \overset{K \rightarrow \infty}{\longrightarrow}  {\bf n}(0)=(n_{1a}(0),n_{1d}(0),n_{1i}(0),n_{2}(0)) \, \, \in \, [0, \infty)^4
$$
and fix some $T>0$.
Then,  \cite[Theorem 11.2.1, p.~456]{EK}) implies the weak convergence (uniformly on $[0,T]$) 
$$
(\mathbf N_t^K)_{t \in [0,T]} \overset{K \to \infty}{\Longrightarrow} (\mathbf n(t))_{t \in [0,T]}=((n_{1a}(t),n_{1d}(t),n_{1i}(t),n_{2}(t)))_{t \in [0,T]}
$$
to the unique solution $({\mathbf n}(t))_{t \geq 0}$ 
of the dynamical system  
\begin{equation}\label{4dimvirus}
\begin{aligned}
\frac{\d n_{1a}(t)}{\d t} & = n_{1a}(t)\big( \lambda_1-\mu_1-C {(n_{1a}(t)+n_{1i}(t)+n_{1d}(t))}-D n_{2}(t) \big)  + \sigma n_{1d}(t) + r n_{1i}(t),\\
\frac{\d n_{1d}(t)}{\d t} & = q D n_{1a}(t)  n_{2}(t) - (\kappa\mu_1+\sigma) n_{1d}(t), \\
\frac{\d n_{1i}(t)}{\d t} & = (1-q) D n_{1a}(t) n_{2}(t) -(r+v) n_{1i}(t), \\
\frac{\d n_{2}(t)}{\d t} & = mv n_{1i}(t) - (1-q)  D n_{1a}(t) n_{2}(t) -  \mu_2 n_{2}(t)
\end{aligned}
\end{equation}
with parameters as in the previous section.
%
Define $$\barn1a:=\smfrac{\lambda_1-\mu_1}{C}.$$ Then, $(0,0,0,0)$ (the zero or `extinction' equilibrium) and $(\barn1a,0,0,0)$ (the `virus-free' equilibrium, e.g.\ after complete recovery) are equilibria of the dynamical system~\eqref{4dimvirus}, and they are distinct and both coordinatewise nonnegative thanks to our assumption that $\lambda_1>\mu_1$. Moreover, we have the following proposition.

\begin{prop}[Coexistence condition]\label{lemma-coexistencevirus}
The system~\eqref{4dimvirus} has a coordinatewise positive coexistence equilibrium $(n_{1a}^*,n_{1d}^*,n_{1i}^*,n_2^*)$ if and only if the \emph{coexistence condition}
\[ 
(m v-(r+v))(1-q) D\barn1a > \mu_2(r+v)
\numberthis\label{viruscoexcond} 
\]
holds. In this case, the coexistence equilibrium is unique, and its active coordinate $n_{1a}^*$ is given by
\[ 
n_{1a}^* = \frac{\mu_2(r+v)}{(1-q) D(mv-(r+v))} < \barn1a. 
\numberthis\label{n1adefvirus} 
\]
\end{prop}
 We refer the reader to Section~\ref{sec-preliminaryproofs} for the simple proof of this proposition, 
which also provides an explicit characterization of the coordinates $n_{1d}^*,n_{1i}^*,n_2^*$ of the coexistence equilibrium. 
Note that the coexistence condition is not only equivalent to the existence of a coexistence equilibrium for~\eqref{4dimvirus}, but will also guarantee the supercriticality of a multitype branching process crucial for our analysis, cf.~Section~\ref{sec-phase1heuristics} below.  
%
The {\em critical case} where \eqref{viruscoexcond} holds with equality is then given by
\[ 
\frac{\lambda_1-\mu_1}{C}=\frac{\mu_2(r+v)}{(1-q)D(mv-(r+v))}. 
\numberthis\label{mstar}
\]
Given all parameters but $m$, the value of $m >0$ ensuring that \eqref{mstar} holds will be denoted by $m^*$ throughout the rest of the paper. Then,~\eqref{viruscoexcond} is equivalent to $m>m^*$. We refer to  $m^*$ as the \emph{transcritical bifurcation point} (this notion will be justified in Section~\ref{sec-bifurcations}).


Let us now analyse the stability of the equilibria $(0,0,0,0)$, $(\barn1a,0,0,0)$, and $(n_{1a}^*,n_{1d}^*,n_{1i}^*,n_2^*)$ of the system~\eqref{4dimvirus} (whenever they exist). 
\begin{prop}[Stability of equilibria]\label{lemma-stabilityvirus}
The following assertions hold for the dynamical system~\eqref{4dimvirus}.
\begin{enumerate}
    \item\label{0saddle} The equilibrium $(0,0,0,0)$ is unstable.
    \item\label{nopersistence} The equilibrium $(\barn1a,0,0,0)$ is unstable if~\eqref{viruscoexcond} holds and asymptotically stable if the strict reverse inequality of~\eqref{viruscoexcond} holds, in other words,
    \[ m v(1-q) D\barn1a < ((1-q)D\barn1a+\mu_2)(r+v). \numberthis\label{lifecond}\]
    \item\label{2or4} Under condition~\eqref{viruscoexcond}, the Jacobi matrix of the system at $(n_{1a}^*,n_{1d}^*,n_{1i}^*,n_2^*)$ has positive determinant and negative trace.
\end{enumerate}
\end{prop}
The proof of Proposition~\ref{lemma-stabilityvirus} will also be carried out in Section \ref{sec-preliminaryproofs}. There, the local stability of $(0,0,0,0)$ and $(\barn1a,0,0,0)$ will simply be determined via linearization. In contrast, assertion \eqref{2or4} only implies that the Jacobi matrix at $(n_{1a}^*,n_{1d}^*,n_{1i}^*,n_2^*)$ has either 2 or 4 eigenvalues with negative real parts, and this equilibrium may indeed be stable or unstable depending on the choice of the parameters, as we will see below.

\begin{remark}[Emergence of a Hopf bifurcation: Stable endemic  vs.\ paradox of enrichment] \label{remark-bifurcation}

While $(n_{1a}^*,n_{1d}^*,n_{1i}^*,n_2^*)$ is always stable for $m>m^*$ sufficiently close to $m^*$ (cf.~Proposition~\ref{prop-firststable} below), 
it may lose stability due to a Hopf bifurcation for $m$ large and $r,q$ sufficiently small, as we will show in Section \ref{sec-bifurcations}.
Such a Hopf bifurcation gives rise to stable periodic trajectories, with minima for the active host population tending to zero as the burst size $m$ tends to infinity. 
Hence, an underlying in reality finite population may die out due to random fluctuations while being close to these small periodic minima. Thus, for the hosts, periodic behaviour for large $m$ can be worse than stable coexistence with the virus population (`endemic behaviour') for smaller $m$. This can be seen as a variant of the `paradox of enrichment' known from predator-prey models (cf.\ Section~\ref{sec-paradoxofenrichment}): Higher within-cell reproduction rates of the virus, encoded by higher average burst sizes in our model, may lead to more likely extinction of the whole system. This is of course rather intuitive in our host--virus set-up.\\
We refer to  Section \ref{sec-dynsyst2} for a more detailed discussion of this phenomenon. The presence of a Hopf bifurcation in the special case of the system~\eqref{4dimvirus} without dormancy and recovery, i.e.\ with $r=q=0$ (and with the dormant coordinate ignored) was shown in~\cite{BK98}. We will see in Section \ref{sec-dynsyst2} (via simulation and in special cases also rigorously) that the presence of dormancy, at least for sufficiently large dormancy initiation probabilities $q$, removes the Hopf bifurcation and guarantees that the coexistence equilibrium is stable for all $m>m^*$. This can be seen as one of the `positive' effects of dormancy for the long-term survival of the host population. We will also explain that the presence of recovery has the same effect, at least for large enough recovery parameters $r$ (compared to $v$).
%
%

\end{remark}

For all choices of parameters where the equilibrium $(n_{1a}^*,n_{1d}^*,n_{1i}^*,n_2^*)$ exists with four positive coordinates (coexistence), we can verify the following results, which will be crucial for our main theorems regarding the stochastic process $(\mathbf N_t)_{t \geq 0}$ in the limit $K\to\infty$. 
The first one tells us  that starting from an initial condition with only positive coordinates, the mono-type equilibrium $(\barn1a,0,0,0)$ will never be reached under the coexistence condition~\eqref{viruscoexcond}.

\begin{prop}[Non-extinction of the virus epidemic]\label{prop-Lyapunov}
Consider the dynamical system \eqref{4dimvirus}. Assume that \eqref{viruscoexcond} holds, and $(n_{1a}(0),n_{1d}(0),n_{1i}(0),n_2(0)) \in (0,\infty)^4$. Then $(n_{1a}(t),n_{1d}(t),n_{1i}(t),n_2(t))$ does not tend to $(\barn1a,0,0,0)$ as $t\to\infty$, not even along a diverging subsequence of time-points.
\end{prop}

Since coordinatewise nonnegative solutions of \eqref{4dimvirus} are bounded, Proposition~\ref{prop-Lyapunov} together with a simple compactness argument implies that started from any initial condition $(n_{1a}(0),n_{1d}(0),n_{1i}(0),n_2(0)) \in (0,\infty)^4$, there exists a $\varrho>0$ such that
\[ \liminf_{t \to \infty} \big\Vert (n_{1a}(t),n_{1d}(t),n_{1i}(t),n_2(t))-(\barn1a,0,0,0) \big\Vert_1 \geq \varrho. \numberthis\label{uniformstrongrepeller}\]

This assertion is known as $(\barn1a,0,0,0)$ being a \emph{uniform strong repeller}; cf.~\cite[Corollary 4.2]{BK98} for its analogue in the recovery- and dormancy-free three-dimensional case. The following corollary is analogous to \cite[Lemma 2.3 and Theorem 4.2]{BK98}, but since that paper provides no explicit proof and our setting is more complex, we present a proof for completeness in Section~\ref{sec-preliminaryproofs}, where we also verify the previous Proposition~\ref{prop-Lyapunov}. 

\begin{cor}[Population bounds]\label{cor-persistence}
Consider the dynamical system \eqref{4dimvirus}. Assume that \eqref{viruscoexcond} holds, and $(n_{1a}(0),n_{1d}(0),n_{1i}(0),n_2(0)) \in (0,\infty)^4$. Then $$\liminf_{t\to\infty} n_{j}(t)>0$$ holds for all $j \in \{ 1a, 1d,1i,2 \}$, and $$\limsup_{t\to\infty} \, n_{1a}(t)+n_{1d}(t)+n_{1i}(t) <\barn1a.$$ Further, $$\limsup_{t\to\infty} n_2(t) < \frac{mv\barn1a}{\mu_2}.$$
\end{cor}

The positivity of the $\liminf$'s of the coordinates $n_{1d}(t), n_{1i}(t),n_2(t)$ is called the uniform strong persistence of the system \eqref{4dimvirus}.
By our uniform approximation result, in this case, the macroscopic virus epidemic will also be present for long times (with high probability) in the stochastic model with large enough carrying capacities $K$.

\begin{remark}[Initial conditions in Corollary~\ref{cor-persistence}] \label{remark-positivity}
Let us emphasize that coordinate-wise positivity cannot be replaced by $n_{1a}(0)$ and at least one of the coordinates $n_{1d}(0), n_{1i}(0), n_2(0)$ being positive. For example, if $n_{1a}(0)>0$ and $n_{1d}(0)>0$ but $n_{1i}(0)=n_{2}(0)=0$, then the solution will converge to $(\barn1a,0,0,0)$. Yet, if $(n_{1d}(0),n_{1i}(0),n_2(0))\in [0,\infty)^3$ with $\max \{ n_{1i}(0),n_2(0) \}>0$, then $\min \{ n_{1d}(t),n_{1i}(t),n_2(t) \}>0$ holds for all $t>0$, and hence Proposition~\ref{prop-Lyapunov} and Corollary~\ref{cor-persistence} also hold for such initial conditions. 
\end{remark}

\subsection{The approximating branching process(es) in the initial stochastic phase}
\label{sec-phase1heuristics}
To prepare our main results, we also need to introduce the necessary tools to describe the behaviour of our system
during the initial stochastic phase. Assume that at time 0, a single virion arrives in a resident population of type 1a individuals with population size being close to its equilibrium  $K\barn1a$. Now, a standard approach (following e.g.~\cite{C+19,BT20}) in order to describe the invasion probability of the newly arrived virus particle is to couple its behaviour to that of a suitable branching process that is independent of the behaviour of the resident type 1a population as long as the virus population is still small. However, in our case we run into technical problems with this approach. We now outline the overall strategy, and provide details in subsequent sections.

In our scenario, it seems natural to consider type 1a as the `resident' population, and all three types 1d, 1i, and 2 as `invaders'. 
Indeed, the goal would then be to find a coupling to a dominated super-critical branching process that is  close enough to the original population so that its survival probability matches, in the limit, the invasion probability of the virion, while the type 1a population size (divided by $K$) stays close to $\barn1a$ with high probability. 


We will see in Section~\ref{sec-phase1virus} below that the three-type population size process $(N_{1d,t}, N_{1i,t}, N_{2,t})$ can indeed be coupled to a three-type linear branching process $$(\widehat{\mathbf N}(t))_{t \geq 0}=((\widehat N_{1d}(t),\widehat N_{1i}(t),\widehat N_{2}(t)))_{t \geq 0}$$ with jump rates obtained by replacing $N_{1a,t}^K$ with the fixed value $\barn1a$ in the rates corresponding to $K \mathbf N_t^K$.  More explicitly, this branching process has the following transition rates for $(x,y,z) \in \N_0^3$: 
\begin{itemize}
\item $(x,y,z) \to (x+1,y,z)$ at rate $q D \barn1a z$,
\item $(x,y,z) \to (x-1,y,z)$ at rate $(\kappa\mu_1+\sigma)x$,
\item $(x,y,z) \to (x,y+1,z-1)$ at rate $(1-q)D\barn1a z$,
\item $(x,y,z) \to (x,y-1,z)$ at rate $ry$,
\item $(x,y,z) \to (x,y-1,z+m)$ at rate $vy$,
\item $(x,y,z) \to (x,y,z-1)$ at rate $\mu_2 z$.
\end{itemize} 

The formal coupling of the two processes will be explained in Section~\ref{sec-phase1virus} (cf.\ the inequalities~\eqref{branchingcouplingvirus} and~\eqref{originalcouplingvirus}). It turns out that with high probability, the coupling holds until the population size $N_{1d,t}+N_{1i,t}+N_{2,t}$ of the infection-related types either reaches $\eps K$,
or goes extinct. (Here, $\eps>0$ is a small parameter that we will let tend to zero after carrying out the limit $K\to\infty$.)
After successful invasion,
the `early' stochastic phase of the epidemic ends, and the `macroscopic phase' starts, in which the rescaled (four-dimensional) population size process $\mathbf N_t^K$ can  be approximated by the solution of the dynamical system~\eqref{4dimvirus} with suitable initial condition. 

However, in our setting, when one wants to treat the early stochastic phase, a technical problem emerges. 
The proof techniques of the papers~\cite{C+19,BT20} strongly rely on the irreducibility of the mean matrix of the branching process under consideration. This condition implies that the branching process survives with positive probability when started from any initial condition with at least one positive coordinate.
Yet, the mean matrix of $(\widehat{\mathbf N}(t))_{t \geq 0}$ is given by
\[ 
J = \begin{pmatrix}
-\kappa\mu_1-\sigma & 0 & 0\\
0 & -r-v & mv \\
q D\barn1a  & (1-q)D\barn1a & -((1-q)D\barn1a+\mu_2) 
\end{pmatrix},
\numberthis\label{Jdefvirus}
\]
which is {\em not} irreducible. Further, we immediately see that $-\kappa\mu_1-\sigma<0$ is an eigenvalue of $J$ (with left eigenvector $(1,0,0)^T$). Hence, our branching process $(\widehat{\mathbf N}(t))_{t \geq 0}$ dies out whenever it starts from an initial condition of the form $(n,0,0)$, $n\in\N$, in other words, its dormant coordinate is sub-critical when being on its own. 

To be still able to carry through the standard program, we thus consider a further two-dimensional branching process with the same rates as for the $y$- and $z$-coordinates of $(\widehat{\mathbf N}(t))_{t \geq 0}$, ignoring all jumps of the form $(x,y,z)\to (x+1,y,z)$ and $(x,y,z) \to (x-1,y,z)$. Started from identical initial conditions, this is the same process as the two-dimensional projection $((\widehat N_{1i}(t),\widehat N_2(t))_{t \geq 0}$ of the branching process $(\widehat{\mathbf N}(t))_{t \geq 0}$. It has mean matrix 
\[ J_2= \begin{pmatrix}
-r-v &mv  \\
  (1-q)D\barn1a& -((1-q)D\barn1a+\mu_2),
\end{pmatrix} \numberthis\label{J2def} \]
from which we see that $(\widehat{\mathbf N}(t))_{t \geq 0}$ is supercritical if and only if $((\widehat N_{1i}(t),\widehat N_2(t))_{t \geq 0}$ is supercritical. 

The idea now is to change our perspective and to consider type 1d as `resident' instead of `invader'. We will then verify that initially, the two-type rescaled population $(N_{1a,t},N_{1d,t})$ stays close to its `equilibrium' $(\barn1a,0)$, so that $(N_{1i,t},N_{2,t})$ can be approximated by the two-type branching process $(\widehat N_{1i}(t),\widehat N_2(t))$, whose mean matrix $J_2$ {\em is} irreducible.  

Now, we have the following simple lemma, the proof of which will be found in Section~\ref{sec-branchingpreliminary}.
\begin{lemma}\label{lemma-lambdatilde}
The eigenvalues of $J_2$ are real, and the largest eigenvalue is given as 
\[ \widetilde \lambda = \smfrac{- (r+v+(1-q)D\barn1a+\mu_2) + \sqrt{((r+v+(1-q)D\barn1a+\mu_2)^2-4((r+v-mv)(1-q)D\barn1a+(r+v)\mu_2)}}{2}. \numberthis\label{lambdatildedefvirus} \] 
This is positive if and only if the coexistence condition~\eqref{viruscoexcond} holds. 
\end{lemma}
This lemma establishes a correspondence between the branching processes and the dynamical system~\eqref{4dimvirus}: The condition $\widetilde\lambda>0$ holds, i.e., the branching processes are supercritical (and thus have a nonzero survival probability) if and only if the dynamical system has a coordinatewise positive coexistence equilibrium. 
Since $-\kappa\mu-\sigma<0$, the largest eigenvalue of $J$ is always equal to $\widetilde\lambda$ if $\widetilde\lambda \geq 0$, otherwise it may be equal to $-\kappa\mu-\sigma$. In contrast, if the strict reverse inequality of~\eqref{viruscoexcond} holds, then $\widetilde\lambda<0$, and thus the branching processes are subcritical and we see quick extintion. 
Finally, if $\widetilde\lambda=0$, when \eqref{viruscoexcond} holds with an equality, the branching processes are critical and eventually go extinct, yet only after a potentially long time. This case is difficult to analyse, and we exclude it in the present paper.



We denote the {\em extinction probability} of the three-type branching process $(\widehat{\mathbf N}(t))_{t \geq 0}$, given that the branching process is started with a single virion (type 2 individual) and no infected or dormant type 1 individuals at time $t=0$, by
\[ s_2  := \P\big( \exists t < \infty \colon \widehat N_{1d}(t)+\widehat N_{1i}(t)+\widehat N_{2}(t)=0 \big|  (\widehat N_{1d}(0),\widehat N_{1i}(0),\widehat N_{2}(t))=(0,0,1) \big), \numberthis\label{qdefvirus} \]
This extinction probability $s_2$ can be computed by a standard first-step analysis, and for $\widetilde\lambda>0$ it is always less than 1, whereas for $\widetilde\lambda<0$ it equals 1, see Section~\ref{sec-1dormantinfected} for details.

\subsection{The phases of a `successful' virus epidemic}
\label{ssn:phases}

I) A `successful' virus epidemic consists of a first, stochastic phase, in which a newly arriving virion invades the resident population so that the free virions and the infected cells in total reach a population size `visible' on the order of the carrying capacity $K$. We formalize this time-point for $(\mathbf N_t)_{t \geq 0}$ by introducing a suitable stopping time  (with respect to its canonical filtration). For $\eps \geq 0$, we define
\[ T_\eps^2: = \inf \big\{ t \geq 0 \colon N_{1i,t}+N_{2,t} = \lfloor \eps K \rfloor \big\}. \numberthis\label{Teps2def} \]
Considering $\eps = 0$, the stopping time $T_0^2$ is the time of \emph{extinction of the epidemic}, i.e.\ when all infected individuals and virions have disappeared.

If $\eps>0$ is small, then with high probability as $K\to\infty$, until time $T_\eps^2 \wedge T_0^2$ the rescaled population size $N_{1a,t}^K$ of type 1a stays close to the equilibrium $\barn1a$, and the rescaled population size $N_{1d,t}^K$ of type 1d stays near zero (up to some error terms that are at most proportional to $\eps$). This enables us to approximate the population size $(N_{1i,t},N_{2,t})$ of infected individuals and virions via the two-type branching process $(\widehat N_{1i}(t),\widehat N_{2}(t))$ up to time $T_\eps^2 \wedge T_0^2$. With high probability in the limit $K\to\infty$ followed by $\eps \downarrow 0$, on the event $\{ T_0^2 < T_\eps^2 \}$ of an `unsuccessful invasion', the branching process also dies out at a time close to $T_0^2$. The probability of this event approaches the extinction probability $s_2$ of the branching process as $K\to\infty$. For $\widetilde\lambda>0$, in other words, for $s_2<1$, on the event 
$\{ T_\eps^2 < T_0^2 \}$ of a `successful invasion', the branching process typically also survives, and by irreducibility, typically both populations reach size $\eps K$ after roughly $\log K/\widetilde \lambda$ amount of time.


II) After successful invasion, the second, deterministic phase begins, where the system can be approximated by the deterministic dynamical system \eqref{4dimvirus}. Note that the system could as well already be started in this second phase, for example by the artificial exposure of the host cells to large numbers of virions, as e.g. in the experiments in \cite{B15}.  In the deterministic phase, we could further distinguish between an `early phase', in which the reproduciton of viruses does not yet play a role (as is the scenario considered in \cite{GW15}), and its `long-term' behaviour, in which a persistent epidemic may emerge. 

To describe the latter, we further define, for $\beta>0$ the {\em persistence set}
\[
\begin{aligned} S_\beta &:= \Big\{ (\widetilde n_{1a},\widetilde n_{1d},\widetilde n_{1i},\widetilde n_2) \in (0,\infty)^4 \colon n_{\upsilon} \geq \beta, \forall \upsilon \in \{ 1a,1d,1i,2\}, \widetilde n_{1a}+\widetilde n_{1d}+\widetilde n_{1i} \leq \barn1a-\beta, \\ & \qquad\widetilde n_2 \leq \frac{mv\barn1a}{\mu_2}-\beta \Big\}.\end{aligned}  \numberthis\label{Sbetadefvirus}
\]
Inside this persistence set, all sub-populations will have size at least $\beta >0$, and the total host population size will already be below its virus-free equilibrium. The last condition ensures that the virus load is not too high, thus this set describes a scenario with a currently persistent but not overwhelming virus infection. Note that $S_\beta$ is always well-defined and non-empty if $\beta \in (0,\barn1a\min \{ 1, \frac{mv}{\mu_2} \})$ (since we assumed that $\lambda_1>\mu_1$). Further,
\[ 
T_{S_\beta} := \inf \{ t\geq 0 \colon (N_{1a,t}^K,N_{1d,t}^K,N_{1i,t}^K,N_{2,t}^K) \in S_\beta \} \numberthis\label{TSbetadefvirus}
\]
is the hitting time for the persistence region $S_\beta$ (it is again a stopping time for the canonical filtration). 

Thanks to Corollary~\ref{cor-persistence}, the solution $(\mathbf n(t))_{t \geq 0}$ of the dynamical system~\eqref{4dimvirus} started from suitable initial conditions will eventually reach the set $S_\beta$ defined in~\eqref{Sbetadefvirus} if $\beta>0$ is sufficiently small. Our goal is to show that the same holds with high probability for our rescaled stochastic process $(\mathbf N_t^K)_{t \geq 0}$ conditional on a successful invasion. Theorem~\ref{thm-successvirus} below tells that this is indeed the case for $\widetilde\lambda>0$, moreover, it follows from the proof of this theorem that $T_{S_\beta}-T_\eps^2$ can be bounded by a finite positive constant depending on $\eps$ (with high probability conditional on the successful invasion). Thus, while the duration of the first, stochastic phase scales like constant times $\log K$, the second, deterministic phase takes $O(1)$ time, just as in the general stochastic invasion models of adaptive dynamics introduced in~\cite{C06}.

Moreover, $(\mathbf N_t^K)_{t \geq 0}$ does not only spend a short period of time in the set $S_\beta$, but for $T>0$ large enough and independent of $K$, $\mathbf N_{T_{S_\beta}+T}^K \in S_\beta$ holds with high probability as $K \to \infty$ (cf.~Corollary~\ref{cor-longterm} below). Note that the strong Markov property of $(\mathbf N_t^K)_{t \geq 0}$ allows us to consider uniform approximations after such random stopping times.

One can expect that starting from time $T_{S_\beta}$, the stability of the coexistence equilibrium $(n_{1a}^*,n_{1d}^*,n_{1i}^*,n_2^*)$ already matters for the dynamics of our stochastic process. In Figure~\ref{fig-phases} we illustrate these phases of a successful virus invasion in the particular scenario when the coexistence equilibrium $(n_{1a}^*,n_{1d}^*,n_{1i}^*,n_2^*)$ is globally asymptotically stable and thus between $T_{S_\beta}$ and $T_{S_\beta}+T$, the rescaled population size process $\mathbf N_{t}^K$ stays close to this equilibrium. 



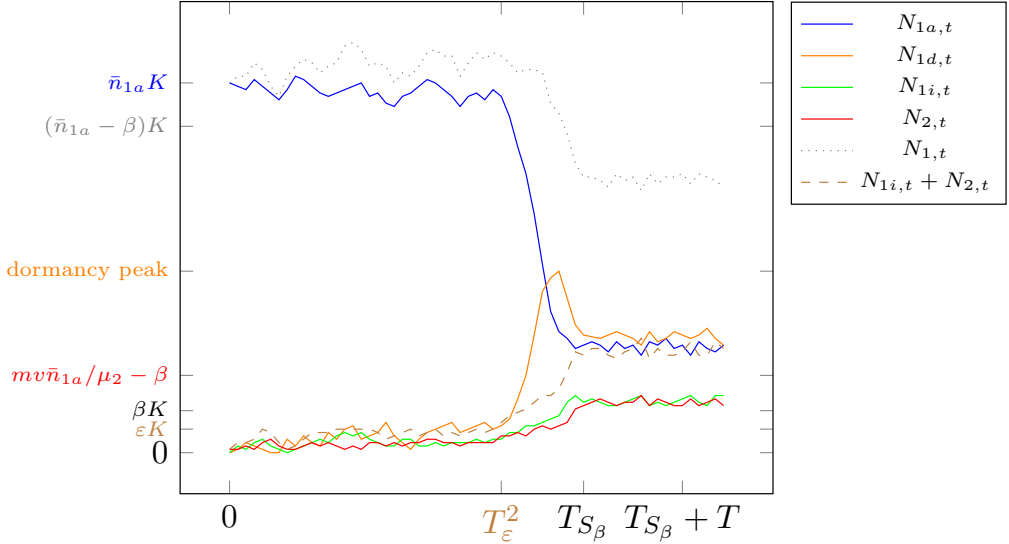
\begin{figure}
    \centering
\begin{tikzpicture}[scale=1.15]
\begin{axis}[
xtick={0,3.3,4.3,5.5},
xticklabels={0,\color{brown}$T_\eps^2$,$T_{S_\beta}$,$T_{S_\beta}+T$},
ytick={0,0.0701,0.125,0.23,0.54,0.971,1.1},
yticklabels={0,\color{brown}\tiny$\eps K$,\color{black}\tiny$\beta K$\color{black},\color{red}\tiny$m v\barn1a/\mu_2-\beta$,\color{orange}\tiny{dormancy peak}\color{black},\color{gray}\tiny$(\barn1a-\beta)K$,\color{blue}\tiny$\barn1a K$},
legend pos=outer north east,
]
\addplot [blue] table {resident.dat};
\addplot [orange] table {dormant.dat};
\addplot [green] table {infected.dat};
\addplot [red] table {virus.dat};
\addplot [gray,dotted] table {alltrait1.dat};
\addplot [brown,dashed] table {infected+virus.dat};
\legend{\tiny $N_{1a,t}$,  \tiny $N_{1d,t}$, \tiny $N_{1i,t}$, \tiny $N_{2,t}$, \tiny $N_{1,t}$, \tiny $N_{1i,t}+N_{2,t}$,};
\end{axis}
\end{tikzpicture}
\begin{comment}{
\begin{subfigure}{1.\textwidth}
\caption{The case when the coexistence equilibrium is unstable and a stable periodic trajectory attracts the coordinatewise positive solutions.}\label{subfigure-periodic}
\begin{tikzpicture}[scale=1.15]
\begin{axis}[
xtick={0,3.3,4.3,5.8},
xticklabels={0,\color{brown}$T_\eps^2$,$T_{S_\beta}$,$T_{S_\beta}+T$},
ytick={0,0.0701,0.125,0.25,0.54,0.971,1.1},
yticklabels={0,\color{brown}\tiny$\eps K$,
\color{black}\tiny$\beta K$\color{black},
\color{red}\tiny$m v\barn1a/\mu_2-\beta$,
\color{orange}\tiny{dormancy peak}\color{black},\color{gray}\tiny$(\barn1a-\beta)K$,\color{blue}\tiny$\barn1a K$},
legend pos=outer north east,
]
\addplot [blue] table {residentp.dat};
\addplot [orange] table {dormantp.dat};
\addplot [green] table {infectedp.dat};
\addplot [red] table {virusp.dat};
\addplot [gray,dotted] table {alltrait1p.dat};
\addplot [brown,dashed] table {infected+virusp.dat};
\legend{\tiny $N_{1a,t}$,  \tiny $N_{1d,t}$, \tiny $N_{1i,t}$, \tiny $N_{2,t}$, \tiny $N_{1,t}$, \tiny $N_{1i,t}+N_{2,t}$,};
\end{axis}
\end{tikzpicture}
\end{subfigure}}\end{comment}
 \caption{\small Schematic illustration of the behaviour of $(\mathbf N_t^K)_{t \geq 0}$ in case of a successful invasion.  We sketch the possible scenario of stable coexistence. 
    Each label 
    has the same colour as the graph of the corresponding subpopulation. Black labels correspond to multiple types. The orange `dormancy peak' corresponds to a similar peak described in \cite{GW15} for large $q$ (see Figure~\ref{fig-GW-peak} for a quantitative result), and gray dotted curves depict to the total host size $N_{1,t}=N_{1a,t}+N_{1d,t}+N_{1i,t}$.
    \\ \smallskip
     }
\label{fig-phases}
\end{figure}

%



\subsection{Statement of the main results}\label{sec-resultsvirus}
Now, we formulate our main results. 
Our first theorem states that the probability of succesful invasion of the virus particles (that is, when a macroscopic epidemic emerges and becomes persistent), i.e.\  of reaching the set $S_\beta$ for some $\beta>0$ before extinction of the invaders, converges to the survival probability $1-s_2$ of the approximating branching process as $K \to \infty$. For this, recall that the eigenvalue $\widetilde\lambda$ from~\eqref{lambdatildedefvirus} is positive if and only if~\eqref{viruscoexcond} holds. Also recall the extinction probability $s_2$ of the approximating branching process started from $(0,0,1)$ from~\eqref{qdefvirus}. 
\begin{theorem}\label{thm-viruscoexprob}
Assume that $\widetilde\lambda \neq 0$. Assume further that
\[ N^K_{1a}(0) \underset{K \to \infty}{\to} \barn1a  \]
almost surely and 
\[ (N^K_{1d}(0),N^K_{1i}(0),N^K_{2}(0))=(0,0,\smfrac{1}{K}). \]
Then for all sufficiently small $\beta>0$, we have
\[ \lim_{K \to \infty} \mathbb P \Big( T_{S_\beta} < T_0^2 \Big) = 1- s_2. \numberthis\label{successvirus} \]
\end{theorem}
The next theorem shows that in case of a macroscopic/persistent epidemic, the time until reaching the set $S_\beta$ (which includes the coexistence equilibrium $(n_{1a}^*,n_{1d}^*,n_{1i}^*,n_2^*)$ of the dynamical system) behaves like $\log K/\widetilde\lambda$.
\begin{theorem}\label{thm-successvirus}
Under the assumptions of Theorem~\ref{thm-viruscoexprob}, in case \eqref{viruscoexcond} holds (equivalently, $s_2<1$), for all sufficiently small $\beta>0$ we have that on the event $\{ T_{S_\beta} < T_0^2 \}$,
\[ \lim_{K \to \infty} \frac{T_{S_\beta}}{\log K} = \frac{1}{\widetilde\lambda} \numberthis\label{invasionvirus} \]
in probability.
\end{theorem}
The final theorem provides information about the time of the extinction of the epidemic and implies that with high probability, the rescaled active population size stays close to its virus-free equilibrium $\barn1a$ and the dormant population stays small until this extinction (after which it decreases to 0). This theorem also holds for $\widetilde\lambda>0$ where both persistence and non-persistence of the epidemic have a positive probability.
\begin{theorem}\label{thm-failurevirus}
Under the assumptions of Theorem~\ref{thm-viruscoexprob}, for all sufficiently small $\beta>0$ we have that on the event $\{ T_0^2 < T_{S_\beta}\}$, 
\[ \lim_{K \to \infty} \frac{T_0^2}{\log K} =0 \numberthis\label{extinctionvirus} \]
and
\[ \mathds 1_{ \{ T_{S_\beta} > T_0^2 \}} \big\Vert (N^K_{1a,T_0^2},N^K_{1d,T_0^2})-(\barn1a,0) \big\Vert \underset{K \to \infty}{\longrightarrow} 0, \numberthis\label{lastoftheoremvirus} \]
both in probability, where $\Vert \cdot \Vert$ is an arbitrary (but fixed) norm on $\R^2$.
\end{theorem}

Theorems~\ref{thm-viruscoexprob} and \ref{thm-successvirus} guarantee that if $s_2<1$, then in the limit $K\to\infty$, for appropriately chosen $\beta$, with probability tending to $1-s_2$, the persistence set $S_\beta$ will be reached before the extinction of the epidemic, and the time $T_{S_\beta}$ scales as $(\smfrac{1}{\widetilde\lambda}+o(1))\log K$. They do not provide information about the fate of our rescaled stochastic population process after time $T_{S_\beta}$, and in fact we expect that this will depend on the stability of the coexistence equilibrium and possible Hopf bifurcations. We will interpret this in Section~\ref{sec-bifurcations}.

Nevertheless, by virtue of Corollary~\ref{cor-persistence}, we can still provide an assertion about the long-term behaviour of our stochastic system. Namely, if $T_{S_\beta}<T_0^2$, then for $T>0$ sufficiently large, with high probability, at time $T_{S_\beta}+T$ the process will again be situated in $S_{\beta}$. This is true because \cite[Theorem 2.1, p.~456]{EK} guarantees that $(\mathbf N_{T_{S_\beta}+t})_{t \in [0,T]}$ is well-approximated by the solution $(\mathbf n_{t})_{t \in [0,T]}$ of \eqref{4dimvirus} given convergence of the initial conditions, and Corollary~\ref{cor-persistence} implies that started from anywhere in $S_\beta$, $\mathbf n_{T}$ will be situated in $S_\beta$ for all $T>0$ large enough.
More precisely, we have the following result, whose proof is now immediate.
\begin{cor}\label{cor-longterm}
Assume that \eqref{viruscoexcond} holds. Then for all sufficiently small $\beta>0$ and sufficiently large $T>0$, we have
\[ \lim_{K\to\infty} \P\big( \mathbf N_{T_{S_\beta}+T} \in S_{\beta} \big| T_{S_\beta}<T_0^2 \big) = 1. \numberthis\label{longpersistence} \]
\end{cor}
Note that \eqref{longpersistence} is equivalent to the fact that the assertions
\[ \lim_{K\to\infty} \P\big( N_{1a,T_{S_\beta}+T}^K+N_{1d,T_{S_\beta}+T}^K+N_{1i,T_{S_\beta}+T}^K \leq \barn1a-\beta \big| T_{S_\beta}<T_0^2 \big) =1, \]
\[ \lim_{K\to\infty} \P\Big( N_{2,T_{S_\beta}+T}^K \leq \frac{mv \barn1a}{\mu_2}-\beta  \Big| T_{S_\beta}<T_0^2 \Big) =1, \]
and
\[ \lim_{K\to\infty} \P\big( N_{\upsilon,T_{S_\beta}+T}^K \geq \beta \big| T_{S_\beta}<T_0^2 \big)=1, \qquad \forall \upsilon \in \{ 1a, 1d, 1i, 2 \} \]
hold. I.e., we have persistence of the epidemic on intervals starting at $T_{S_\beta}$ whose length does not scale with $K$.
The proofs of Theorems~\ref{thm-viruscoexprob}, \ref{thm-successvirus}, and \ref{thm-failurevirus} will be carried out in Section~\ref{sec-proofsvirus}.

\subsection{Discussion of model and results}\label{sec-discussionvirus} 

\subsubsection{Interpretation of the coexistence condition}\label{sec-interpretationcoex}
Condition~\eqref{viruscoexcond} is equivalent to
\[ \Big(\frac{v}{r+v}m-1\Big) (1-q) D\barn1a > \mu_2. \numberthis\label{viruscoexcond2} \]
To understand this intuitively, imagine that our population process is observed shortly after time 0 so that it consists of approximately $K\barn1a$ type 1a individuals and only a few virions. Then, the right-hand side is the death rate of each single virion, whereas the left-hand side is the rate at which it produces new virus particles (via infection and subsequent lysis). Indeed, $(1-q)D\barn1a$ is the rate at which the virion successfully invades a type 1a individual, the $-1$ corresponds to the loss of this invader during the attack, the fraction $\frac{v}{r+v}$ is the probability that the infected individual does not recover, and $m$ is the number of new virions released from the infected cell. This is precisely the approximation of the population size that corresponds to the definition of the branching process $(\widehat{\mathbf N}(t))_{t \geq 0}$ (or its projection to types 1i and 2). Thus, we see that \eqref{viruscoexcond2} ensures that the branching process is supercritical, or equivalently, that \eqref{4dimvirus} has a coordinatewise positive equilibrium. 

Proposition~\ref{lemma-coexistencevirus} states that if $(n_{1a}^*,n_{1d}^*,n_{1i}^*,n_2^*)$ exists, then $n_{1a}^*<\barn1a$, i.e.\ a persistent virus epidemic always reduces the population size of type 1a (compared to its virus-free equilibrium); 
else, the epidemic is not persistent. The fact that $\barn1a > n_{1a}^*$ whenever coexistence is possible also indicates that coexistence with type 2 is always detrimental for type 1 (not surprisingly).

Let us point out that Condition~\ref{viruscoexcond2} neither depends on the death rate factor $\kappa$ nor the resuscitation rate $\sigma$ of dormant individuals. That is, for fixed $q \in (0,1)$, we get the same probability for the persistence of the epidemic in the limit $\kappa=\infty,\sigma\geq 0$ (where an unsuccessful virus attack kills the affected type 1a individual without giving a chance to reproduction of viruses) as well as in the opposite limit $\kappa \geq 0,\sigma=\infty$ (where an unsuccessful virus attack keeps the affected individual alive and active). This is the consequence of the fact that the dynamics of type 1d does not directly affect the one of types 1i and 2 (cf.\ $(1,0,0)$ is a left eigenvector of $J$), but only via type 1a, which has a nearly constant rescaled population size during the first (very early) phase of the invasion. Note that $\kappa$ and $\sigma$ do not influence the time until a successful invasion either since the eigenvalue $\widetilde\lambda$ (cf.~\eqref{lambdatildedefvirus}) does not depend on these parameters. 

After succesful invasion, in the second phase of the epidemic, the values of $\kappa$ and $\sigma$ play a more prominent role, governing important properties of the coexistence equilibrium. For example, it follows from the proof of Proposition~\ref{lemma-coexistencevirus} (see Section~\ref{sec-preliminaryproofs} below) that under Condition~\eqref{viruscoexcond2}, the dormant coordinate $n_{1d}^*$ of the coexistence equilibrium $(n_{1a}^*,n_{1d}^*,n_{1i}^*,n_2^*)$ satisfies
\[ n_{1d}^* = \frac{q D n_{1a}^*n_2^*}{\kappa\mu_1+\sigma}. \numberthis\label{dormantcoord2} \]
Thus, $n_{1d}^*$ is an increasing function of $q$ and a decreasing function of $\kappa\mu_1+\sigma$. See Example~\ref{ex-largesigma} for a related discussion and simulation about the dynamical system~\eqref{4dimvirus} in the case of diverging $\sigma$.

It is also clear that (given $\lambda_1>\mu_1$) \eqref{viruscoexcond2} can only hold if
$mv>r+v$, else the left-hand side is nonpositive, whereas the right-hand side is positive by assumption. This condition says that each virus attack increases the number of virions on average, i.e.\ the mean number of virions that is created when the infected individual leaves state 1i (which equals $m$ in case the individual dies and $0$ if it recovers) exceeds 1 (which is the number of viruses lost at each virus attack). If $\frac{mv}{r+v} \leq 1$, then $(\widehat{\mathbf N}(t))_{t \geq 0}$ will be strictly subcritical, regardless of the values of $\barn1a$, $D$, $q$, and $\mu_2$. 

\subsubsection{Dormancy-related reproductive trade-offs in the light of the threat of persistent epidemics}\label{sec-lambdaversusq}

Note that condition~\eqref{viruscoexcond2} implies that large reproduction rates $\lambda_1$ can be hazardous when facing a virus infection (with or without dormancy mechanism). 
At first glance, there may be hypothetical scenarios where a population threatened by recurring virus invasions might not realize its full reproductive potential in order to avoid persistent epidemics. The way to maximize its long-term average fitness in the face of virus epidemics could then be to invest remaining resources into a dormancy-defense, which allows for higher carrying capacities during infections, and the `reproductive trade-off' vanishes (at least to some degree). However, such a self-constraining strategy might be vulnerable to the invasion of selfish cheaters, i.e.\ of other species investing in a higher reproduction rate instead of dormancy.
Investigating the balance of classical fitness (in competition with other species) and strategies (e.g.\ dormancy-based) reducing reproductive rates in order to cope with recurring infections could be a topic for future work.

%

\subsubsection{Host--virus dynamics started with a single infected individual}\label{sec-1dormantinfected}
 By~\cite[Section 7.2]{AN72}, the extinction probability $s_2$ defined in~\eqref{qdefvirus} equals 1 if $\widetilde \lambda<0$, whereas if $\widetilde\lambda>0$, then $s_2$ equals the last coordinate of the coordinatewise smallest nonnegative solution $(s_{1d},s_{1i},s_2)$ of the system of generating equations
\[
\begin{aligned}
(\kappa\mu_1+\sigma)(1-s_{1d})&=0, \\
r(1-s_{1i})+v(s_2^m-s_{1i})&=0, \\
qD\barn1as_{2}(s_{1d}-1) + (1-q)D\barn1a(s_{1i}-s_2) +\mu_2(1-s_2) & =0. 
\end{aligned} \numberthis\label{s1ds1is2}
\]
Here, $s_{1d}$ and $s_{1i}$ are the extinction probabilities of the branching process $(\widehat{\mathbf N}(t))_{t \geq 0}$ started from $(1,0,0)$ and $(0,1,0)$, respectively. In particular, it follows from the first equation in~\eqref{s1ds1is2} that $s_{1d}=1$. I.e.\ started with one dormant individual (or any positive number of dormant individuals) and no infected individuals or virions, the branching process dies out almost surely, as anticipated before. Indeed, dormant individuals are created only in presence of virions, while dormant individuals do not produce any virions, not even indirectly via producing infected individuals. 
According to the second equation in~\eqref{s1ds1is2}, $s_{1i}$, which is the extinction probability of the process started from $(0,1,0)$, satisfies
\[ s_{1i} = \frac{r + v s_2^m}{r+v}. \numberthis\label{1i2extinctionprob}\]
This implies that $s_{1i}$ equals one if and only if $s_2$ equals one.

\subsubsection{The reproduction number of the virus epidemic, and relation to stochastic epidemic models}\label{sec-R0}
The distinction between an initial stochastic phase, where an invader can be described by a branching process, followed by deterministic behaviour, where the whole system is well-described by a dynamical system, is of course reminiscent of stochastic and deterministic epidemic modelling. 
In stochastic epidemic models like the standard SIR (susceptible--infected--removed) model, the \emph{basic reproduction number} $R_0$ of the epidemic is defined as the expected number of infections generated by one infectious individual in a large susceptible population, cf.~\cite[Section 2.1]{AB00}. Despite not treating pathogens as individuals and assuming that the population size is constant (or decreases only due to deaths caused by the infectious disease), the quantity $R_0$ can already be introduced in the basic SIR model the same way as in our model.

Note that we can still define $R_0$ in our model in such a way that it still fulfills the heuristic definition of \cite{AB00} (where we always assume that $\lambda_1>\mu_1$). In order to obtain `a large susceptible population', we will have to assume that $K$ is large, since the equilibrium population size scales like $K(\barn1a+o(1))$ as $K\to\infty$. Then, similarly to the branching process approximation of types 1d, 1i, and 2 during the initial phase of the epidemic, we will assume that the rescaled susceptible population size is fixed as $\barn1a$ (ignoring also the question of whether this number is an integer). Let us now look at an infected individual in this situation. It either recovers with probability $r/(r+v)$ or dies due to lysis, giving rise to $m$ new virions, with probability $v/(r+v)$. Each of these new virions will eventually either degrade, which happens at rate $\mu_2$, or successfully attack a susceptible individual. Since there are $K\barn1a$ susceptibles, the probability that the latter event occurs is $\frac{(1-q)D\barn1a}{(1-q)D \barn1a+\mu_2}$. The number of infected individuals emerging from attacks by these $m$ viruses is the average number of infections generated by the originally infected individual. Thus, we obtain  the expression
\[ 
R_0 = \frac{mv(1-q)D\barn1a}{(r+v) ((1-q)D\barn1a + \mu_2)} 
\]
for the
reproduction number in our model.
Note that $R_0$ depends on $q$ but not on $\kappa$ and $\mu$, and in particular it is the same as for a dormancy-free epidemic with lower infectivity if we replace $q$ by $0$ and $D$ by $(1-q)D$. This gives a rather natural interpetation of the effect of dormancy from an epidemiological point of view.

Indeed, $R_0>1$ holds if and only if
\[ (mv-(r+v))(1-q)D\barn1a > (r+v)\mu_2, \]
which is precisely our coexistence condition \eqref{viruscoexcond}. 

Note further that $R_0$ can also be interpreted as the average number viruses who are the `offspring' of a single given virus, obtained via infection of a susceptible individual producing secondary viruses via lysis. We see that $R_0>1$ is equivalent to condition~\eqref{viruscoexcond2}, which we interpret as the average number of `offspring' of a given virus being at least 1. This provides a heuristic reason why $s_{1i} \neq 1$ is equivalent to $s_2 \neq 1$, which we have verified in Section~\ref{sec-1dormantinfected}.

\subsubsection{Modelling choices and extensions}

As mentioned before, we opted for a lytic virus release mechanism after reproduction, always killing the host cell (as opposed to modeling chronic infection of individual cells, cf~e.g. \cite{GW18} for a related set-up).

Other modelling choices may be taken regarding the fate of the dormancy initiating virion. In our case,  we opted for a reversible host--virus contact, where the free virion is retained after dormancy initiation of the target cell. While this seems to model contact-mediated dormancy as discussed in \cite{B15}, a CRISPR-Cas based response as in \cite{MNM19} would require the virion to enter the cell, and thus the virion should be erased after dormancy initiation. 

Further, while it seems natural from a biological point of view to assume that active host individuals feel competitive pressure also from their infected siblings, they may not necessarily feel competitively challenged by their dormant siblings, given that these are metabolically inactive. Yet we opted to follow \cite{GW15} in including competition with dormant host individuals in order to stay close to their modeling frame.

A related question is why infected and dormant individuals impose competitive pressure onto the active
population but not the other way around. We also mainly chose this modelling decision in order to be consistent with~\cite{GW15}. Regarding the infected individuals, this modelling choice originates from~\cite[Section 1]{BK98}, where the authors argue that this is a reasonable assumption because the mortality of infected individuals is almost completely due to lysis. (In contrast, if one considers chronically infected cells, their lifespan is typically much longer than the one of lytically infected ones, and hence the competitive pressure that they feel is not negligible, cf.\ \cite{GW18}.) 

Finally, the burst size $m$, which we assumed to be constant for simplicity, could be replaced in each virus reproduction event by an independent integer-valued random variable $M$ with $\E[M]=m$ in order to make the model more realistic. Then, the underlying dynamical system would still be~\eqref{4dimvirus}, and the supercriticality of the corresponding branching process would still be equivalent to the coexistence condition~\eqref{viruscoexcond}, but e.g.\ the probability of a successful invasion would change. 

We refrain from discussing all possible mechanisms and their consequences here and leave them for future research. 

\section{Deterministic phase: Further analysis of the dynamical system}\label{sec-dynsyst2}
Unfortunately, the stability of the coexistence equilibrium and the associated question whether bifurcations emerge in the system \eqref{4dimvirus} are in general difficult (and tedious) to analyse and beyond the scope of the present paper. In Section~\ref{sec-bifurcations}, we present some partial results and conjectures in this regard. These are supported by numerical results for various parameter regimes, see Section~\ref{sec-simulationsvirus}. Finally, in Section~\ref{sec-paradoxofenrichment} we recall the notion of paradox of enrichment for predator--prey systems and explain its relation to our model. The biological relevance of an asymptotically periodic behaviour was already anticipated in Remark~\ref{remark-bifurcation}, and we will elaborate on it in Section~\ref{sec-paradoxofenrichment}. 

\subsection{Stability of the coexistence equilibrium, Hopf bifurcations}\label{sec-bifurcations}
The three-dimensional, dormancy-free analogue
\[
\begin{aligned}
\frac{\d n_{1a}(t)}{\d t} & = n_{1a}(t)\big( \lambda_1-\mu_1-C (n_{1a}(t)+n_{1i}(t))-D n_{2}(t) \big) + r n_{1i}(t),\\
\frac{\d n_{1i}(t)}{\d t} & = D n_{1a}(t) n_{2}(t) -(r+v) n_{1i}(t), \\
\frac{\d n_{2}(t)}{\d t} & = mv n_{1i}(t) -  D n_{1a}(t) n_{2}(t) -  \mu_2 n_{2}(t)
\end{aligned}
\numberthis\label{3dimvirus}
\]
of our dynamical system~\eqref{4dimvirus} is the main object of study of the paper~\cite{BK98} in the case $r=0$ when recovery is also absent from the system. Ignoring the dormant coordinate $n_{1d}(\cdot)$, it is straightforward to extend Propositions~\ref{lemma-coexistencevirus}, \ref{lemma-stabilityvirus}, and~\ref{prop-Lyapunov}, as well as Corollary~\ref{cor-persistence} to the system~\eqref{3dimvirus}; in fact for $r=0$ they are all stated and proved in the aforementioned paper. In particular, we denote the coordinatewise positive coexistence equilibrium of~\eqref{3dimvirus} by $(n_{1a}^*,n_{1i}^*,n_2^*)$, whenever it exists, analogously to the case of~\eqref{4dimvirus}. Let us recall the transcritical bifurcation point $m^*$ of the system~\eqref{4dimvirus} from~\eqref{mstar}, which is defined for~\eqref{3dimvirus} analogously via putting $q=0$.
For $r=q=0$, there exist $m^{**},m'>0$ with $m^{**}>m'>m^*$ such that the following assertions hold, along with the corresponding analogues of Proposition~\ref{prop-Lyapunov} and Corollary~\ref{cor-persistence}:
\begin{enumerate}[(A)]
    \item\label{first-BK98} The equilibrium $(\barn1a,0,0)$ is globally asymptotically stable for $m \in (0,m^*)$ with all eigenvalues of the corresponding Jacobi matrix being real, and thus solutions of~\eqref{3dimvirus} started from coordinatewise positive initial conditions converge to $(\barn1a,0,0)$ in a coordinatewise eventually monotone way,
    \item\label{second-BK98} all eigenvalues of the Jacobi matrix at $(n_{1a}^*,n_{1i}^*,n_2^*)$ are real (and strictly negative) for $m \in (m^*,m')$, and thus solutions started from coordinatewise positive initial conditions tend to $(n_{1a}^*,n_{1i}^*,n_2^*)$ in a coordinatewise eventually monotone way,
    \item\label{third-BK98} for $m \in (m',m^{**})$, one eigenvalue is still negative, while there is a pair of complex eigenvalues with negative real part, giving rise to an oscillatory convergence of solutions started from coordinatewise positive initial conditions to $(n_{1a}^*,n_{1i}^*,n_2^*)$,
    \item\label{last-BK98} and for $m \in (m^{**},\infty)$, one eigenvalue is still negative, while the pair of complex eigenvalues now has positive real part, thus $(n_{1a}^*,n_{1i}^*,n_2^*)$ is unstable, and solutions started from coordinatewise positive initial conditions apart from $(n_{1a}^*,n_{1i}^*,n_2^*)$ tend to a periodic limiting trajectory.
\end{enumerate}
To be more precise, at $m^{**}$ there is a Hopf bifurcation, so that in an open neigbourhood of $m^{**}$,
the Jacobi matrix of~\eqref{3dimvirus} at $(n_{1a}^*,n_{1i}^*,n_2^*)$ has a pair of complex eigenvalues, whose real part changes sign from negative to positive in $m^{**}$ with a nonvanishing derivative, and the Lyapunov coefficient is nonzero at $m^{**}$.

Note that in~\cite{BK98}, only the fact that $(n_{1a}^*,n_{1i}^*,n_2^*)$ is unstable was verified for all $m>m^{**}$, the existence of stable periodic orbits only for $m>m^{**}$ sufficiently close to $m^{**}$. Further, in the setting of the assertions~\eqref{second-BK98}, \eqref{third-BK98}, the authors of the paper only proved that $(n_{1a}^*,n_{1i}^*,n_2^*)$ is locally asymptotically stable, the general convergence of solutions mentioned in \eqref{second-BK98}, \eqref{third-BK98} was not shown. However, it is clear that the authors expect all the assertions \eqref{first-BK98}--\eqref{last-BK98} to hold (and in particular they also show that the Hopf bifurcation point is unique), and we share their opinion. Given a rigorous proof for all these assertions, it follows by continuity that they also hold for $r>0$ sufficiently small.


Now, we present some partial results that generalize the above assertions to the case when dormancy or recovery is present in the system. Given \cite[Proposition 3.2]{BK98}, the proof of the following assertion is immediate by continuity.
\begin{prop}\label{prop-thereisbifurcation}
Fix all parameters of the model but $m$. If $q>0$ and $r \geq 0$ are sufficiently small (in particular also if $r=0$ if $q$ is small enough), there exists $m^{**}>m^*$ such that the coexistence equilibrium $(n_{1a}^*,n_{1i}^*,n_2^*)$ of the dynamical system~\eqref{4dimvirus} is asymptotically stable for all $m \in (m^*,m^{**})$ and unstable for all $m \in (m^{**},\infty)$.
\end{prop}
Next, let us point out that recovery indeed has a qualitative effect on the behaviour of the dynamical system. Namely, if $r$ is sufficiently large compared to $v$, then for all sufficiently large $m$ the coexistence equilibrium is asymptotically stable, at least for $q$ small. For $q=0$, it is actually satisfied if recovery is more frequent than death by lysis (i.e.\ if the virus infection has mortality less than 50\%).
\begin{prop}\label{prop-thereisnobifurcation}
Fix all parameters of the model but $m$.
\begin{enumerate}[(i)]
    \item\label{first-thereisnobifurcation} There exists $r^*>0$ such that for all $r>r^*$, the coexistence equilibrium $(n_{1a}^*,n_{1i}^*,n_2^*)$ of \eqref{3dimvirus} is asymptotically stable for all sufficiently large $m>m^*$. In particular, we have $r^* \leq v$.
    \item\label{second-thereisnobifurcation} For $q \in (0,1)$ sufficiently small, there exists $r^*>0$ such that for all $r>r^*$, the coexistence equilibrium $(n_{1a}^*,n_{1d}^*,n_{1i}^*,n_2^*)$ of \eqref{4dimvirus} is asymptotically stable for all sufficiently large $m>m^*$. 
\end{enumerate}
\end{prop}
While \eqref{first-thereisnobifurcation} follows from ~\eqref{second-thereisnobifurcation} by continuity, \eqref{second-thereisnobifurcation} requires a proof because the case of $r>0$ was not treated in~\cite{BK98}. We carry out the proof of Proposition~\ref{prop-thereisnobifurcation} in Appendix~\ref{sec-appendixproofs}, as well as the one of the next assertion.
\begin{prop}\label{prop-firststable}
Fix all parameters of the model but $m$. Then, for all $m>m^*$ sufficiently small, the coexistence equilibrium $(n_{1a}^*,n_{1d}^*,n_{1i}^*,n_2^*)$ of \eqref{4dimvirus} is asymptotically stable. 
\end{prop}
We conclude the current section with a number of related conjectures which are based on numerical evidence (see Section~\ref{sec-simulationsvirus}). 
\begin{enumerate}
    \item\label{first-conj} Fix all parameters but $m$. For $r >0$ small enough, the assertions~\eqref{first-BK98}--\eqref{last-BK98} regarding the three-dimensional system~\eqref{3dimvirus} still hold. For $r$ sufficiently large, they hold with $m^{**}=\infty$.
    \item\label{third-conj} Fixing all parameters but $m$ and $r$ (resp.\ $m$ and $q$), $m^{*},m',m^{**}$ are monotone increasing in $q$ (resp.\ $r$). 
    \item\label{second-conj} For the dynamical system~\eqref{4dimvirus}, for $r \geq 0$ and $q \in (0,1)$ sufficiently small, the corresponding analogues of \eqref{first-BK98}--\eqref{last-BK98} hold (with $m^{**} \in (m^*,\infty)$). (That is, with $(\barn1a,0,0)$ replaced by $(\barn1a,0,0,0)$, $(n_{1a}^*,n_{1i}^*,n_2^*)$ with $(n_{1a}^*,n_{1d}^*,n_{1i}^*,n_2^*)$, and with the additional fourth eigenvalue of the Jacobi matrices mentioned in \eqref{first-BK98}--\eqref{last-BK98} always being real and negative.)
    In contrast, for any $q \in (0,1)$ we can choose $r>0$ sufficiently large such that the corresponding analogues of \eqref{first-BK98}--\eqref{last-BK98} hold with $m^{**} =\infty$, and for any $r \geq 0$ we can choose $q \in (0,1)$ sufficiently close to 1 such that the same is true. In particular, $ m'$ is always included in $(m^*,m^{**})$ .
\end{enumerate}

If one could verify the monotonicity in $q$ claimed in assertion~\eqref{third-conj}, Proposition~\ref{prop-thereisnobifurcation} would immediately generalize to all $q \in (0,1)$. However, what one could obtain this way would still be substantially weaker than the corresponding part of assertion~\eqref{second-conj}, which tells that if $r$ is large, then $(n_{1a}^*,n_{1i}^*,n_{1d}^*,n_2^*)$ is not only stable for $m$ very large, but for any $m>m^*$, and it is even globally attracting.

By assertion~\eqref{second-conj}, alone increasing $q$ can eliminate the Hopf bifurcation, i.e.\ make $m^{**}$ explode at some $q<1$. Certainly, this claim has no analogue for the three-dimensional system~\eqref{3dimvirus}, and we expect it to be difficult to prove. Our simulations presented in Example~\ref{ex-varyingq} nevertheless indicate that this assertion must be true, even in the case $r=0$ when there is no recovery but only dormancy.

Given that there exists $r>0$ large enough such that for any $q \in [0,1)$, $m^{**}=\infty$ holds, an interesting related open question is whether one can choose $q \in (0,1)$ sufficiently large such that for any $r \geq 0$, $m^{**}=\infty$ holds. Note that the assertions \eqref{first-conj}--\eqref{second-conj} do not necessarily imply that the answer to this question is positive.

\subsection{Numerical results}\label{sec-simulationsvirus}

To gain an understanding of the concrete behaviour of the dynamical system
~\eqref{4dimvirus} and its three-dimensional variant \eqref{3dimvirus}, in particular on the effects of dormancy and recovery, we now provide exact plots and numerical simulations of the critical burst sizes $m^*,m', m^{**}$ for various choices of the parameters and simulations of the solutions in some concrete cases. These results support and illustrate our conjectures presented in Section~\ref{sec-bifurcations}. We will work with the choice of parameters presented in Table~\ref{table-parameters} (abbreviating $n_I(0):=n_{1d}(0)+n_{1i}(0)+n_{2}(0)$), apart from those parameters that we vary in the given simulation. 

\begin{table}
\centering
\begin{tabular}{|l|l|l|l|l|} 
\hline
$\lambda_1$ & 5 &$v$ & 1.1 \\
$\mu_1$ & 4  & $\mu_2$ & 0.3 \\
$C$ & 1  & $n_{1a}(0)$ & $1(=\barn1a)$ \\
$\kappa$ & 1  & $n_I(0)$ & 0.1 \\
$q$ & 0.1  & $n_{1d}(0)$ & $\pi_{1d} \cdot n_I(0)$ \\
$r$ & 0.1   & $n_{1i}(0)$  & $\pi_{1i} \cdot n_I(0)$ \\
$D$ & 0.5   & $\sigma$ & $2$ \\
\hline
\end{tabular}
\caption{Default choice of the parameters for the simulations of the dynamical systems~\eqref{4dimvirus} and \eqref{3dimvirus}, where $n_I(0)=n_{1d}(0)+n_{1i}(0)+n_{2}(0)$.}\label{table-parameters}
\end{table}

Here, $\pi_{1d}$, $\pi_{1i}$, and $\pi_2$ are the dormant, infected, and virus coordinates of the coordinatewise positive (`Kesten--Stigum') left eigenvector of the mean matrix $J$ associated to the eigenvalue $\widetilde\lambda$ normalized so that $\pi_{1d}+\pi_{1i}+\pi_2=1$. 
 Heuristically, the reason why this initial condition is natural is that for $\widetilde\lambda>0$, conditional on survival of the approximating branching process $(\widehat{\mathbf N}(t))_{t \geq 0}$, the proportions of its dormant, infected, and virus coordinates converge to the corresponding proportions of $(\pi_{1d},\pi_{1i},\pi_2)$ thanks to the Kesten--Stigum theorem (cf.\ e.g.\ \cite{GB03}). 
 
 \begin{example}[Varying $r$ for fixed $q$]
 With the default choice of parameters apart from $r$, in Figure~\ref{fig-basic2} we plot the transcritical bifurcation point $m^{*}$, the point $m'$ where a pair of eigenvalues of the Jacobi matrix of~\eqref{4dimvirus} at $(n_{1a}^*,n_{1d}^*,n_{1i}^*,n_2^*)$ becomes complex, and the Hopf bifurcation point $m^{**}$, as functions of $r$. Note that the recovery-free case $r=0$ is also included in the images (for this particular choice of $q$). Given that we have fixed all other parameters, $m^*$ is a linear function of $r$ (it is computed from \eqref{mstar}). Assuming that assertions~\eqref{first-BK98}--\eqref{last-BK98} hold true, we know that there is precisely one value of $m>m^*$, namely $m=m'$, where a pair of eigenvalues of the Jacobi matrix becomes complex, and for $m^{**}<\infty$ there is a unique value of $m>m'$, namely $m=m^{**}$, where these eigenvalues are purely imaginary, while for $m^{**}=\infty$ the real part of these eigenvalues remains negative for all $m>m'$. We evaluate $m'$ and $m^{**}$ in a discrete set of points, and we conclude that the dependency of $m'$ on $r$ also seems linear. As expected, there exists $r_0>0$ such that for $r>r_0$, the Hopf bifurcation point $m^{**}$ explodes and becomes infinite, i.e., $(n_{1a}^*,n_{1d}^*,n_{1i}^*,n_2^*)$ stays locally asymptotically stable for all $m>m^*$ (despite the fact that $q$ is relatively small, while it is not necessarily small enough to apply Proposition~\ref{prop-thereisnobifurcation}). In this case we have $r_0 \approx 0.69$, given that for $r$ above this value, as one increases $m$, the real parts of the two complex eigenvalues seem to converge to a strictly negative value.  We know from Proposition~\ref{prop-thereisnobifurcation} that such $r_0$ also exists for $q=0$. It is not included in the images, but its value is about $0.73$.
 \begin{figure}
    \centering
    \includegraphics[scale=1]{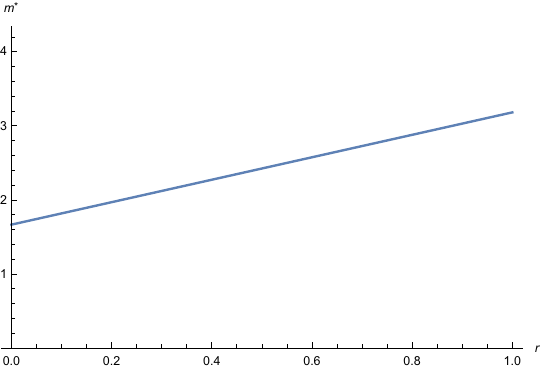}
    \includegraphics[scale=1]{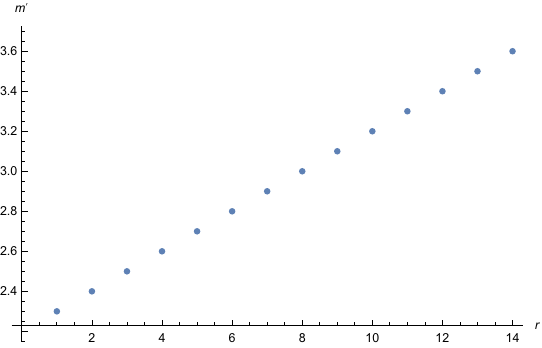} \\
    \includegraphics[scale=1]{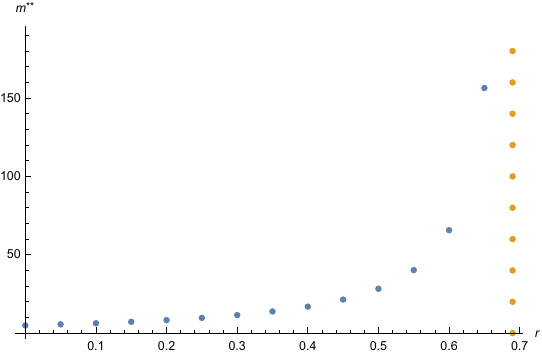}
     \caption{Values of the critical burst sizes $m^*,m^{**},m'$ as functions of $r$ with all other parameters (in particular $q$) fixed. At $r \approx 0.69$ (orange dotted line) $m^{**}$ explodes and becomes infinite.}
    \label{fig-basic2}
\end{figure}
\end{example}
 \begin{example}[Varying $q$ for fixed $r$]\label{ex-varyingq}
 With the default choice of parameters apart from $q$, in Figure~\ref{fig-basic3} we plot $m^*,m',m^{**}$ as functions of $q$. Note that the dormancy-free case $q=0$ is also included in the images (for this particular choice of $r$). The equation~\eqref{mstar} again gives an explicit formula for $m^*$ as a function of $q$, which is finite for all $q \in [0,1)$, monotone increasing in $q$, and tends to $\infty$ as $q \uparrow 1$. Also $m'$ seems to only explode in the limit $q \uparrow 1$. In contrast, for $q= 0.93$, the Hopf bifurcation point $m^{**}$ already seems to be infinite. 
 
 This provides numerical evidence that choosing the dormancy initiation probability $q \in (0,1)$ sufficiently large eliminates the Hopf bifurcation, although recovery is relatively weak so that for $q=0$ the Hopf bifurcation is present. It is not included in the images, but we also checked the recovery-free case $r=0$ with otherwise unchanged parameters, and for $q = 0.97$ we also found $m^{**}=\infty$ there. In other words, dormancy can help avoid Hopf bifurcations even in the absence of recovery, albeit this may require $q$ to be unrealistically high. 
 \begin{figure}
    \centering
    \includegraphics[scale=1]{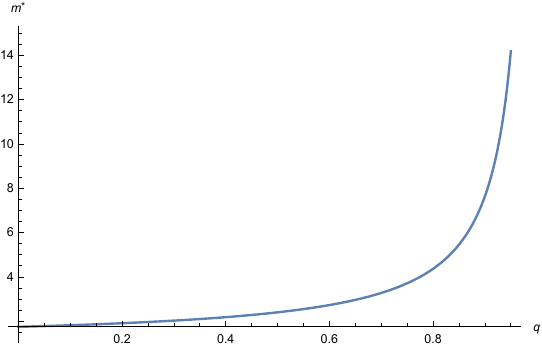}
    \includegraphics[scale=1]{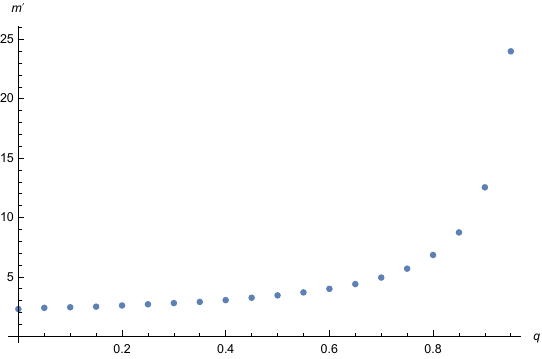} \\
    \includegraphics[scale=1]{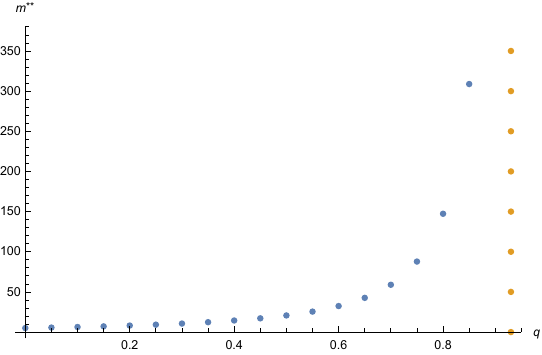}
    \caption{Values of the critical burst sizes $m^*,m',m^{**}$ as functions of $q$ with all other parameters (in particular $r$) fixed. At $q \approx 0.93$ (orange dotted line), $m^{**}$ explodes and becomes infinite.}\label{fig-basic3}
\end{figure}

We complement our precise numerical results from Figure~\ref{fig-basic3} with a schematic illustration of the critical burst sizes $m^*,m'$, and $m^{**}$ as functions of $q \in [0,1)$ with all parameters but $m$ and $q$ fixed, for $r$ small resp.\ large compared to $v$ (in the left resp.\ right picture of Figure~\ref{fig-smallrfantasy}). Note that this illustration is based on our conjectures listed in Section~\ref{sec-bifurcations}, while some of the properties of the curves $q \mapsto m^*(q), m'(q), m^{**}(q)$ and the regions between them can be justified based on the results of our paper (cf.~Section~\ref{ssn-dynsyst}). The shape of the curves is not meant to be precise, but qualitatively correct, in particular we expect them to be convex (and thus lower semicontinuous) as $[0,\infty]$-valued functions. 

\begin{figure}
\begin{tikzpicture}
\draw[scale=2.2] (-0.1, 0) -- (0, 0);
  \draw[scale=2.2, ->, name path=x] (0, 0) -- (3.1, 0) node[right] {$q$};
    \draw[scale=2.2, domain=0:2.9203, smooth, variable=\x, yellow, name path=transcritical] plot ({\x}, {0.25/(3-\x)});
     \draw[scale=2.2, domain=0:2.8404, smooth, variable=\x, orange, name path=secondorder] plot ({\x}, {0.5/(3-\x)});
     \tikzfillbetween[of=transcritical and secondorder, on layer=ft]{yellow, opacity=0.2};
          \draw[scale=2.2, domain=0:1.8405, smooth, variable=\x, red, name path=Hopf] plot ({\x}, {0.5/(2-\x)});
       \tikzfillbetween[of=Hopf and secondorder, on layer=ft]{orange, opacity=0.2};
         \draw[scale=2.2, domain=1.83:3, dashed, variable=\x, red] plot ({\x}, {3.2});
           \draw[scale=2.2] (0.85,2) node[below] {\Large IV.};
                  \draw[scale=2.2] (2.6,0.27) node[below] {\Large I.};
                \draw[scale=2.2] (1.8,0.9) node[below] {\Large III.};
\draw[scale=2.2] (2.586,0.65) node[left] {\Large II.};
    \draw[scale=2.2] (0, 0) -- (0, -0.1) node[above] {};
  \draw[scale=2.2, name path=y] (0, 0) -- (0, 3.13) node[above] {};
    \draw[scale=2.2, ->] (0, 3.13) -- (0, 3.3) node[above] {};
    \draw[scale=2.2] (0,3.2) node[left] {$+\infty$};
    \draw[scale=2.2] (-0.03,3.2) -- (0.03,3.2);
  \draw[scale=2.2] (3, 0) -- (3, -0.1) node[above] {};
   \draw[scale=2.2, name path=y2] (3, 0) -- (3, 3.13) node[above] {};
        \tikzfillbetween[of=y and y2, on layer=bg]{teal, opacity=0.2};
              \tikzfillbetween[of=Hopf and y, on layer=ft]{red, opacity=0.2};
                  \draw[scale=2.2] (0.1,3.4) node[left] {$m$};
                            \draw[scale=2.2] (1.65,0.5) node[left] {$m^{**}(q)$};
\draw[scale=2.2] (2.288,0.395) node[left] {$m'(q)$};
\draw[scale=2.2] (2.94,0.4) node[left] {$m^*(q)$};
 \draw[scale=2.2] (0.2,-0.1) node[left] {$0$};
      \draw[scale=2.2] (3,-0.1) node[left] {$1$};
\end{tikzpicture}
\begin{tikzpicture}
\draw[scale=2.2] (-0.1, 0) -- (0, 0);
  \draw[scale=2.2, ->, name path=x] (0, 0) -- (3.1, 0) node[right] {$q$};
    \draw[scale=2.2, domain=0:2.904, smooth, variable=\x, yellow, name path=transcritical] plot ({\x}, {0.3/(3-\x)});
     \draw[scale=2.2, domain=0:2.8078, smooth, variable=\x, orange, name path=secondorder] plot ({\x}, {0.6/(3-\x)});
     \tikzfillbetween[of=transcritical and secondorder, on layer=ft]{yellow, opacity=0.2};
         \draw[scale=2.2, domain=0:3, dashed, variable=\x, red] plot ({\x}, {3.2});
                  \draw[scale=2.2] (2.6,0.3) node[below] {\Large I.};
                \draw[scale=2.2] (0.8,1.9) node[below] {\Large III.};
\draw[scale=2.2] (2.45,0.6) node[left] {\Large II.};
    \draw[scale=2.2] (0, 0) -- (0, -0.1) node[above] {};
  \draw[scale=2.2, name path=y] (0, 0) -- (0, 3.13) node[above] {};
    \draw[scale=2.2, ->] (0, 3.13) -- (0, 3.3) node[above] {};
    \draw[scale=2.2] (0,3.2) node[left] {$+\infty$};
    \draw[scale=2.2] (-0.03,3.2) -- (0.03,3.2);
  \draw[scale=2.2] (3, 0) -- (3, -0.1) node[above] {};
   \draw[scale=2.2, name path=y2] (3, 0) -- (3, 3.13) node[above] {};
        \tikzfillbetween[of=y and y2, on layer=bg]{teal, opacity=0.2};
         \tikzfillbetween[of=y and secondorder, on layer=ft]{orange, opacity=0.2};
          \draw[scale=2.2] (0.1,3.4) node[left] {$m$};
              \draw[scale=2.2] (1.7,3.32) node[left] {$m^{**}(q)$};
\draw[scale=2.2] (2.2,0.43) node[left] {$m'(q)$};
\draw[scale=2.2] (2.85,0.4) node[left] {$m^*(q)$};
 \draw[scale=2.2] (0.2,-0.1) node[left] {$0$};
      \draw[scale=2.2] (3,-0.1) node[left] {$1$};
\end{tikzpicture}
\\
\begin{tabular}{c|c|c|c|c}
Region & I. & II. & III. & IV. \\ \hline
Characterization & $0<m<m^*(q)$ & $m^*(q)<m<m'(q)$ & $m'(q)<m<m^{**}(q)$ & $m>m^{**}(q)$ \\ \hline
$\begin{smallmatrix} \text{Stability of} \\ (\bar n_{1a},0,0,0) \end{smallmatrix}$ & stable & unstable & unstable & unstable  \\ \hline
$\begin{smallmatrix} \text{Existence of} \\ (n_{1a}^*,n_{1d}^*,n_{1i}^*,n_2^*) \end{smallmatrix}$ & does not exist & exists & exists & exists \\ \hline 
$\begin{smallmatrix} \text{Stability of} \\ (n_{1a}^*,n_{1d}^*,n_{1i}^*,n_2^*) \end{smallmatrix}$ & - & stable & stable & unstable \\ \hline
$\begin{smallmatrix} \text{Asymptotic} \\ \text{behaviour} \\ \text{of positive} \\ \text{solutions to~\eqref{4dimvirus}} \end{smallmatrix}$ & $\begin{smallmatrix} \text{eventually coord.} \\ \text{monotone convergence} \\ \text{to } (\bar n_{1a},0,0,0) \end{smallmatrix}$ & $\begin{smallmatrix}  \text{eventually coord.} \\ \text{monotone convergence} \\ \text{to } (n_{1a}^*,n_{1d}^*,n_{1i}^*,n_2^*) \end{smallmatrix}$ & $\begin{smallmatrix} \text{oscillatory convergence} \\ \text{to } (n_{1a}^*,n_{1d}^*,n_{1i}^*,n_2^*) \end{smallmatrix}$ & $\begin{smallmatrix} \text{periodic} \\ \text{behaviour} \end{smallmatrix}$
\end{tabular}
\caption{Left: The case when $r$ is small compared to $v$ (e.g., $r=0$). The Hopf bifurcation point $m^{**}$ reaches $+\infty$ at some value $q \in (0,1)$, 
Right: If $r$ is large compared to $v$, $m^{**}(q)=\infty$ holds for all $q \in [0,1)$, thus $(n_{1a}^*,n_{1d}^*,n_{1i}^*,n_2^*)$ is stable for all $m>m^*(q)$. In both cases, $m'$ and $m^*$ only diverge as $\uparrow 1$. In the coloured regions, we expect the behaviour explained in the tabular below the images.
}\label{fig-smallrfantasy}
\end{figure}
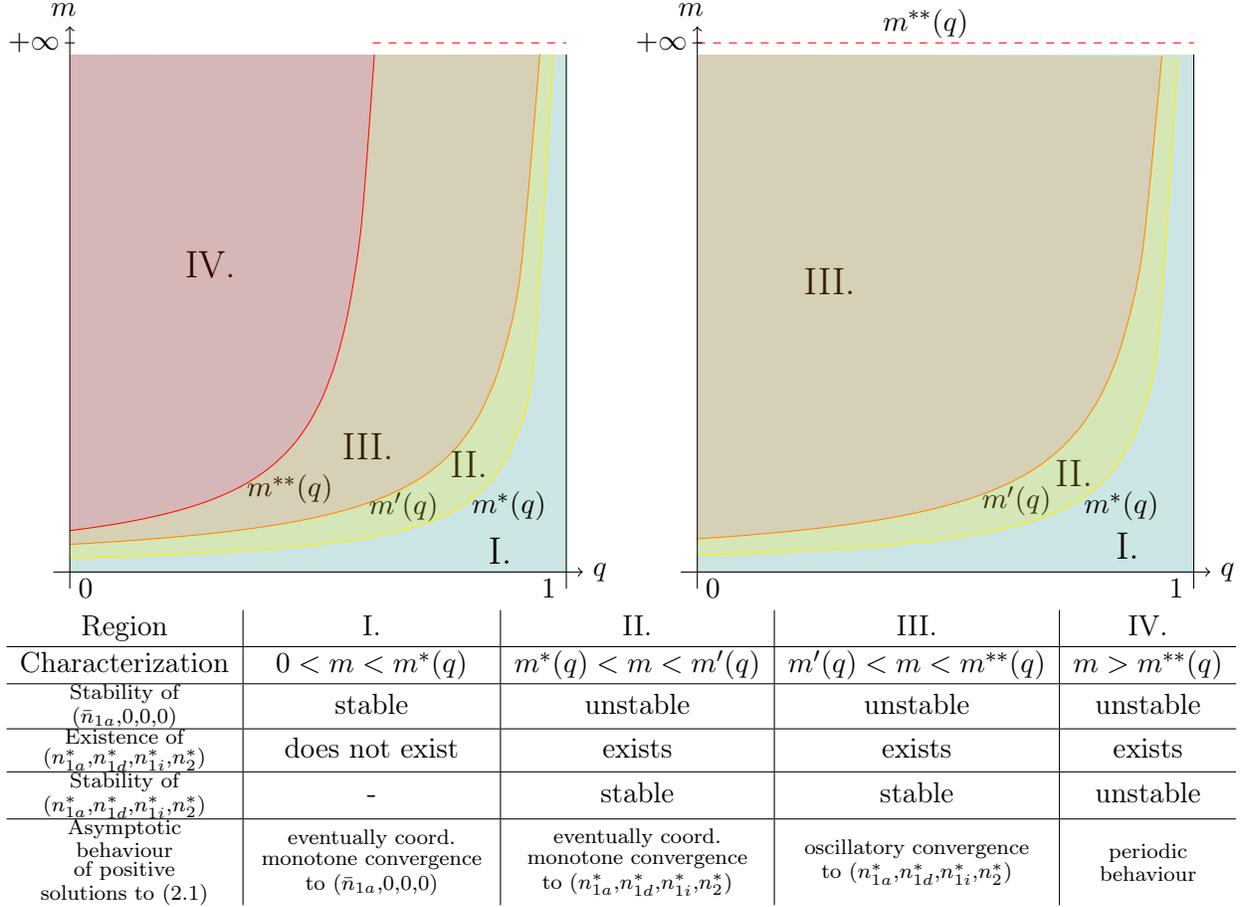
\end{example}

\begin{example}[The dormancy peak]\label{ex-GW-peak}
As previously observed in~\cite{GW15}, for the choice of parameters of that paper, at time close to zero, 
the virus population triggers dormancy in the active population, with the majority of the active population becoming dormant after a short time (orange line in Figure~\ref{fig-GW-peak}), given that $q$ is close enough to 1. The images of Figure~\ref{fig-GW-peak} depict the total population sizes of all types (1a, 1d, 1i, and 2), the proportion of dormant individuals among all host (type 1, i.e.\ type 1a, 1d or 1i) individuals, and the total amount of type 1 individuals, respectively. In terms of our stochastic process that is approximated by~\eqref{4dimvirus}, this means that in the beginning of the early macroscopic phase, shortly after the dynamical system approximation has become applicable, the majority of the active individuals becomes dormant.
\begin{figure}
\caption{
Here, the parameters are chosen according to \cite[Figures 5a, 5b, 6]{GW15}: $\lambda_1-\mu_1=0.23$, $C=\frac{\lambda_1-\mu_1}{9 \times 10^8}$, $D=1.02 \times 10^{-7}$, $q=50/51$, $\sigma=r=1/72$, $\kappa=v=1/24$, $\mu_2=1/12$, and $m=10$. The dormant, infected, and virus coordinates of the initial condition are given analogously to the other simulations in Section~\ref{sec-simulationsvirus}. Let us note that in this case, any $m \geq 2$ leads to periodic behaviour for the dynamical system~\eqref{4dimvirus}.}
\hspace{-20pt} \includegraphics[scale=1]{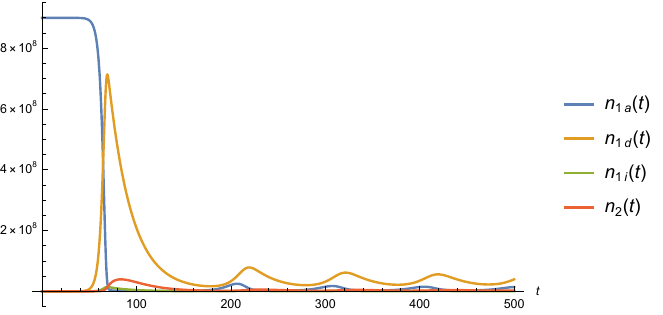} \\
\hspace{23pt} \includegraphics[scale=1]{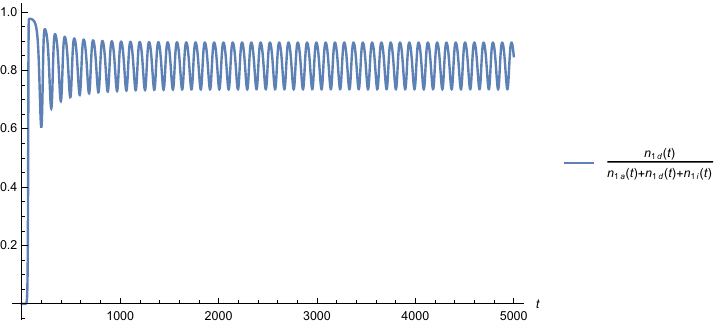} \\
\hspace{32pt} \includegraphics[scale=1]{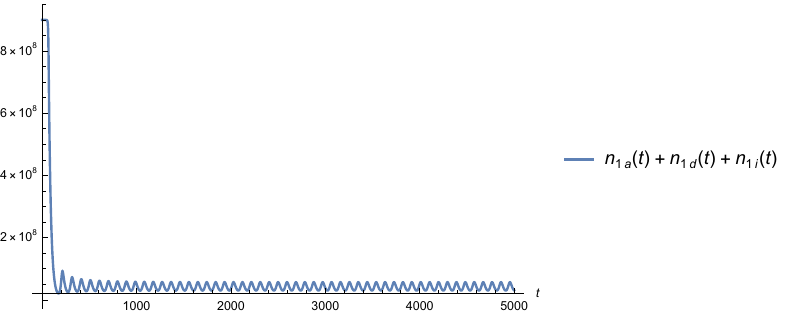} \\ \label{fig-GW-peak}
\end{figure}
\end{example}

Finally, let us comment on the case of quick resuscitation of dormant cells, that is, diverging $\sigma$. 
\begin{example}[Effects of large $\sigma$]\label{ex-largesigma}
The case of very large $\sigma$ corresponds to almost instantaneous resuscitation of dormant individuals after falling dormant. Thus, it is plausible to think that the qualitative behaviour of the active, infected, and virus coordinates of the system \eqref{4dimvirus} behaves very similarly to the case where there is no dormancy but the parameter $D$ of virus attacks is reduced by a factor of $1-q$. 
In Figure~\ref{fig-largesigma} we consider a solution of~\eqref{4dimvirus} with the default choice of parameters, apart from $\sigma$ which we choose as very large ($\sigma=100$, as opposed to the default value $\sigma=2$), also in comparison to the value of $\kappa$ (being equal to 1). We see that the behaviour of this solution is very similar to the one of \eqref{3dimvirus} with the same initial condition and with the same choice of the parameters apart from $q$ being altered to $0$ and $D$ to $(1-q)D$.
\begin{figure}
    \centering
    \includegraphics[scale=1]{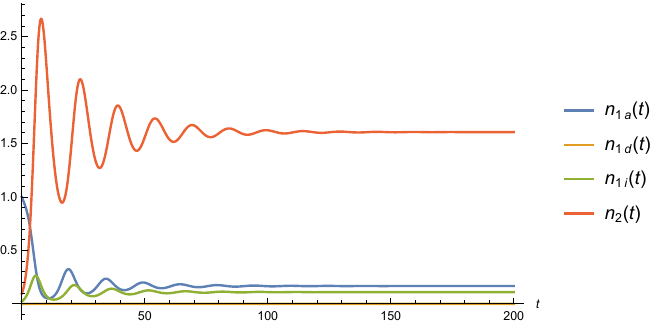}
    \includegraphics[scale=1]{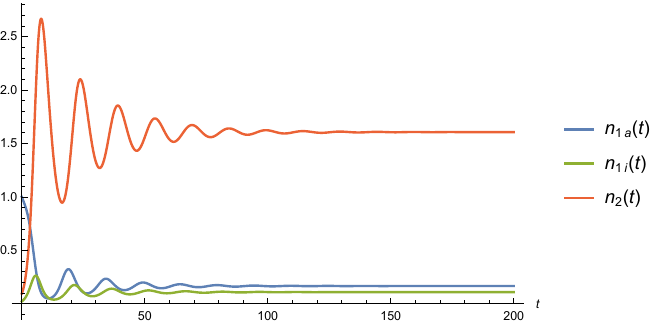}
    \caption{Solution of~\eqref{4dimvirus} with large $\sigma$ ($\sigma=100$) (top) and the one of the corresponding solution of~\eqref{3dimvirus} (bottom), for $m=5$. In the solution of~\eqref{3dimvirus} there is no dormant coordinate, whereas the dormant coordinate of the solution of~\eqref{4dimvirus} stays very close to zero.}
    \label{fig-largesigma}
\end{figure}
\end{example}


\subsection{Paradox of enrichment}\label{sec-paradoxofenrichment}
The fact that the coexistence equilibrium can lose its stability is a variant of the phenomenon called \emph{paradox of enrichment} in ecology, which is well-known from the context of predator--prey type dynamical systems (see e.g.\ \cite{MM90}, or \cite{Kuwamura2009b} for a case with predator dormancy). It rests on a bifurcation that appears in our model as well as in the variant studied in~\cite{BK98} in the following way: When the burst size $m$ reaches a critical threshold, the coexistence equilibrium emerges and is initially stable. However, further increase in the burst size  destabilizes it, giving rise to periodic limiting behaviour. 

In the predator--prey context, the analogue of the burst size expresses how much energy the predator can gain out of a consumed unit prey, and the analogue of the equilibrium population size $\barn1a$ of active hosts is the carrying capacity of the prey population. 

Now, increasing the carrying capacity of the system while keeping all other parameters constant leads to periodic cycles with increasing amplitudes, where the lowest population size during a period approaches zero for both for the prey and the predators. This corresponds to an increased danger of extinction due to small stochastic fluctuations in the underlying individual-based model. The `paradox' consists in the counter-intuitive effect that increasing carrying capacities may actually increase the risk of extinction for the whole system.

In our model, for $r$ small compared to $v$ and $q$ not too close to 1, a similar high-amplitude periodicity (with low minimum value) can be observed (see e.g.\ Figure~\ref{fig-GW-peak}). While varying $\bar n_{1a}$ (which can e.g.\ be achieved via fixing $\lambda_1-\mu_1$ and varying $C$, or the other way around) would be analogous to the predator--prey setting, we used $m$ as a bifurcation parameter. Nevertheless, we know from~\cite[Sections 3 and 5]{BK98} that under suitable assumptions on the other parameters that we fix, a Hopf bifurcation can also be observed while varying $\bar n_{1a}$ (more precisely, its analogue in a rescaled variant of the system~\eqref{3dimvirus} for $r=0$). Note that if we use $\bar n_{1a}$ as a bifurcation parameter, fixing $m$, it is crucial to assume that $mv>r+v$, otherwise the coexistence equilibrium will never exist and thus it cannot lose its stability.

It is further remarkable that as long as the coexistence equilibrium exists, its active coordinate $n_{1a}^*$ does not depend on $\lambda_1,\mu_1,C$ (which are the only parameters $\bar n_{1a}$ depends on) but on the other parameters of the model. This is in analogy to the fact that in certain predator--prey models, the prey coordinate of the coexistence equilibrium between predators and prey does not depend on the carrying capacity of the prey, see e.g.\ \cite[Section 2]{Kuwamura2009b}.

In that paper, the particular effect  of predator dormancy is studied. The authors showed that this evolutionary trait can help avoid periodic behaviour or at least decrease the amplitude of the periods, even in the case where dormant predators never resuscitate. This is  a clear analogy to our model, where we also observe that dormancy may prevent bifurcations (see Section~\ref{sec-simulationsvirus}). Here, virions correspond to the predators and active individuals to the prey. However, there are also significant differences between the two models, for example the lack of an analogue of the infected state in the predator--prey setting. Further, in our setting, the burst size $m$ and the active equilibrium size $\barn1a$ are two independent parameters, both of which correspond to the carrying capacity of prey. In turn, the predator--prey interaction in~\cite{Kuwamura2009b} is more complex than the quadratic interaction between viruses and active individuals in~\eqref{4dimvirus}. 

\section{Proofs}\label{sec-proofsvirus}
In this section, we first prove Lemma~\ref{lemma-lambdatilde} regarding the mean matrix of the approximating branching processes in Section~\ref{sec-branchingpreliminary}, as well as the assertions of Section~\ref{ssn-dynsyst} regarding the dynamical system~\eqref{4dimvirus} in Section~\ref{sec-preliminaryproofs}, and then we turn to the analysis of our original stochastic process $(\mathbf N_t)_{t \geq 0}$. Section~\ref{sec-phase1virus} contains the proofs of results about the `very early' (stochastic) phase of the host--virus dynamics, whereas in Section~\ref{sec-finalproofvirus} we put these together with the assertions about the dynamical system to conclude Theorems~\ref{thm-viruscoexprob}, \ref{thm-successvirus}, and \ref{thm-failurevirus}.
\subsection{The eigenvalues of the mean matrices of the branching processes}\label{sec-branchingpreliminary}
In this section, we carry out the proof of Lemma~\ref{lemma-lambdatilde}.
\begin{proof}[Proof of Lemma~\ref{lemma-lambdatilde}.]
The trace of $J_2$ is strictly negative by assumption, hence $J_2$ can have at most one eigenvalue with positive real part. We claim that there is such an eigenvalue with positive real part if and only if $\det J_2<0$, in other words,
\[ (m v-(r+v))(1-q) D\barn1a > \mu_2(r+v), \numberthis\label{viruscoexcondbranchingprocess} \]
i.e., \eqref{viruscoexcond} holds.
Indeed, if the determinant is negative, then the two eigenvalues must be real because the product of two complex (and hence conjugate) eigenvalues would be positive. This implies that one of the eigenvalues must be positive and the other negative.

The largest eigenvalue $\widetilde \lambda$ of $J_2$ is given as the largest solution $\lambda$ to
\[ \lambda^2 + (r+v+(1-q) D \barn1a + \mu_2)\lambda+(r+v)((1-q)D\barn1a+\mu_2) -(1-q) D \barn1amv=0, \]
and thus it is given by \eqref{lambdatildedefvirus}.
\end{proof}
\subsection{Proofs regarding the dynamical system}\label{sec-preliminaryproofs} In this section we verify our results regarding the system~\eqref{4dimvirus}. 
We start with the proof of Proposition~\ref{lemma-coexistencevirus}.
\begin{proof}[Proof of Proposition~\ref{lemma-coexistencevirus}.]
Let us compute the coordinates of the coexistence equilibrium $(n_{1a}^*,n_{1d}^*,n_{1i}^*,n_2^*)$; this will also yield its uniqueness. Setting the second and the third line of \eqref{4dimvirus} equal to zero, we obtain
\[ n_{1d}^* = \frac{q D n_{1a}^*n_2^*}{\kappa\mu_1+\sigma} \numberthis\label{dormantcoord} \]
and
\[ n_{1i}^* = \frac{(1-q) D n_{1a}^* n_2^*}{r+v}. \numberthis\label{infectedcoord} \]
Hence, according to the fourth line of the same dynamical system,
\[ n_{1i}^*=\frac{(1-q) D n_{1a}^* n_{1i}^* (mv-(r+v))}{\mu_2(r+v)}. \]
Under the condition that $n_{1i}^* \neq 0$, this is equivalent to
\[ n_{1a}^* = \frac{\mu_2(r+v)}{(1-q) D(mv-(r+v))}, \numberthis\label{n1adefvirusoncemore} \]
which is \eqref{n1adefvirus}.
Under the condition that $mv>r+v$, we have $n_{1a}^*>0$. Hence, by \eqref{dormantcoord} and \eqref{infectedcoord}, $n_{1d}^*$, $n_{1i}^*$, and $n_2^*$ all have the same sign. If $mv=r+v$, then $n_{1a}^*$ is not well-defined under our assumption that $n_{1i}^* \neq 0$. In particular, a coordinatewise positive equilibrium cannot exist. Finally, if $mv<r+v$, then $n_{1i}^* \neq 0$ implies $n_{1a}^*<0$ and thus there is no coordinatewise positive equilibrium.

Let us now assume that $mv>r+v$. Then, setting the first equation of \eqref{4dimvirus} equal to zero and substituting the equations \eqref{dormantcoord} and \eqref{infectedcoord} into it, we obtain
\[ \begin{aligned}
 0 & = n_{1a}^*(\lambda_1-\mu_1-C(n_{1a}^*+n_{1d}^*+n_{1i}^*)-Dn_2^*)+\sigma n_{1d}^*+r n_{1i}^* \\ & = n_{1a}^*\Big(\lambda_1-\mu_1-C\big(n_{1a}^*+\frac{q D n_{1a}^*n_2^*}{\kappa\mu_1+\sigma} +\frac{(1-q) D n_{1a}^* n_2^*}{r+v}\big)-D n_2^*\Big) +  \frac{\sigma q D n_{1a}^* }{\kappa\mu_1+\sigma} n_2^* + \frac{r (1-q) D n_{1a}^*}{r+v} n_2^*.\end{aligned}  \numberthis\label{n2expressed}\]
From this equation and \eqref{n1adefvirusoncemore}, an explicit expression for $n_{2}^*$ can be obtained, and this together with \eqref{dormantcoord} and \eqref{infectedcoord} can be used in order to derive explicit formulas for $n_{1d}^*$ and $n_{1i}^*$. 

Now we can complete the proof of Proposition~\ref{lemma-coexistencevirus}. Setting the right-hand side of \eqref{n2expressed} (which is equal to the right-hand side of the first equation in \eqref{4dimvirus}) equal to zero and dividing both sides by $n_{1a}^*$ implies that
\[ \lambda_1-\mu_1-C (n_{1a}^*+n_{1d}^*+n_{1i}^*) = n_{2}^* D \Big( 1 -\frac{q\sigma}{\kappa\mu_1+\sigma}-\frac{(1-q)r}{r+v} \Big). \numberthis\label{firstlinecoex} \]
Now, since $q \in (0,1)$, $\kappa\geq 0$, and $r,v,\sigma>0$, we have that $1 -\frac{q\sigma}{\kappa\mu_1+\sigma}-\frac{(1-q)r}{r+v}$ is strictly positive. Further, by \eqref{dormantcoord} and \eqref{infectedcoord}, $n_{1d}^*,n_{1i}^*,n_2^*$ all have the same sign.
Hence, thanks to \eqref{firstlinecoex}, we have the following. 
\begin{enumerate}[(i)]
\item\label{first-coexsign} Assume that $n_{1a}^*+n_{1d}^*+n_{1i}^* \geq \barn1a=\frac{\lambda_1-\mu_1}{C}$. Then $n_2^* \leq 0$. Hence, $n_{1d}^* \leq 0$ and $n_{1i}^*\leq 0$. Thus, $n_{1a}^*\geq \barn1a$. 
\item\label{second-coexsign} On the other hand, assume that $n_{1a}^*+n_{1d}^*+n_{1i}^* \leq \barn1a$. Then $n_2^* \geq 0$. Hence, $n_{1d}^* \geq 0$ and $n_{1i}^* \geq 0$. Thus, $n_{1a}^*\leq \barn1a$.
\end{enumerate}
In particular, we have obtained that if $n_{1a}^*+n_{1d}^*+n_{1i}^*=\barn1a$, then $n_{1d}^*,n_{1i}^*,n_2^*$ must be zero and hence $n_{1a}^*=\barn1a$. Finally, the arguments \eqref{first-coexsign} and \eqref{second-coexsign} also hold with `$\geq$' replaced by `$>$' and `$\leq$' by `$<$' everywhere. From this we derive that $n_{1d}^*,n_{1i}^*,n_2^*$, and $\lambda_1-\mu_1-Cn_{1a}^*$ (as well as $\barn1a-n_{1a}^*$) have the same sign. These expressions being positive is equivalent to the condition~\eqref{viruscoexcond} thanks to the formula~\eqref{n1adefvirusoncemore} for $n_{1a}^*$, which implies the proposition. 
\end{proof}

We continue with the proof of Proposition~\ref{lemma-stabilityvirus}.
\begin{proof}[Proof of Proposition~\ref{lemma-stabilityvirus}.]
In any equilibrium $(\widetilde n_{1a},\widetilde n_{1d},\widetilde n_{1i},\widetilde n_{2}) \in [0,\infty)^4$, writing $\widetilde n_{1}=\widetilde n_{1a}+\widetilde n_{1d}+\widetilde n_{1i}$, we have the Jacobi matrix
\[ 
 A(\widetilde n_{1a},\widetilde n_{1d},\widetilde n_{1i},\widetilde n_2)=\begin{pmatrix}
\lambda_1-\mu_1-C (\widetilde n_{1a}+\widetilde n_{1})-  D \widetilde n_2 & \sigma-C \widetilde n_{1a} & r-C \widetilde n_{1a} & - D \widetilde n_{1a} \\
q D \widetilde n_{2} & -\kappa\mu_1-\sigma & 0 & q D \widetilde n_{1a} \\
(1-q) D \widetilde n_2 & 0 & -r-v & (1-q) D \widetilde n_{1a} \\
-(1-q) D \widetilde n_2 & 0 & mv & -(1-q) D \widetilde n_{1a}-\mu_2
\end{pmatrix}.
\]
Choosing $(\widetilde n_{1a},\widetilde n_{1d},\widetilde n_{1i},\widetilde n_{2})=(0,0,0,0)$, we obtain
\[ A(0,0,0,0)=\begin{pmatrix}
\lambda_1-\mu_1 & \sigma & r & 0 \\
0 & -\kappa\mu_1-\sigma & 0 & 0 \\
0 & 0 & -r-v & 0 \\
0 & 0 & mv & -\mu_2
\end{pmatrix}. \]
Clearly, the eigenvalues of this matrix are its diagonal entries. Thanks to the assumptions that $\lambda_1>\mu_1>0$, $\kappa\geq 0$, $\sigma,r,v,\mu_2>0$, we see that $A(0,0,0,0)$ has precisely one positive eigenvalue: $\lambda_1-\mu_1$, whereas all other eigenvalues are strictly negative. Hence, the equilibrium $(0,0,0,0)$ is unstable, as claimed in part~\eqref{0saddle} of the proposition. On the other hand, at $(\barn1a,0,0,0)$, the Jacobi matrix is given as follows
\[ A(\barn1a,0,0,0)=\begin{pmatrix}
-(\lambda_1-\mu_1) & \sigma-C\bar n_{1a} & r-C\bar n_{1a} & -D\barn1a \\
0 & -\kappa\mu_1-\sigma & 0 & q D \barn1a \\
0 & 0 & -r-v & (1-q)D\barn1a \\
0 & 0 & mv & -(1-q)D\barn1a -\mu_2
\end{pmatrix} \numberthis\label{Abarn1a}
\]
We immediately see that $-(\lambda_1-\mu_1)<0$ is an eigenvalue of this matrix (with eigenvector $(1,0,0,0)^T$). The remaining three eigenvalues are the ones of the last $3\times 3$ block of $A(\barn1a,0,0,0)$. We recognize this block as the transpose of the mean matrix $J$ defined in \eqref{Jdefvirus}, which has the same eigenvalues as $J$. Consequently, if the condition~\eqref{viruscoexcond} holds, then $A(\barn1a,0,0,0)$ is indefinite with three negative eigenvalues and one positive one, and hence $(\barn1a,0,0,0)$ is unstable. In contrast, if \eqref{lifecond} holds,
then all eigenvalues of $A(\barn1a,0,0,0)$ are strictly negative, and thus $(\barn1a,0,0,0)$ is asymptotically stable. Thus, part~\eqref{nopersistence} of the proposition follows.

Let us finally consider the coexistence equilibrium $(n_{1a}^*,n_{1d}^*,n_{1i}^*,n_2^*)$ under condition~\eqref{viruscoexcond}. Writing $n_1^*=n_{1a}^*+n_{1d}^*+n_{1i}^*$, the Jacobi matrix is given as follows
\[ 
 A(n_{1a}^*,n_{1d}^*,n_{1i}^*,n_2^*)=\begin{pmatrix}
\lambda_1-\mu_1-C (n_1^*+n_{1a}^*)-  D n_2^* & \sigma-C n_{1a}^* & r-C n_{1a}^* & - D n_{1a}^* \\
q D n_{2}^* & -\kappa\mu_1-\sigma & 0 & q D n_{1a}^* \\
(1-q) D n_2^* & 0 & -r-v & (1-q) D n_{1a}^* \\
-(1-q) D n_2^* & 0 & mv & -(1-q) D n_{1a}^*-\mu_2
\end{pmatrix}.
\]
Our first goal is to show that the determinant of this matrix is positive. Substracting $\frac{q}{1-q}$ times the third row from the second row, we obtain the matrix
\[  \widetilde A(n_{1a}^*,n_{1d}^*,n_{1i}^*,n_2^*)=\begin{pmatrix}
\lambda_1-\mu_1-C(n_1^*+n_{1a}^*)-  D n_2^* & \sigma-C n_{1a}^* & r-C n_{1a}^* & - D n_{1a}^* \\
0 & -\kappa\mu_1-\sigma & \frac{q(r+v)}{1-q} & 0 \\
(1-q) D n_2^* & 0 & -r-v & (1-q) D n_{1a}^* \\
-(1-q) D n_2^* & 0 & mv & -(1-q) D n_{1a}^*-\mu_2
\end{pmatrix},
\]
which has the same determinant as $ A(n_{1a}^*,n_{1d}^*,n_{1i}^*,n_2^*)$.  We now claim that the last $3\times 3$ block of  $\widetilde A(n_{1a}^*,n_{1d}^*,n_{1i}^*,n_2^*)$ has determinant zero. Indeed, we have
\[
\begin{aligned}
 \det& \begin{pmatrix} -\kappa\mu_1-\sigma & \frac{q}{1-q}(r+v) & 0 \\ 0 & -r-v & (1-q) Dn_{1a}^* \\ 0 & mv & -(1-q) D n_{1a}^*-\mu_2 \end{pmatrix}\\ & =(-\kappa\mu_1-\sigma) \big( ((1-q) D n_{1a}^*+\mu_2) (r+v)-(1-q)mvDn_{1a}^* \big) \\ &
 = (-\kappa\mu_1-\sigma) ((1-q)Dn_{1a}^*(r+v-mv)+\mu_2 (r+v)) =0,
 \end{aligned} \]
where in the last step we used the definition~\eqref{n1adefvirus} of $n_{1a}^*$. Hence, by Laplace's expansion theorem applied to the first column of $\widetilde A(n_{1a}^*,n_{1d}^*,n_{1i}^*,n_2^*)$, we have
\[
\begin{aligned}
\det& A(n_{1a}^*,n_{1d}^*,n_{1i}^*,n_2^*)=\det \widetilde A(n_{1a}^*,n_{1d}^*,n_{1i}^*,n_2^*) \\ &=
(1-q)Dn_2^* \Big[ (\sigma-Cn_{1a}^*)\frac{q}{1-q}(r+v)(-(1-q)Dn_{1a}^*-\mu_2) + Dn_{1a}^*(\kappa\mu_1+\sigma) mv  \\ & \qquad -(r-Cn_{1a}^*)(\kappa\mu_1+\sigma)((1-q)Dn_{1a}^*+\mu_2) + (\sigma-Cn_{1a}^*)\frac{q}{1-q}(r+v)(1-q)Dn_{1a}^* \\ & \qquad -Dn_{1a}^*(\kappa\mu_1+\sigma)(r+v)   +(\kappa\mu_1+\sigma)(r-Cn_{1a}^*)(1-q)Dn_{1a}^* \Big] \\
& = 
(1-q)Dn_2^* \Big[ Cn_{1a}^*\frac{q}{1-q}(r+v)\mu_2+Cn_{1a}^*(\kappa\mu_1+\sigma)\mu_2 -\sigma\frac{q}{1-q}(r+v)\mu_2 + Dn_{1a}^*(\kappa\mu_1+\sigma) mv  \\ & \qquad -r(\kappa\mu_1+\sigma)\mu_2 -Dn_{1a}^*(\kappa\mu_1+\sigma)(r+v)   \Big] 
\\ & = (1-q) D n_2^* \Big[ Cn_{1a}^*\frac{q}{1-q}(r+v)\mu_2+Cn_{1a}^*(\kappa\mu_1+\sigma)\mu_2 \\ & \qquad +(\kappa\mu_1+\sigma) \big( \frac{\mu_2(r+v)}{1-q} -\mu_2 r \big) -\sigma\frac{q}{1-q} \mu_2(r+v)\Big] \\ & = D n_2^* \big[ Cn_{1a}^*q(r+v)\mu_2+Cn_{1a}^*(1-q)(\kappa\mu_1+\sigma)\mu_2 \\ & \qquad + \kappa\mu_1\mu_2r+\kappa\mu_1\mu_2 v + \sigma \mu_2 r +\sigma \mu_2 v - \kappa\mu_1(1-q)\mu_2 r - \sigma (1-q)\mu_2 r-\sigma q \mu_2 r - \sigma q \mu_2 v \big]
\\ & = D n_2^* \mu_2[Cn_{1a}^*q(r+v)+Cn_{1a}^*(1-q)(\kappa\mu_1+\sigma)+\kappa\mu_1r q +\sigma  v(1-q) + \kappa\mu_1 v] >0.
\end{aligned}   \numberthis\label{detcomp} \]
where in the third equality we again used~\eqref{n1adefvirus} and in the last step we used the positivity of $n_{1a}^*$ and $n_{2}^*$. 
Further, the trace of the matrix is negative because all its diagonal entries are negative. Indeed, also the first entry of the first column is negative, which follows from the fact that since $(n_{1a}^*,n_{1d}^*,n_{1i}^*,n_2^*)$ is an equilibrium of \eqref{4dimvirus} with four positive coordinates, we have
\[ \lambda_1-\mu_1-C(n_{1a}^*+n_{1d}^*+n_{1i}^*)-Dn_2^* = -\frac{\sigma n_{1d}^*+ r n_{1i}^*}{n_{1a}^*}, \]
and hence
\[ \lambda_1-\mu_1-2Cn_{1a}^*-C(n_{1d}^*+n_{1i}^*)-Dn_2^* < -\frac{\sigma n_{1d}^*+ r n_{1i}^*}{n_{1a}^*}<0.\]
This concludes the proof of the proposition. 
\end{proof}


Next, we carry out the proof of Proposition~\ref{prop-Lyapunov}.

\begin{proof}[Proof of Proposition~\ref{prop-Lyapunov}]
Proposition~\ref{prop-Lyapunov} is the analogue of the assertion \cite[Lemma 4.1]{BK98} that treated the case without dormancy or recovery and slightly with different competition. Our proof is indeed the analogue of the one in \cite{BK98}, which relies on the idea of Chetaev's instability theorem \cite{C61}. The main additional step of our proof is that the dormant coordinate should not be treated analogously to the infected and the virus coordinate, but it should just be ignored. Since it is possibly not straightforward to see that this approach works out, we present a self-contained proof as follows.

Let $V \colon [0,\infty)^4 \to \R$, $(\widetilde n_{1a},\widetilde n_{1d},\widetilde n_{1i},\widetilde n_2) \mapsto w_{1i} \widetilde n_{1i} + w_2 \widetilde n_2$ for some $w_{1i},w_2>0$. Let us write  the system \eqref{4dimvirus} as $\dot{ \mathbf n}(t)= f(\mathbf n(t))$ and fix $\eps>0$. Then, the standard Euclidean scalar product of the gradient of $V$ and $f$ at $(\widetilde n_{1a},\widetilde n_{1d},\widetilde n_{1i},\widetilde n_2) \in [0,\infty)^4$ with $\widetilde n_{1a} > \barn1a-\eps$ equals 
\[ 
\begin{aligned}
\langle \nabla V, f \rangle|&_{(\widetilde n_{1a},\widetilde n_{1d},\widetilde n_{1i},\widetilde n_2)} = w_{1i} ((1-q) D \widetilde n_{1a} \widetilde n_2 - \widetilde n_{1i} (r+v)) + w_2 (-(1-q)D\widetilde n_{1a}\widetilde n_2 +mv \widetilde n_{1i} -\mu_2 \widetilde n_2)
\\ & = \widetilde n_{1i} \big[mvw_2-w_{1i} (r+v)\big] + \widetilde n_2 \big[ (1-q) D \widetilde n_{1a} w_{1i} - (1-q) D \widetilde n_{1a} w_2 -\mu_2 w_2 \big] 
\\ & >  \widetilde n_{1i} \big[mvw_2-w_{1i} (r+v)\big] + \widetilde n_2 \big[ (1-q) D (\barn1a-\eps) w_{1i} - (1-q) D \widetilde n_{1a} w_2 -\mu_2 w_2 \big]. 
\end{aligned}\]
Hence, $\langle \nabla V, f \rangle|_{(\widetilde n_{1a},\widetilde n_{1d},\widetilde n_{1i},\widetilde n_2)}$ is positive once
\[ m v w_2 > (r+v) w_{1i} \qquad \text{ and } \qquad (1-q) w_{1i} D( \barn1a-\eps) > ((1-q) D(\barn1a-\eps)+\mu_2)w_2, \numberthis\label{posdefvirus} \]
in other words,
\[ \frac{mv}{r+v} w_2 > w_{1i} > w_2 \frac{(1-q)+ \frac{\mu_2}{D(\barn1a-\eps)}}{1-q}= w_2 \Big[ 1 + \frac{ \mu_2}{(1-q)D(\barn1a-\eps)} \Big]. \]
Since $w_{1i}>0,w_2>0$, this requires
\[ \barn1a-\eps > \frac{\mu_2(r+v)}{(1-q)D(mv-(r+v))}, \]
which holds whenever $\eps \in (0, \barn1a-n_{1a})$, where we recall that $\barn1a>n_{1a}^*$ under the condition \eqref{viruscoexcond}. Then we can indeed choose $w_{1i},w_2>0$ satisfying \eqref{posdefvirus}, and thus we can find $d>0$ such that for such a choice of $w_{1i},w_2$, and $\eps$, we have
\[ \nabla V > d V \qquad \text{on } B_\eps((\barn1a,0,0,0)) \cap (0,\infty)^4 \numberthis\label{Vgrowsexponentially} \]
where for $x \in\R^4$ and $\varrho>0$, $B_\varrho(x)$ denotes the open $\ell^2$-ball of radius $\varrho$ around $x$. 

Now, let us assume that $(n_{1a}(0),n_{1d}(0),n_{1i}(0),n_2(0)) \in (0,\infty)^4$. Then it is clear that for all $t>0$, $n_{1i}(t) \neq 0$ and $n_{2}(t) \neq 0$. Now, if $\lim_{t\to\infty} (n_{1a}(t),n_{1d}(t),n_{1i}(t),n_2(t)) = (\barn1a,0,0,0)$, there exists $t_0>0$ such that for all $t>0$, $(n_{1a}(t),n_{1d}(t),n_{1i}(t),n_2(t)) \in B_\eps((\barn1a,0,0,0)) \cap (0,\infty)^4$. Hence, by \eqref{Vgrowsexponentially}, $\lim_{t\to\infty} V(n_{1a}(t),n_{1d}(t),n_{1i}(t),n_2(t)) = \infty$, which contradicts the assumption that $\lim_{t\to\infty} (n_{1i}(t),n_2(t))=(0,0)$. 

From this it is in fact easy to derive that $(n_{1a}(t),n_{1d}(t),n_{1i}(t),n_2(t))$  cannot even converge to $(\barn1a,0,0,0)$ along any diverging sequence of times, but let us provide the details for completeness.  Since $V$ is positive definite on  $\mathcal B:=B_\eps((\barn1a,0,0,0)) \cap (0,\infty)^4$, the $\omega$-limit set $\Omega_0$ of any solution of \eqref{4dimvirus} (i.e., the set of subsequential limits of the solution as $t\to\infty$) started from $\mathcal B$ satisfies
\[
\Omega_0 \cap \overline{\mathcal B} \subseteq \{ \langle \nabla V, f \rangle = 0 \}
\] 
where $\overline{\mathcal B}$ is the closure of $\mathcal B$. 
In terms of these objects, we have already verified that $(\barn1a,0,0,0) \in \{ \langle \nabla V, f \rangle = 0 \}$ and that $(\barn1a,0,0,0) \neq \Omega_0 \cap \overline{\mathcal B}$.

Using the definition of $V$ and the fact that $[0,\infty)^4$ is positively invariant under \eqref{4dimvirus}, we conclude that $\Omega_0 \cap \overline{\mathcal B}$ contains only points of the form $(\widetilde n_{1a},\widetilde n_{1d},0,0)$, where $\widetilde n_{1a}, \widetilde n_{1d} >0$. However, if a coordinatewise nonnegative solution of \eqref{4dimvirus} started from $(0,\infty)^4$ is such that its infected and virus coordinate tend to zero, then its dormant coordinate must also tend to zero and hence its active coordinate to $\barn1a$. We conclude that $\Omega_0 \cap \overline{\mathcal B} \subseteq \{ (\barn1a,0,0,0) \}$. But since $(\barn1a,0,0,0) \neq \Omega_0 \cap \overline{\mathcal B}$, it follows that $\Omega_0 \cap \overline{\mathcal B} = \emptyset$, and thus the proposition is proven.
\end{proof}

Now, we verify Corollary~\ref{cor-persistence}.
\begin{proof}[Proof of Corollary~\ref{cor-persistence}.]
According to the properties of the linearized variant of the system~\eqref{4dimvirus} near $(0,0,0,0)$ (see the Jacobi matrix $A(0,0,0,0)$ in the proof of Proposition~\ref{lemma-stabilityvirus}), if $(n_{1a}(0),n_{1d}(0),n_{1i}(0),n_2(0)) \in [0,\infty)^4$ with $n_{1a}(0) >0$, then $\liminf_{t\to\infty} n_{1a}(t)>0$. 

Next, note that if $\mathbf n(0) \in [0,\infty)^4$ with $n_{1a}(0)>0$, then there are two possibilities. Either $\max \{ n_{1i}(0), n_{2}(0) \}>0$ and hence $n_{1d}(t),n_{1i}(t),n_2(t)>0$ for all $t>0$, or $\max \{ n_{1i}(0),n_2(0)\}=0$ and hence $\lim_{t\to\infty} \mathbf n(t)=(\barn1a,0,0,0)$. Thanks to the invariance of $\omega$-limit sets, this implies that if the $\omega$-limit set of $(\mathbf n(t))_{t \geq 0}$ contains a point with a zero coordinate (which is necessarily not the type 1a coordinate), then in fact the $\omega$-limit set contains $(\barn1a,0,0,0)$, i.e.\ the solution converges to $(\barn1a,0,0,0)$ at least along a subsequence of times. But for a coordinatewise positive initial condition, that would contradict Proposition~\ref{prop-Lyapunov}, hence the positivity part of the proposition.
 
Next, let us verify that $\limsup_{t\to\infty} n_{1a}(t)+n_{1d}(t)+n_{1i}(t)<\barn1a$. Summing the first three lines of \eqref{4dimvirus}, we obtain
 \[ \dot n_{1a}(t)+\dot n_{1d}(t)+\dot n_{1i}(t) = n_{1a}(t)(\lambda_1-\mu_1-C (n_{1a}(t)+n_{1d}(t)+n_{1i}(t))-\kappa\mu n_{1d}(t)-v n_{1i}(t). \]
 Let us choose $\eps>0$ such that $\liminf_{t\to\infty} \kappa\mu n_{1d}(t)+v n_{1i}(t)>\eps$. Then if for some $t>0$ we have $n_{1a}(t)+n_{1d}(t)+n_{1i}(t) \geq \barn1a$, then we have
 \[ \frac{\d}{\d t} (n_{1a}(t)+n_{1d}(t)+n_{1i}(t)) < -\eps. \]
 Now, solutions of \eqref{4dimvirus} are continuously differentiable thanks to the Picard--Lindelöf theorem, and hence we obtain that there exists $\delta>0$ such that whenever $n_{1a}(t)+n_{1d}(t)+n_{1i}(t) \geq \barn1a-\delta$, we have
 \[ \frac{\d}{\d t} (n_{1a}(t)+n_{1d}(t)+n_{1i}(t)) < -\eps/2. \]
 This implies the time
 \[ t_{\barn1a-\delta} = \inf \big\{ t \geq 0 \colon n_{1a}(t)+n_{1d}(t)+n_{1i}(t)<\barn1a-\delta \big\} \]
 is finite, and for all $t>t_{\barn1a-\delta}$ we have $n_{1a}(t)+n_{1d}(t)+n_{1i}(t) \leq \barn1a-\delta < \barn1a$. Thus, $\limsup_{t\to\infty} n_{1a}(t)+n_{1d}(t)+n_{1i}(t)<\barn1a$.
 
 Finally, the asymptotic upper bound on $n_2(t)$ as $t\to\infty$ is the analogue of \cite[Lemma 2.3]{BK98} in our model. Relying on the already proven parts of Corollary~\ref{cor-persistence}, we can now provide a short proof for it.
Recall that under the assumptions of the corollary we have
\[ \limsup_{t\to\infty} n_{1a}(t)+n_{1d}(t)+n_{1i}(t)< \barn1a, \qquad \liminf_{t\to\infty} n_{j}(t)>0, \forall j \in \{1a,1i,2\}, \]
and hence there exists $\beta>0$ such that
\[ \limsup_{t\to\infty} n_{1i}(t) \leq \barn1a-\beta. \]
Thus, we obtain for all $t$ sufficiently large
\[ \dot n_{2} (t) = -(1-q)Dn_{1a}(t)n_2(t)+mv n_{1i}(t)-\mu_2 n_2(t) < mv (\barn1a-\beta)-\mu_2 n_2(t).  \numberthis\label{n2estimate} \]
 This shows that for such $t$, if $n_2(t) \geq \frac{mv(\barn1a-\beta)}{\mu_2}$, then $s \mapsto n_2(s)$ is decreasing at $t$. Consequently,
 \[ \limsup_{t\to\infty} n_2(t) \leq \frac{mv(\barn1a-\beta)}{\mu_2}<\frac{mv\barn1a}{\mu_2},  \]
 as wanted.
 \end{proof}

\begin{comment}{\begin{remark}
Note that the proofs of Proposition~\ref{prop-Lyapunov} and Corollary~\ref{cor-persistence} only rely on the fact that $n_{1i}(0)$ and $n_2(0)$ are positive, $n_{1d}(0)$ can as well be zero. This will substantially simplify the application of the Kesten--Stigum theorem in the proof of our main results because the dormant coordinate, which is individually subcritical and tedious to handle, can be ignored for part of the arguments.
\end{remark}}\end{comment}

\subsection{First phase: growth or extinction of the epidemic}\label{sec-phase1virus}
Our first goal is to prove that during the first phase of the virus invasion, with high probability, the rescaled type 1a population stays close to equilibrium $\barn1a$ and the type 1d population stays small compared to $K$ until the total type 1i and 2 population reaches a size of order $K$ or goes extinct. Recall the stopping time $T_\eps^2$, $\eps \geq 0$, from~\eqref{Teps2def}.
Let us further define the stopping time
\[ Q_{\eps} = \inf \big\{ t \geq 0 \colon (N_{1a,t}^K,N_{1d,t}^K) \notin [ \barn1a-\eps,\barn1a+\eps ] \times [0,\eps] \big\}, \]
the first time when the rescaled type 1a population leaves a neighbourhood of radius $\eps$ around the equilibrium $\barn1a$ or the rescaled 1d population reaches size $\eps$. Our main results regarding the first phase are summarized in the following proposition. 
\begin{prop}\label{prop-firstphasevirus}
Assume that $\widetilde\lambda \neq 0$. Let $K \mapsto m_{1a}^K$ be a function from $(0,\infty)$ to $[0,\infty)$ such that $m_{1a}^K \in \smfrac{1}{K} \N_0$ and $\lim_{K \to \infty} m_{1a}^K=\barn1a$. Then there exists a constant $b \geq 2$ and a function $f \colon (0,\infty) \to (0,\infty)$ tending to zero as $\eps \downarrow 0$ such that
\[ \limsup_{K \to \infty} \Big| \P \Big[ T_{\eps}^2 < T_0^2 \wedge Q_{b\eps}, \Big| \frac{T_{\eps}^2}{\log K} - \frac{1}{\widetilde \lambda} \Big| \leq f(\eps)~ \Big| \mathbf N_0^K=\big(m_{1a}^K,0,0,\frac{1}{K}\big)  \Big]-(1-s_2) \Big| = o_\eps(1) \numberthis\label{invasiontimevirus} \]
and
\[ \limsup_{K \to \infty} \Big| \P\Big[T_0^2 <T_{\eps}^2 \wedge Q_{b\eps}~ \Big|  \mathbf N_0^K=\big(m_{1a}^K,0,0,\frac{1}{K}\big)   \Big] - s_2 \Big|=o_{\eps}(1), \numberthis\label{secondofpropvirus} \]
where $o_{\eps}(1)$ tends to zero as $\eps \downarrow 0$.
\end{prop}
In order to prove the proposition, we first verify the following lemma.
\begin{lemma}\label{lemma-residentsstayvirus}
Under the assumptions of Proposition~\ref{prop-firstphasevirus}, there exist constants $\eps_0>0,b \geq 2$ such that for any $\eps \in (0, \eps_0]$,
\[ \limsup_{K \to \infty} \P \big(Q_{b\eps} \leq T^2_{\eps} \wedge T_0^2 \big)=0. \]
\end{lemma}
\begin{proof}
We follow the approach of the proof of \cite[Lemma 3.4]{C+19}, with the main difference being that instead of considering only types that are initially of size comparable to $K$ as `resident' types and all other types as `mutant', we also consider the initially microscopic type 1d like as resident (as we discussed in Section~\ref{sec-phase1heuristics}), which requires a bit more care.

We verify our lemma via coupling the rescaled population size $((N_{1a,t}^K,N_{1d,t}^K))_{t \geq 0}$ of types 1a and 1d with two birth-and-death processes, $(N_{1a}^1,N_{1d}^1)=((N_{1a,t}^1,N_{1d,t}^1))_{t \geq 0}$ and $(N_{1a}^2,N_{1d}^2)=((N_{1a,t}^2,N_{1d,t}^2))_{t \geq 0}$, on time scales where the total population of types 1i, and 2 is still small compared to $K$. (These processes will also depend on $K$, but we omit the notation $K$ from their nomenclature for simplicity.) To be more precise, we want to choose  $((N_{1a,t}^1,N_{1d,t}^1))_{t \geq 0}$ and $((N_{1a,t}^2,N_{1d,t}^2))_{t \geq 0}$ so that
\[ N_{\upsilon,t}^1 \leq N_{\upsilon,t}^K \leq N_{\upsilon,t}^2, \qquad \text{a.s.} \qquad \forall t \leq T_0^2 \wedge T_{\eps}^2,~\forall \upsilon \in \{ 1a, 1d \}. \numberthis\label{firstresidentcouplingvirus} \]

In order to satisfy \eqref{firstresidentcouplingvirus}, for all $\eps>0$ small enough, the processes $((N_{1a,t}^1,N_{1d,t}^1))_{t \geq 0}$ and $((N_{1a,t}^2,N_{1d,t}^2))_{t \geq 0}$ can be chosen with the following birth and death rates, for all $i,j \in \N_0$ (where certain kinds of transitions occur for only one of the two processes): 
\begin{align*}
(N_{1a}^1,N_{1d}^1) \colon  & \Big( \frac{i}{K},  \frac{j}{K} \Big) \to \Big( \frac{i+1}{K}, \frac{j}{K} \Big) & \text{at rate } & i\lambda_1, \\
 & \Big( \frac{i}{K}, \frac{j}{K} \Big)  \to  \Big( \frac{i-1}{K}, \frac{j}{K} \Big) & \text{at rate } &i \big( \mu_1+C \Big(\frac{i}{K}+\eps \Big) + D \eps \big), \\
  & \Big( \frac{i}{K}, \frac{j}{K} \Big) \to  \Big( \frac{i}{K}, \frac{j-1}{K} \Big) & \text{at rate } & j \kappa \mu_1, \\
  & \Big( \frac{i}{K}, \frac{j}{K} \Big) \to  \Big( \frac{i+1}{K}, \frac{j-1}{K} \Big) & \text{at rate } & j \sigma,
\end{align*}
and
\begin{align*}
(N_{1a}^2,N_{1d}^2) \colon  & \Big( \frac{i}{K}, \frac{j}{K} \Big) \to \Big( \frac{i+1}{K}, \frac{j}{K} \Big) & \text{at rate } & i\lambda_1+r\eps, \\
 & \Big( \frac{i}{K}, \frac{j}{K} \Big) \to  \Big( \frac{i-1}{K}, \frac{j}{K} \Big) & \text{at rate } &i \big( \mu_1+C \frac{i}{K}\big), \\
 & \Big( \frac{i}{K}, \frac{j}{K} \Big) \to  \Big( \frac{i}{K}, \frac{j+1}{K} \Big) & \text{at rate } & i q D  \eps,  \\
 & \Big( \frac{i}{K}, \frac{j}{K} \Big) \to  \Big( \frac{i}{K}, \frac{j-1}{K} \Big) & \text{at rate } & j \kappa \mu_1, \\
  & \Big( \frac{i}{K}, \frac{j}{K} \Big) \to  \Big( \frac{i+1}{K}, \frac{j-1}{K} \Big) & \text{at rate } & j \sigma,
\end{align*}
both started from $(m_{1a}^K,0)$ at time zero. 
Indeed, given the initial conditions, $((N_{1a,t}^1,N_{1d,t}^1))_{t \geq 0}$ has a dormant coordinate absorbed at zero, since it starts at zero and is monotonically decreasing. In particular, in this process, active individuals that would become dormant (having a chance to resuscitate) in our original population process $(\mathbf N_t^K)_{t \geq 0}$ just die immediately, and recoveries are also ignored. The rate of death of individuals of type 1a due to competition is maximized under the constraint that the total population of types 1i and 2 is at most $\eps K$. In contrast, in the process  $((N_{1a,t}^2,N_{1d,t}^2))_{t \geq 0}$, successful virus attacks are ignored, whereas the birth rate is increased by the maximal possible amount of recovered individuals on $[0,T_0^2 \wedge T_\eps^2]$. Unsuccessful virus attacks are replaced by events where a dormant individual emerges but no active individual gets lost. Further, deaths of type 1a individuals due to competition with other types (i.e.\ 1d or 2) are ignored. Together with the fact that the rates for birth, natural death of actives, death by competition, death of dormant individuals, and resuscitation rates are the same for all these three processes, it follows that the coupling~\eqref{firstresidentcouplingvirus} is satisfied.

Let us estimate the time until which the processes $(N_{1a}^1,N_{1d}^1) $ and $(N_{1a}^2,N_{1d}^2) $ stay close to the value $(\barn1a,0)$. We define the stopping times
\[ Q^{j}_{\eps}:=\inf \big\{ t \geq 0 \colon N_{1a,t}^j \notin [\barn1a-\eps,\barn1a+\eps] \text{ or } N_{1d,t}^j > \eps  \big\}, \qquad j\in \{1,2\},\eps>0. \]
As $K \to \infty$, according to \cite[Theorem 2.1, p.~456]{EK}, uniformly on any fixed time interval of the form $[0,T]$, $T>0$, $(N_{1a,t}^{1},N_{1d,t}^1)$ converges in probability to the unique solution to
\[ 
\begin{aligned}
\dot{n}_{1a,1}(t) & = n_{1a,1}(t)(\lambda_1-\mu_1-C (n_{1a,1}(t)+\eps)-D\eps)+\sigma n_{1d,1}(t), \\
\dot{n}_{1d,1}(t) & = -(\kappa\mu_1+\sigma) n_{1d,1}(t),
\end{aligned} \numberthis\label{1ODEvirus}
\]
given that the initial conditions converge in probability to the initial condition of the limiting differential equation.
Similarly, for large $K$, the dynamics of $(N_{1a,t}^{2},N_{1d,t}^2)$ is close to the one of the unique solution to
\[ 
\begin{aligned}
\dot{n}_{1a,2}(t) & = n_{1a,1}(t)(\lambda_1-\mu_1-C n_{1a,2}(t))+\sigma n_{1d,2}(t)+r \eps, \\
\dot{n}_{1d,2}(t) & = qD\eps n_{1a,2}(t)-(\kappa\mu_1+\sigma) n_{1d,2}(t),
\end{aligned} \numberthis\label{2ODEvirus}
\]
where for both systems of ODEs, the corresponding initial condition is $(\barn1a,0)$. Both systems are such that all their coordinatewise nonnegative solutions are bounded (whereas the positive orthant is positively invariant under both systems). Further, if $\eps>0$ is sufficiently small, then the equilibrium $(0,0)$ is unstable, and there is a unique additional coordinatewise nonnegative equilibrium, given as $({\bar n}_{1a}^{(1)},{\bar n}_{1d}^{(1)}):=(\frac{\lambda_1-\mu_1-(C+D)\eps}{C},0)$ (for \eqref{1ODEvirus}) respectively as $({\bar n}_{1a}^{(2)},{\bar n}_{1d}^{(2)})$ characterized by
\[ {\bar n}_{1d}^{(2)} = \frac{q D {\bar n}_{1a}^{(2)} \eps}{\kappa\mu_1+\sigma} \]
and thus ${\barn1a}^{(2)}$ being the unique positive zero locus of the quadratic polynomial
\[ h(x)= x\Big(\lambda_1-\mu+\frac{\sigma q D\eps}{\kappa\mu_1+\sigma} - C x \Big)+r\eps =0 \]
(for \eqref{2ODEvirus}). Hence, we have that 
\[ \lim_{\eps \downarrow 0} ({\bar n}_{1a}^{(1)},{\bar n}_{1d}^{(1)}) =  \lim_{\eps \downarrow 0} ({\bar n}_{1a}^{(2)},{\bar n}_{1d}^{(2)}) = (\barn1a, 0), \]
and thus in particular
\[ 0 \leq {\bar n}_{1d}^{(2)} = \frac{q D {\bar n}_{1a}\eps}{\kappa\mu_1+\sigma} (1+o_\eps(1))=O(\eps). \]
We can therefore find $\eps_0>0$ and $b\geq 2$ such that for all $\eps \in (0,\eps_0)$, $j \in \{ 1,2\}$, and $\upsilon \in \{ 1a, 1d\}$ we have
\[ \big| \bar n_{\upsilon}-\bar n_{\upsilon}^{(j)} \big| \leq b \eps \qquad \text{ and } 0\notin [\barn1a-b\eps,\barn1a+b\eps ], \numberthis\label{bdefvirus}\]
where we wrote $\bar n_{1d}=0$. 

Now, thanks to a result about exit of jump processes from a domain by Freidlin and Wentzell \cite[Chapter 5]{FW84} (see~\cite[Section 4.2]{C06} for details in a very similar situation), 
there exist two families (over $K$) of Markov jump processes $(\widetilde { N}^1_{1a},\widetilde N^1_{1d})=((\widetilde {N}^1_{1a,t},\widetilde N^1_{1d,t}))_{t \geq 0}$ and $(\widetilde { N}^2_{1a},\widetilde N^2_{1d})=((\widetilde {N}^2_{1a,t},\widetilde N^2_{1d,t}))_{t \geq 0}$ with positive, bounded, Lipschitz continuous transition rates that are uniformly bounded away from 0 such that for
\[ \widetilde Q^{j}_{\eps}:=\inf \big\{ t \geq 0 \colon \widetilde N^j_{1a,t} \notin [\barn1a-\eps,\barn1a+\eps] \text{ or } \widetilde N^j_{1d,t} > \eps \big\}, \qquad j \in \{1,2\},~\eps>0, \]
there exists $V>0$ such that
\[ \P(Q^{1}_{b\eps}>\e^{KV}) = \P(\widetilde Q^{1}_{b\eps}>\e^{KV}) \underset{K \to \infty}{\longrightarrow} 1. \numberthis\label{FW1virus} \] 
Using similar arguments for $N_{1a}^2$, we derive that for $\eps>0$, $V>0$ small enough, we have that
\[ \P(Q^2_{b\eps}>\e^{KV})\underset{K \to \infty}{\longrightarrow} 1.  \numberthis\label{FW2virus} \]
Now, on the event $\{ Q_{b\eps} \leq T_0^2 \wedge T_{\eps}^2 \}$ we have $Q_{b\eps}\geq \widetilde Q_{b\eps}^{1} \wedge \widetilde Q_{b\eps}^{2}$. Using \eqref{FW1virus} and \eqref{FW2virus}, we derive that
\[ \limsup_{K \to \infty} \P \big( Q_{b\eps} \leq \e^{KV}, Q_{b\eps} \leq T_0^2 \wedge T_{\eps}^2 \big) = 0. \]
Moreover, using Markov's inequality,
\begin{align*}
    \P(Q_{b\eps} \leq T_0^2 \wedge T_{\eps}^2) & \leq \P \big( Q_{b\eps} \leq \e^{KV}, Q_{b\eps} \leq T_0^2 \wedge T_{\eps}^2 \big) + \P(Q_{b\eps} \wedge T_0^2 \wedge T_{\eps}^2 \geq \e^{KV} \big) \\ &  \leq \P \big( Q_{b\eps} \leq \e^{KV}, Q_{b\eps} \leq T_0^2 \wedge T_{\eps}^2 \big) + \e^{-KV} \E(Q_{b\eps} \wedge T_0^2 \wedge T_{\eps}^2 ). 
\end{align*}
Since we have
\[ \E \big[Q_{b\eps}\wedge T_0^2 \wedge T_{\eps}^2 \big] \leq \E \Big[\int_0^{Q_{b\eps}\wedge T_0^2 \wedge T_{\eps}^2} K (N_{1i,t}^K+N_{2,t}^K) \d t \Big], \]
it suffices to show that there exists $\widetilde C>0$ such that 
\[ \E \Big[\int_0^{Q_{b\eps}\wedge T_0^2 \wedge T_{\eps}^2}  K (N_{1i,t}^K+N_{2,t}^K) \d t \Big] \leq \widetilde C \eps K. \numberthis\label{expectedstoppingtimevirus} \]
This can be done similarly to \cite[Section 3.1.2]{C+19}, the only difference being again that type 1d is considered as `resident' and not as `mutant' type, but for the reader's convenience we provide the details. We claim that it is enough to show that there exists a function $g \colon (\smfrac{1}{K} \N_0)^4 \to \R$ defined as
\[ g(\widetilde n_{1a},\widetilde n_{1d},\widetilde n_{1i},\widetilde n_{2})=\gamma_{1i} \widetilde n_{1i}+\gamma_2  \widetilde n_2 \numberthis\label{gammavirus} \]
for suitably chosen $\gamma_{1i},\gamma_2 \in \R$, such that
\[ \Lcal g(\mathbf N_{t}^K) \geq N_{1i,t}^K+N_{2,t}^K \numberthis\label{largegeneratorvirus}\]
where $\Lcal$ is the infinitesimal generator of $(\mathbf N_t^K)_{t \geq 0}$. 
Indeed, if \eqref{largegeneratorvirus} holds, then thanks to Dynkin's formula we have
\[ 
\begin{aligned}
\E& \Big[\int_0^{Q_{b\eps} \wedge T_0^2 \wedge T_{\eps}^2} K (N_{1i,t}^K+N_{2,t}^K) \d t \Big]  \leq  \E \Big[\int_0^{Q_{b\eps} \wedge T_0^2 \wedge T_{\eps}^2} K\Lcal g(\mathbf N_{t}^K)\d t \Big] \\
&=\E\big[ K g(\mathbf N_{Q_{b\eps} \wedge T_0^2 \wedge T_{\eps}^2}^K)-Kg(\mathbf N_{0}^K) \big]  \leq (|\gamma_{1i}|+|\gamma_{2}|) (\eps K-1),
\end{aligned}
\]
which implies the existence of $\widetilde C>0$ such that \eqref{expectedstoppingtimevirus} holds for all $\eps>0$ small enough, independently of the signs of $\gamma_{1i},\gamma_2$. 
Here, Dynkin's formula can indeed be applied because $ \E [Q_{b\eps}\wedge T_0^2 \wedge T_{\eps}^2]$ is finite. That holds because given our initial conditions, with positive probability the single initial type 2 individual dies due to natural death within a unit length of time before any event of the process $\mathbf N_t^K$ occurs, and hence already $T_0^2$ is stochastically dominated by a geometric random variable, which has all moments. We now apply the infinitesimal generator $\mathcal L$ to the function $g$ introduced in \eqref{gammavirus} once again. The infinitesimal generator $\mathcal L$ is such that $\widetilde{\mathcal L}(\cdot)=\mathcal L(K\cdot)$ maps a bounded measurable function $h \colon \N^4 \to \R$ to $\widetilde{\mathcal L} h \colon \N^4 \to \R$ defined as follows
\[
\begin{aligned}
\widetilde{\mathcal L} h(x,y,z,w) & = (h(x+1,y,z,w)-h(x,y,z,w)) \lambda_1 x \\ & \qquad + (h(x-1,y,z,w)-h(x,y,z,w))(\mu_1+ C \smfrac{x(x+y+z)}{K} ) \\ & \qquad + (h(x-1,y+1,z,w)-h(x,y,z,w))\smfrac{Dqxw}{K}  \\ & \qquad + (h(x-1,y,z+1,w-1)-h(x,y,z,w))\smfrac{D(1-q)xw}{K} \\ & \qquad + (h(x+1,y-1,z,w)-h(x,y,z,w))\sigma y +(h(x,y-1,z,w)-h(x,y,z,w))\kappa\mu_1 y \\ & \qquad + (h(x,y,z-1,w+m)-h(x,y,z,w))v z + (h(x+1,y,z-1,w)-h(x,y,z,w)) r z \\ & \qquad + (h(x,y,z,w-1)-h(x,y,z,w))\mu_2 w.
\end{aligned}
\] 
This yields
\[ 
\begin{aligned}
\Lcal g(\mathbf N_t^K) & = D N_{1a,t}^KN_{2,t}^K ( (1-q)(\gamma_{1i}-\gamma_2)) +(m\gamma_2-\gamma_{1i}) v N_{1i,t}^K -\gamma_{1i} r N_{1i,t}^K -\gamma_2 \mu_2 N_{2,t}^K.
\end{aligned} \]
Hence, according to \eqref{gammavirus}, it suffices to show that there exist $\gamma_{1i},\gamma_2 \in \R$ such that the following system of inequalities is satisfied:
\begin{align*}
-(r+v)N_{1i,t}^K \gamma_{1i}+mv N_{2,t}^K \gamma_2 & >N_{1i,t}^K, \\
 (1-q)D N_{1a,t}^K N_{2,t}^K \gamma_{1i} -((1-q)D N_{1a,t}^K N_{2,t}^K+\mu_2) \gamma_2 & >N_{2,t}^K. \numberthis\label{lastofineq}
\end{align*}
We claim that as long as $t \leq Q_{b\eps} \wedge T_0^2 \wedge T_{\eps}^2$, the transpose $J_2^K(t)$ of the matrix 
\[ (J_2^K(t))^T = \begin{pmatrix} -r-v & m v \\  (1-q) D N_{1a,t}^K & -(1-q)DN_{1a,t}^K-\mu_2 \end{pmatrix}\]
is entrywise close to the $2\times 2$ mean matrix $J_2$ introduced in \eqref{J2def}. This is certainly true because the first rows of the two matrices are equal, and for $t \leq Q_{b\eps} \wedge T_0^2 \wedge T_{\eps}^2$ we have that 
\[ \big| \big( J_2^K(t)-J_2 \big)_{lj} \big| \leq D b\eps, \qquad  \forall j,l\in \{1,2\}. \numberthis\label{JJclosevirus} \]
Let us now choose $(\gamma_{1i},\gamma_2)$. 
Thanks to the assumption that $\lambda_1>\mu_1$, given that $\eps>0$ is sufficiently small, $J_2^K(t) + (r+v+(1-q)D\barn1a+\mu_2)\mathrm{Id} $ is a matrix with positive entries, where $\mathrm{Id}$ denotes the $2\times 2$ identity matrix. Hence, writing $u_0= r+v+(1-q)D\barn1a+\mu_2$, it follows from the Perron--Frobenius theorem that there exists a strictly positive right eigenvector $\widetilde \Gamma=(\widetilde \gamma_{1i},\widetilde\gamma_{2})$ of $J_2+u_0\mathrm{Id}$ corresponding to the eigenvalue $\widetilde \lambda +u_0$. Then we have 
\[ (J_2+u_0 \mathrm{Id}) \widetilde \Gamma^T = (\widetilde \lambda + u_0)\widetilde \Gamma^T, \]
and thus also $J_2\widetilde \Gamma^T = \widetilde \lambda \widetilde \Gamma^T$. Since by assumption $\widetilde \lambda \neq 0$ and $\widetilde \Gamma$ has two positive coordinates, we obtain that
\[ \Gamma: = (\gamma_{1i},\gamma_{2}):=2(\widetilde\lambda(\widetilde \gamma_{1i} \wedge \widetilde \gamma_2))^{-1} \widetilde \Gamma\]
is well-defined, and it solves 
\[ J_2\Gamma^T = \widetilde \lambda \Gamma^T, \numberthis\label{GammaEVeqvirus} \] further, $\widetilde\lambda \gamma_j \geq 2$ holds for all $j \in \{ 1i,2\}$. Now, using \eqref{JJclosevirus} and \eqref{GammaEVeqvirus}, we obtain
\[
\begin{aligned}
\big| & ((1-q)D N_{1a,t}^K (\gamma_{1i}-\gamma_2)-\mu_2\gamma_2)-((1-q)D\barn1a(\gamma_{1i}-\gamma_2)-\mu_2\gamma_2) \big| \leq Db\eps (|\gamma_{1i}|+|\gamma_{2}|). 
\end{aligned}\]
Finally, since $\widetilde \lambda \gamma_j \geq 2$ holds for all $j\in\{1i,2\}$, it follows that if $\eps>0$ is small enough, then as long as $t \leq Q_{b\eps} \wedge T_0^2 \wedge T_{\eps}^2$, the inequality~\eqref{lastofineq} is satisfied. This, together with the fact that $J_2$ and $J_2^K(t)$ have the same first rows, implies \eqref{largegeneratorvirus}, and hence the proof of Lemma~\ref{lemma-residentsstayvirus} is concluded. 
\end{proof}
\begin{proof}[Proof of Proposition~\ref{prop-firstphasevirus}]
Now, similarly to the proof of Proposition~\cite[Proposition 3.1]{C+19}, we consider our population process on the event
\[ A_\eps := \{ T_0^2 \wedge T_{\eps}^2 < Q_{b\eps} \} \numberthis\label{Aepsdefvirus}\]
for sufficiently small $\eps>0$, where we fix $b \geq 2$ corresponding to Lemma~\ref{lemma-residentsstayvirus} for the rest of the proof. On this event, the invasion or extinction of the type 1i and 2 populations will happen before the rescaled type 1a population leaves a small neighbourhood of the equilibrium $\barn1a$ of radius $b\eps$ or the rescaled type 1d population reaches size $b\eps$. On $A_\eps$ we couple the process $(KN_{1i,t}^K,KN_{2,t}^K)$ with two 2-type branching processes $N^{\eps,-}=((N_{1i,t}^{\eps,-},N_{2,t}^{\eps,-}))_{t \geq 0}$ and $N^{\eps,+}=((N_{1i,t}^{\eps,+},N_{2,t}^{\eps,+}))_{t \geq 0}$ (which do not depend on $K$) such that almost surely, for any $0\leq t <  t_\eps:=T_0^2 \wedge T_{\eps}^2 \wedge Q_{b\eps}$,
\begin{align}
N_{j,t}^{\eps,-} & \leq \widehat N_{j}(t) \leq N_{j,t}^{\eps,+}, \qquad \forall j \in \{1i,2\}, \text{ and} \label{branchingcouplingvirus} \\ 
N_{j,t}^{\eps,-} & \leq K N_{j,t}^K \leq N_{j,t}^{\eps,+}, \qquad \forall j \in \{1i,2\} \label{originalcouplingvirus}
\end{align} 
where we again recall the approximating branching process $(\widehat{\mathbf N}(t))_{t \geq 0}=((\widehat N_{1d,t},\widehat N_{1i,t},\widehat N_{2,t}))_{t \geq 0}$ defined in Section~\ref{sec-phase1heuristics}, more precisely its two-dimensional projection $((\widehat N_{1i,t},\widehat N_{2,t}))_{t \geq 0}$.
We claim that in order to satisfy \eqref{branchingcouplingvirus} and \eqref{originalcouplingvirus}, these processes can be defined as follows: $N^{\eps,-}$ having transition rates 
\begin{itemize}
\item $(y,z) \to (y+1,z-1)$ at rate $(1-q)D(\barn1a-b\eps) z$,
\item $(y,z) \to (y-1,z)$ at rate $ry$,
\item $(y,z) \to (y-1,z+m)$ at rate $vy$,
\item $(y,z) \to (y,z-1)$ at rate $\mu_2 z+2 b \eps(1-q) Dz$;
\end{itemize} 
and $N^{\eps,+}$ having transition rates
\begin{itemize}
\item $(y,z) \to (y+1,z-1)$ at rate $(1-q)D(\barn1a-b\eps) z$,
\item $(y,z) \to (y+1,z)$ at rate $2b\eps (1-q)Dz$,
\item $(y,z) \to (y-1,z)$ at rate $ry$,
\item $(y,z) \to (y-1,z+m)$ at rate $vy$,
\item $(y,z) \to (y,z-1)$ at rate $\mu_2 z$,
\end{itemize} 
for $y,z \in \N_0$. 
In order to ensure that the coupling~\eqref{branchingcouplingvirus}--\eqref{originalcouplingvirus} holds, we define $N^{\eps,-}$, $N^{\eps,+}$, $(\widetilde{\mathbf N}(t))$, and $\mathbf N_{t}$ for $t \in [0,t_\eps)$ using the same Poissonian construction, which we provide in Appendix~\ref{sec-coupling}, where we assume that in accordance with the initial condition of Theorem~\ref{thm-viruscoexprob}, the last (virus) coordinate of all these branching processes equals 1 and their penultimate (infected host) coordinate equals 0 at time 0, so that the coupling holds for $t=0$.

Let us now argue that this construction is suitable for the coupling~\eqref{branchingcouplingvirus} and \eqref{originalcouplingvirus}. All the four processes involved in these two chains of inequalities have the same rates for reproduction, death by lysis, and death of viruses. The only difference is the infection mechanism in case of successful virus attacks. Indeed, for all the four processes involved, starting from a state $(y,z) \in \N_0^2$, at rate at least $(1-q)D(\barn1a-b\eps) z$ a successful virus attack happens. However, at additional rate $2b\eps (1-q) D z$, for the two coupled branching processes different events occur: $N^{\eps,-}$ has additional death of viruses without infecting any active individual, whereas $N^{\eps,+}$ has additional successful virus attacks where the involved virus does not even die. At the same time, for the approximating branching process $((\widehat N_{1i,t},\widehat N_{2}(t)))_{t \geq 0}$, at additional rate $b\eps(1-q) D z$, successful virus attacks happen. These are less beneficial for types 1i and 2 than the additional events of $N^{\eps,+}$, since each successful virus attack kills a virus, and the total rate of virus attacks is reduced. 
Hence the second inequality in \eqref{branchingcouplingvirus}. On the other hand, any of these additional virus attacks of $((\widehat N_{1i,t},\widehat N_{2}(t)))_{t \geq 0}$ is better for the type 1i and 2 population than the additional virus deaths of $N^{\eps,-}$, and even no event at all is better than those virus deaths, which implies the first inequality in \eqref{branchingcouplingvirus}. As for our original population process, on the event $A_\eps$, on the time interval $[0,T_0^2 \wedge T_{\eps}^2 \wedge Q_{b\eps}]$, starting from $(y,z)$, $(N_{1i,t}^K,N_{2,t}^K)$ has successful virus attacks at rate between $(1-q)D(\barn1a-b\eps) z$ and $(1-q)D(\barn1a+b\eps) z$. Hence, in \eqref{originalcouplingvirus}, the second inequality follows analogously to the one of \eqref{branchingcouplingvirus}. Recalling that an additional successful virus attack or an event where nothing happens is better for types 1i and 2 than a virus death (without virus attack), also the first inequality of \eqref{originalcouplingvirus} follows.

For $\diamond \in \{ +,-\}$, let $s_2^{(\eps,\diamond)}$ denote the extinction probability of the process $N_{t}^{\eps,\diamond}$ started from $ N_{0}^{\eps,\diamond}=(0,1)$. The extinction probability of a 2-type branching process
having the same kind of transitions as $(\widehat {N}_{1i}(t),\widehat N_{2}(t))_{t \geq 0}$ is continuous with respect to the rates of virus-induced dormancy, death of dormant individuals, resuscitation of dormant individuals, infection, recovery, lysis, and death of viruses, further, if additionally there are $(y,z)\to(y+1,z)$ (i.e.\ `infection without death of the involved virus') type transitions, then the extinction probability is also continuous with respect to the rate of these. These assertions can be proven analogously to \cite[Section A.3]{C+19}.

Hence, it follows from \eqref{branchingcouplingvirus} that for fixed $\eps>0$,
\[ s_2^{(\eps,+)} \leq s_2 \leq s_2^{(\eps,-)} \]
and for $\diamond \in \{ +,-\}$,
\[ 0 \leq \liminf_{\eps \downarrow 0} \big| s_2^{(\eps,\diamond)}-s_2 \big| \leq \limsup _{\eps \downarrow 0} \big| s_2^{(\eps,\diamond)}-s_2 \big| \leq \limsup_{\eps \downarrow 0} \big| s_2^{(\eps,-)}-s_2^{(\eps,+)} \big| = 0, \numberthis\label{qineqvirus} \]
where we recall the extinction probability $s_2$ of the approximating branching process $((\widehat N_{1i}(t),\widehat N_2(t)))_{t \geq 0}$ started from $(0,1)$ (equivalently, the one of $(\widehat {\mathbf N}(t))_{t \geq 0}$ started from $(0,0,1)$, see \eqref{qdefvirus}). 

Next, we show that the probabilities of extinction and invasion of the infected and virus coordinates $((N_{1i,t}^K,N_{2,t}^K))_{t \geq 0}$ of the original population process $(\mathbf N_t^K)_{t \geq 0}$ started from $(m_{1a}^K,0,0,1/K)$ also converge to $s_2$ and $1-s_2$, respectively, with high probability as $K \to \infty$. We define the stopping times, for $\diamond \in \{ +,- \}$,
\[ T_x^{(\eps,\diamond),2} := \inf \{ t >0 \colon N^{(\eps,\diamond)}_{1i,t}+N^{(\eps,\diamond)}_{2,t} = \lfloor K x \rfloor \}, \qquad x \geq 0. \]
Using the coupling in \eqref{originalcouplingvirus}, which is valid on $A_\eps$, we have
\[ \P\big(T_{\eps}^{(\eps,-),2} \leq T_0^{(\eps,-),2}, A_\eps\big) \leq \P\big(T_{\eps}^{2} \leq T_0^{2},A_\eps \big) \leq  \P\big(T_{\eps}^{(\eps,+),2} \leq T_0^{(\eps,+),2}, A_\eps\big). \numberthis\label{couplednonextinctionvirus} \] 
Indeed, if a process reaches the size $K\eps$ before dying out, then the same holds for a larger process as well. Now, we obtain
\[ \limsup_{K \to \infty}  \P\big( T_{\eps}^{(\eps,+),2}\leq T_0^{(\eps,+),2}, A_\eps \big) \leq \limsup_{K \to \infty}   \P\big( T_{\eps}^{(\eps,+),2} \leq T_{0}^{(\eps,+),2} \big) \leq (1-s_2^{(\eps,+)})(1+o_\eps(1)).\numberthis\label{eps+UB-virus} \]
Further,
\[ \limsup_{K\to\infty} \P \big( T_{0}^{(\eps,+),2}\leq T_\eps^{(\eps,+),2}, A_\eps \big)  \leq \limsup_{K \to \infty}   \P\big( T_{0}^{(\eps,+),2} \leq T_{\eps}^{(\eps,+),2} \big) \leq s_2^{(\eps,+)} (1+o_\eps(1)),  \]
and thus 
\[ 
\begin{aligned}
\liminf_{K\ \to\infty}  \P\big( T_{\eps}^{(\eps,+),2}\leq T_0^{(\eps,+),2}, A_\eps \big) &\geq \liminf_{K\to\infty} \P(A_\eps) - \limsup_{K\to\infty} \P \big( T_{0}^{(\eps,+),2} \leq T_\eps^{(\eps,+),2}, A_\eps \big) \\ &\geq (1-o_\eps(1))-s_2^{(\eps,+)}(1+o_\eps(1)) = (1-s_2^{(\eps,+)})(1-o_\eps(1)). 
\end{aligned} \numberthis\label{eps-LB-virus} \] \color{black}

Letting $K \to \infty$ in \eqref{couplednonextinctionvirus} and applying \eqref{eps+UB-virus} and \eqref{eps-LB-virus}  yields that 
\[
\begin{aligned}
(1-s_2^{(\eps,-)})(1-o_\eps(1)) & \leq \liminf_{K \to \infty} \P\big( T_{\eps}^{(\eps,-),2} \leq T_0^{(\eps,-),2}, A_\eps \big) \leq \limsup_{K \to \infty}  \P\big( T_{\eps}^{(\eps,+),2}\leq T_0^{(\eps,+),2}, A_\eps \big) \\ & \leq (1-s_2^{(\eps,+)})(1+o_\eps(1)).
\end{aligned} \]
Hence,
\[ \limsup_{K \to \infty} \big| \P(T_{\eps}^{2} \leq T_0^{2},A_\eps )-(1-s_2) \big|=o_\eps(1), \]
as required. The equation \eqref{secondofpropvirus} can be derived similarly. 

Finally, we show that in the case when the epidemic becomes macroscopic (which happens with probability tending to $1-s_2$), the time before the mutant population reaches size $K \eps$ is of order $\log K/\widetilde \lambda$, where we recall the largest eigenvalue $\widetilde\lambda$ of the mean matrix $J_2$, which was defined in \eqref{lambdatildedefvirus}. Having \eqref{secondofpropvirus}, we can without loss of generality assume that $s_2<1$, which is equivalent to the condition \eqref{viruscoexcond} and the one $\widetilde\lambda>0$.

For $\diamond \in \{ +,- \}$, let $\widetilde\lambda^{(\eps,\diamond)}$ denote the largest eigenvalue of the mean matrix corresponding to the branching process $N^{\eps,\diamond}$. Since $s_2<1$, this eigenvalue is positive for all sufficiently small $\eps>0$ and converges to $\widetilde\lambda$ as $\eps \downarrow 0$. In other words, there exists a nonnegative function $f \colon (0,\infty) \to (0,\infty)$ with $\lim_{\eps \downarrow 0} f(\eps)=0$ such that for any $\eps>0$ small enough,
\[ \Big| \frac{\widetilde\lambda^{(\eps,\diamond)}}{\widetilde\lambda}-1\Big| \leq \frac{f(\eps)}{2}. \numberthis\label{eigenvaluesclose}\]
Let us fix $\eps$ small enough such that \eqref{eigenvaluesclose} holds. Then from the coupling \eqref{originalcouplingvirus} we deduce that
\[ \P \Big( T^{(\eps,-),2}_{\eps} \leq T^{(\eps,-),2}_{0} \wedge \frac{\log K}{\widetilde\lambda} (1+f(\eps)),A_\eps \Big) \leq \P \Big( T^{2}_{\eps} \leq T^{2}_{0} \wedge \frac{\log K}{\widetilde \lambda} (1+f(\eps)),A_\eps \Big). \]
Using this 
and employing \cite[Section 7.5 in Chapter V]{AN72}, we obtain for $\eps>0$ small enough (in particular such that $f(\eps)<1$) 
\[ 
\begin{aligned}
& \liminf_{K \to \infty} \P \Big( T^{(\eps,-),2}_{\eps} \leq T^{(\eps,-),2}_{0} \wedge \frac{\log K}{\widetilde \lambda} (1+f(\eps)),A_\eps \Big)
\\ &\geq \liminf_{K \to \infty} \P \Big( T^{(\eps,-),2}_{\eps} \leq T^{(\eps,-),2}_{0} \wedge \frac{\log K}{\widetilde \lambda} (1+f(\eps))  \Big) - \limsup_{K\to\infty} \P\big( A_\eps^c \big)
\\ &= \liminf_{K \to \infty} \P \Big( T^{(\eps,-),2}_{\eps} \leq \frac{\log K}{\widetilde \lambda} (1+f(\eps)) \Big) - o_\eps(1) \\
&\geq  \liminf_{K \to \infty} \P \Big( T^{(\eps,-),2}_{\eps} \leq \frac{\log K}{\widetilde \lambda^{(\eps,-)}} \big( 1-\frac{f(\eps)}{2} \big) (1+f(\eps))  \Big) - o_\eps(1) \\ & \geq  \liminf_{K \to \infty} \P \Big( T^{(\eps,-),2}_{\eps} \leq \frac{\log K}{\widetilde \lambda^{(\eps,-)}} \Big) - o_\eps(1)
\geq  (1-s_2^{(\eps,-)})(1-o_\eps(1)),
\end{aligned} \numberthis\label{festimatesvirus}
\]
where in the second line we used that $\P(E \cap A_\eps) \geq \P(E)-\P(A_\eps^c)$ holds for any event $E$.

\color{black}
Similarly, using the coupling \eqref{originalcouplingvirus}, we derive that for all sufficiently small $\eps>0$ 
\begin{align*}
    \P \Big( T^{(\eps,+),2}_{\eps} \leq T^{(\eps,+),2}_{0}, T^{(\eps,+),2}_{\eps}  \geq \frac{\log K}{\widetilde \lambda} (1-f(\eps)),A_\eps \Big) \geq \P \Big( T^{2}_{\eps} \leq T^{2}_{0},T^{2}_{\eps} \geq \frac{\log K}{\widetilde \lambda} (1-f(\eps)),A_\eps \Big),
\end{align*}
and arguments analogous to the ones used in \eqref{festimatesvirus} imply that
\[ \liminf_{K \to \infty} \P \Big( T^{(\eps,+),2}_{\eps} \leq T^{(\eps,+),2}_{0}, T^{(\eps,+),2}_{\eps}  \geq \frac{\log K}{\widetilde \lambda} (1-f(\eps)),A_\eps \Big) \geq (1-s_2^{(\eps,+)})(1-o_\eps(1)). \]
These together imply \eqref{invasiontimevirus}, which concludes the proof of the proposition.
\end{proof}

\subsection{Proof of Theorems~\ref{thm-viruscoexprob}, \ref{thm-successvirus}, and \ref{thm-failurevirus}}\label{sec-finalproofvirus}

Putting together Propositions~\ref{prop-firstphasevirus} and~\ref{prop-Lyapunov} and Corollary~\ref{cor-persistence}, we now prove our main results, employing some arguments from \cite[Section 3.4]{C+19}, somewhat similarly to the case of stable coexistence in \cite[Section 6.4]{BT20}. The differences from these proofs stem from the fact that in the model of the present paper there is no case where the formerly resident type (here type 1a) goes extinct, and that we cannot verify a convergence of the dynamical system to the convergence equilibrium. Recall from Proposition~\ref{prop-thereisbifurcation} that in some cases, such a convergence cannot hold for large $m$, due to a Hopf bifurcation. 

Our proof strongly relies on the coupling \eqref{branchingcouplingvirus}--\eqref{originalcouplingvirus}. To be more precise, we define a Bernoulli random variable $B$ as the indicator of nonextinction
\[ B:= \mathds 1_{\{ \forall t>0, \widehat N_{1i}(t)+\widehat N_{2}(t)>0 \}} \]
of the two-type approximating branching process $((\widehat N_{1i}(t),\widehat N_{2}(t)))_{t \geq 0}$ defined Section~\ref{sec-phase1heuristics}, which is initially coupled between the same two branching processes $N^{\eps,-}$ and $N^{\eps,+}$ as $((K N_{1i,t}^K, KN_{2,t}^K))_{t \geq 0}$, according to \eqref{branchingcouplingvirus}. 
Let $f$ be a function such that Proposition~\ref{prop-firstphasevirus} holds for $f/3$ (and hence also for $f$). Throughout the rest of the proof, we will assume that $\eps>0$ is so small that $f(\eps) <1$, further, we fix $b\geq 2$ such that Proposition~\ref{prop-firstphasevirus} holds for $b$.

Our goal is to show that
\[ \liminf_{K \to \infty} \mathcal E(K,\eps) \geq s_2-o_\eps(1) \numberthis\label{extinctionlowervirus}\]
holds for 
\[ \mathcal E(K,\eps):= \P \Big( \frac{T_0^2}{\log K} \leq f(\eps), T_0^2 < T_{S_\beta}, B=0 \Big),\]
where we recall the stopping times $T_0^2$ and $T_{S_\beta}$ from Section~\ref{sec-resultsvirus}.
Further, we want to show that in case $s_2<1$, 
\[ \liminf_{K \to \infty} \mathcal I(K,\eps) \geq 1-s_2-o_\eps(1), \numberthis\label{survivallowervirus}\]
where we define
\[ \mathcal I(K,\eps):=\P \Big( \Big| \frac{T_{S_\beta}  \wedge T_0^2}{\log K}- \frac{1}{\widetilde \lambda}  \Big| \leq f(\eps), T_{S_\beta} < T_0^2, B=1 \Big). \]
Throughout the proof, $\beta>0$ is to be understood as sufficiently small; later we will explain what conditions precisely it has to satisfy.

The assertions~\eqref{extinctionlowervirus} and~\eqref{survivallowervirus} together will imply Theorem~\ref{thm-viruscoexprob}, Theorem~\ref{thm-successvirus}, and Equation \eqref{extinctionvirus} in Theorem~\ref{thm-failurevirus}. The other assertion of Theorem~\ref{thm-failurevirus}, Equation \eqref{lastoftheoremvirus}, follows already from \eqref{secondofpropvirus}.

Let us start with the case of extinction of the epidemic in the first phase of the invasion and verify \eqref{extinctionlowervirus}. Clearly, we have 
\[ \mathcal E(K,\eps)\geq \P \Big( \frac{T_0^2}{\log K} \leq f(\eps), T_0^2 < T_{S_\beta}^2, B=0, T_0^2 < T_{\eps}^2 \wedge Q_{b\eps} \Big), \]
where we recall the stopping times $T_0^2$ and $T_\eps^2$ from Section~\ref{sec-phase1virus}.
Now, considering our initial conditions, one can choose $\beta>0$ sufficiently small such that for all sufficiently small $\eps>0$ we have
\[ T_{\eps}^2 \wedge Q_{b\eps} < T_{S_\beta}, \]
almost surely. We assume further on during the whole section that $\beta$ satisfies this condition. Then,
\[ \mathcal E(K,\eps)\geq \P \Big( \frac{T_0^2}{\log K} \leq f(\eps), B=0, T_0^2 < T_{\eps}^2 \wedge Q_{b\eps} \Big). \numberthis\label{andisandvirus} \]
Moreover, similarly to the proof of Proposition~\ref{prop-firstphasevirus}, we obtain
\[ \limsup_{K \to \infty} \P \big( \{ B=0 \} \Delta \{ T_0^2 < T_{\eps}^2 \wedge Q_{b\eps}  \}  \big)=o_\eps(1), \numberthis\label{undefinedsymmdiffvirus} \]
where $\Delta$ denotes symmetric difference, and
\[ \limsup_{K \to \infty} \P \big( \{ B=0 \} \Delta \{ T_0^{(\eps,+),2} < \infty \}  \big)=o_\eps(1). \]
Together with \eqref{andisandvirus} and the coupling~\eqref{branchingcouplingvirus}--\eqref{originalcouplingvirus}, it follows that
\begin{align*}
\liminf_{K \to \infty} \mathcal E(K,\eps)  &\geq \liminf_{K\to\infty} \P \Big( \frac{
 T_0^2}{\log K} \leq f(\eps), B=0, T_0^2 \leq T_\eps^2 \wedge Q_{b\eps} \Big) \\
 &\geq \liminf_{K\to\infty} \P \Big( \frac{
 T_0^{(\eps,+),2}}{\log K} \leq f(\eps), B=0, T_0^2 \leq T_\eps^2 \wedge Q_{b\eps} \Big)
 \numberthis\label{secondline} \\ &\geq  \liminf_{K\to\infty} \P \Big( \frac{
T_0^{(\eps,+),2}}{\log K} \leq f(\eps), T_0^{(\eps,+),2} < \infty \Big) + o_\eps(1),
\end{align*}
Thus, employing \eqref{qineqvirus}, we obtain \eqref{extinctionlowervirus}, which implies \eqref{extinctionvirus}. 

Let us continue with the case of persistence of the epidemic and verify \eqref{survivallowervirus}. 
Let us fix a constant $b \geq 2$ satisfying the condition of Lemma~\ref{lemma-residentsstayvirus}. Arguing analogously to \eqref{undefinedsymmdiffvirus}, we get
\[ \limsup_{K \to \infty} \P \big( \{ B=1 \} \Delta \{ T_{\eps}^2  < T_0^2 \wedge Q_{b\eps}  \}  \big)=o_\eps(1). \] 
Thus,
\begin{equation}\label{beforesetsviruscoex}
\begin{aligned}
\liminf_{K \to \infty} \mathcal I(K,\eps) & = \liminf_{K \to \infty} \P \Big( \Big| \frac{T_{S_\beta}}{\log K} - \frac{1}{\widetilde \lambda}  \Big| \leq f(\eps), T_{S_\beta}<T_0^2,  
T^2_{\eps} < T^2_0 \wedge Q_{b\eps} \Big) + o_\eps(1).
\end{aligned}
\end{equation}
\begin{comment}{For $\eps>0,\beta>0$, we introduce the set \color{red} quantifier for $\delta$? This is also a problem in the homogamy paper (and in our papers) \color{black}
\[ \begin{aligned}
\mathfrak B^1_\eps &:= [\pi_{1d}-\delta,\pi_{1d}+\delta] \times [\pi_{1i}-\delta,\pi_{1i}+\delta] \times [\pi_{2}-\delta,\pi_{2}+\delta] \times [\eps/\widehat C,\eps] \times [\barn1a-b\eps,\barn1a+b\eps] \\
\end{aligned}
\]
and the stopping time
\[
\begin{aligned}
T'_\eps:=& \inf \Big\{ t \geq 0 \colon \Big( \frac{N_{1d,t}^K}{N_{1d,t}^K+N_{1i,t}^K+N_{2,t}^K}, \frac{N_{1i,t}^K}{N_{1d,t}^K+N_{1i,t}^K+N_{2,t}^K}, \frac{N_{2,t}^K}{N_{1d,t}^K+N_{1i,t}^K+N_{2,t}^K}, \\ & \qquad N_{1d,t}^K+N_{1i,t}^K+N_{2,t}^K, N_{1a,t}^K \Big)\in \mathfrak B^1_\eps \Big\}.
\end{aligned}
\]
Informally speaking, we want to show that with high probability the process has to pass through $\Bcal^1_\eps$ in order to reach $S_\beta$. Then, thanks to the Markov property, we can estimate $T_{S_\beta}$ by estimating $T'_\eps$ and $T_{S_{\beta}}-T'_\eps$. }\end{comment}
Now, \eqref{beforesetsviruscoex} implies that
\begin{align*}
    \liminf_{K \to \infty}\mathcal I & (K,\eps) \geq \liminf_{K\to\infty}   \P \Big( \Big| \frac{T_{S_\beta}}{\log K} -\frac{1}{\widetilde \lambda}\Big| \leq f(\eps), T^2_{\eps} < T^2_0 \wedge Q_{b\eps}, T_{S_\beta}<T_0^2 \Big) + o_\eps(1) \\
    \geq &  \liminf_{K\to\infty} \P \Big( \Big| \frac{T_\eps^2}{\log K} -\frac{1}{\widetilde \lambda} \Big| \leq \frac{f(\eps)}{3}, \Big| \frac{T_{S_\beta} -T^2_\eps}{\log K} \Big| \leq \frac{f(\eps)}{3},  T^2_{\eps} < T^2_0 \wedge Q_{b\eps} , T_{S_\beta}<T_0^2 \Big) + o_\eps(1) ,
\end{align*}
Note that for $\beta>0$ sufficiently small and $\eps>0$ sufficiently small chosen accordingly, $Q_{b\eps} \leq T_{S_\beta}$ almost surely. We assume during the rest of the proof that $\beta$ satisfies this condition. Hence, defining
\[ M_\eps=\big\{ (n_{1a}^0,n_{1d}^0,n_{1i}^0,n_2^0) \in [0,\infty)^4 \colon |n_{1a}^0-\barn1a|\leq b\eps, |n_{1d}^0|\leq b\eps, n_{1i}^0+n_2^0=\eps \big\} \]
and for $K>0$
\[M_\eps(K)=\big\{ (n_{1a}^0,n_{1d}^0,n_{1i}^0,n_2^0) \in [0,\infty)^4 \colon |n_{1a}^0-\barn1a|\leq b\eps, |n_{1d}^0|\leq b\eps, n_{1i}^0+n_2^0=\smfrac{\lfloor \eps K \rfloor}{K} \big\}  ,\]
the strong Markov property applied at time $T_\eps^2$ implies 
\[ \begin{aligned}
  &  \liminf_{K \to \infty} \mathcal I(K,\eps)\geq \liminf_{K \to \infty} \Big[ \P \Big( \Big| \frac{T^2_\eps}{\log K} -\frac{1}{\widetilde \lambda} \Big| \leq \frac{f(\eps)}{3}, T^2_{\eps} < T^2_0 \wedge Q_{b\eps} \Big) \\
    &  \times \inf_{\begin{smallmatrix}(n_{1a}^0,n_{1d}^0,n_{1i}^0,n_{2}^0) \in M_\eps(K) \end{smallmatrix}} \P \Big(  \Big| \frac{T_{S_\beta} -T^2_\eps}{\log K} \Big| \leq \frac{f(\eps)}{3},  T_{S_\beta} < T_0^2  \Big| \mathbf N^K_0=(n_{1a}^0,n_{1d}^0,n_{1i}^0,n_{2}^0) \Big) \Big]\\
   \end{aligned} \numberthis\label{productformcoexvirus} \]
It remains to show that the right-hand side of \eqref{productformcoexvirus} is close to $1-s_2$ as $K \to \infty$ if $\eps$ is small. The fact the limes inferior of the first factor on the right-hand side of \eqref{productformcoexvirus} is at least $1-s_2+o_\eps(1)$ follows analogously to \eqref{beforesetsviruscoex} (since Proposition~\ref{prop-firstphasevirus} holds not only for $f$ but also for $f/3$).

Now, let us treat the second, nearly deterministic phase of the epidemic. 
For $\mathbf m=(m_{1a},m_{1d},m_{1i},m_{2}) \in [0,\infty)^4$, let $\mathbf n^{(\mathbf m)}=(n_{1a}^\mathbf m(t),n_{1d}^\mathbf m(t),n_{1i}^\mathbf m(t),n_2^\mathbf m(t))$ denote the unique solution of the dynamical system \eqref{4dimvirus} with initial condition $\mathbf m$. Note that for any $\eps>0$, for any initial condition $\mathbf n^0$ contained in $M_\eps$ or in $M_\eps(K)$ in some $K>0$, for any $t>0$, $\mathbf n^{\mathbf n^0}(t)$ is a suitable initial condition for Proposition~\ref{prop-Lyapunov} and Corollary~\ref{cor-persistence}. Indeed, starting from a coordinatewise nonnegative initial condition with nonzero active coordinate and nonzero sum of infected and virus coordinate, for all positive times, all coordinates of the solution of \eqref{4dimvirus} will be positive. Then, thanks to Corollary~\ref{cor-persistence} and the continuity of flows of this dynamical system with respect to the initial condition, we deduce that if $\beta,\eps>0$ are small enough and $K_0>0$ is large enough, then there exists $t_{\beta,\eps}>0$ such that for all $\mathbf n^0 \in M_\eps  \cup (\bigcup_{K\geq K_0} M_\eps(K))$ there exists $t \leq t_{\beta,\eps}$ satisfying
\[ n_{1d}^{\mathbf n^0}(t),n_{1i}^{\mathbf n^0}(t)>\beta, n_{1a}^{\mathbf n^0}(t) \in (\beta,\barn1a-\beta),  \text{ and } n_2^{\mathbf n^0} (t) \in \big( \beta, \frac{mv\barn1a}{\mu_2}-\beta \big) \]
in other words, $\mathbf n^{\mathbf n^0}(t) \in S_\beta^o$, where for $A \subseteq \R^4$, $A^o$ denotes the interior of the set $A$. We assume for the rest of the proof that $\beta,\eps,K$ satisfy this assumption.

Now, using \cite[Theorem 2.1, p.~456]{EK} and the strong Markov property, we conclude that if $\eps,\beta$ are sufficiently small, then the following holds:
\[
\begin{aligned}
\lim_{K\to\infty} \P\Big( T_{S_\beta}-T^2_\eps \leq t_{\beta,\eps} \Big) = \lim_{K \to \infty} & \P\Big( T_{S_\beta} \leq t_{\beta,\eps} \Big| \mathbf N_{0}^K \in M_\eps(K) \Big)  =1-o_\eps(1).
\end{aligned} \]
Thus, the second term on the right-hand side of \eqref{productformcoexvirus} is close to 1 when $K$ tends to $\infty$, $\beta$ is small, and $\eps>0$ is small enough chosen according to $\beta$. Hence, together with the fact that the first factor on the right-hand side of \eqref{productformcoexvirus} is asymptotically at least $1-s_2-o_\eps(1)$, we have obtained
\[ \liminf_{K \to \infty} \mathcal I(K,\eps) \geq 1-s_2-o_\eps(1), \]
which implies \eqref{invasionvirus} and \eqref{successvirus}.



\appendix
\section{Proof of Propositions~\ref{prop-thereisnobifurcation} and \ref{prop-firststable}}\label{sec-appendixproofs}
In the following, we carry out the proof of Proposition~\ref{prop-thereisnobifurcation}.
\begin{proof}[Proof of Proposition~\ref{prop-thereisnobifurcation}.]
Let us first start with the three-dimensional case, i.e.\ let us put $q=0$, consider the dynamical system \eqref{3dimvirus} and show that for $r>v$, $(n_{1a}^*,n_{1i}^*,n_2^*)$ is asymptotically stable for all sufficiently small $m>m^*$. From this, one can derive the assertion of the proposition regarding the four-dimensional system \eqref{4dimvirus} for small $q$ and large $r$ using continuity, analogously to the proof of Proposition~\ref{prop-thereisbifurcation} (where we considered the case of small $q$ and small $r$). 

Let $q=0$ and $r>v$. Now for fixed $m$, we consider the Jacobi matrix $A_m(n_{1a}^*,n_{1i}^*,n_2^*)$ corresponding to \eqref{3dimvirus} at the coexistence equilibrium, which is defined analogously to $A(n_{1a}^*,n_{1d}^*,n_{1i}^*,n_2^*)$ but with $q=0$ (where the coexistence equilibrium corresponds to the same $m$ as the Jacobi matrix), ignoring the dormant coordinate:
\[ A_m(n_{1a}^*,n_{1i}^*,n_2^*)  = \begin{pmatrix} \lambda_1-\mu_1-2 C n_{1a}^*-Cn_{1i}^*-D n_2^* & r-Cn_{1a}^* &- D n_{1a}^* \\ D n_2^* & -(r+v) &  D n_{1a}^* \\ -D n_2^* & mv & - D n_{1a}^*-\mu_2 \end{pmatrix}. \]
Note that several entries of this matrix depend on $m$, but we ignore this in the notation for simplicity.
We want to analyse the stability of $(n_{1a}^*,n_{1i}^*,n_2^*)$ for large $m$ using the Routh--Hurwitz criterion (see e.g.~\cite[Section 3]{BK98}). We write the characteristic equation of $A_m(n_{1a}^*,n_{1i}^*,n_2^*)$ as 
\[ \lambda^3 + a_1(m) \lambda^2 + a_2(m) \lambda + a_3(m) =0. \]
Then, we have that $a_1(m)=-\mathrm{Tr}~A_m(n_{1a}^*,n_{1i}^*,n_2^*)$, $a_2(m)$ is the factor of $-\lambda$ in the characteristic polynomial when writing it as $\det(A_m(n_{1a}^*,n_{1i}^*,n_2^*)-\lambda I)$, and $a_3(m)=-\det A_m(n_{1a}^*,n_{1i}^*,n_2^*)$, where $I$ is the $3\times 3$ identity matrix. Then all eigenvalues of $A_m(n_{1a}^*,n_{1i}^*,n_2^*)$ have a strictly negative real part if and only if $a_1(m)>0$, $a_1(m)a_2(m)>a_3(m)$, and $a_3(m)>0$. Further, if $a_1(m)>0$, $a_3(m)>0$, and $a_1(m)a_2(m)<a_3(m)$, then the matrix is strictly unstable: since its determinant and trace are negative, its eigenvalue with the largest absolute value is negative, but the other two eigenvalues have strictly positive real parts (of course, these eigenvalues are either both real or complex conjugate).

Note that according to \eqref{n1adefvirus} (for $q=0$), we have $\lim_{m\to\infty} n_{1a}^*=0$. For $m>m^*$, the right-hand side of the first equation of~\eqref{3dimvirus} equal to zero, and dividing it with $n_{1a}^* >0$ we obtain
\[ \lambda_1-\mu_1-C(n_{1a}^*+n_{1i}^*)-D n_2^* = \frac{r}{r+v} D n_2^*. \]
Hence, it follows that
\[ \lim_{m\to\infty} \lambda_1-\mu_1-Cn_{1i}^* - \frac{v}{r+v} D n_2^*=0, \]
and thus in particular
\[ \limsup_{m\to\infty} n_2^* \leq \frac{\lambda_1-\mu_1}{D} \frac{r+v}{v}. \]
In particular, $n_2^*$ is bounded as a function of $m$. Hence, from~\eqref{infectedcoord} we conclude that $\lim_{m\to\infty} n_{1i}^*=0$ and thus
\[ \lim_{m\to\infty} n_2^* =  \frac{\lambda_1-\mu_1}{D} \frac{r+v}{v}. \]

Hence, we obtain
\[ \lim_{m\to\infty} a_1(m)= \lim_{m\to\infty} -(\lambda_1-\mu_1-D n_{2}^*)+r+v+\mu_2 = (\lambda_1-\mu_1) \frac{r}{v}+(r+v)+\mu_2>0, \]
further,
\[ \lim_{m\to\infty} a_2(m)= \lim_{m\to\infty} -(\lambda_1-\mu_1-D n_{2}^*) (r+v) - D n_2^* r + (\lambda_1-\mu_1) \frac{r}{v} \mu_2 + \mu_2(r+v)- D n_{1a}^* m v  = (\lambda_1-\mu_1) \frac{r}{v}\mu_2, \]
and
\[ \begin{aligned}
\lim_{m\to\infty} a_3(m) & =  -\lim_{m\to\infty}  (\lambda_1-\mu_1-Dn_2^*)(r+v)\mu_2-Dn_2^*n_{1a}^*mv-(\lambda_1-\mu_1-Dn_2^*)mvDn_{1a}^*+Dn_2^*r\mu_2 \\&=(\lambda_1-\mu_1)(r+v)\mu_2>0 
\end{aligned}\]
since $\lambda_1>\mu_1$. 
Therefore, for $r>v$ we obtain that
\[ 
\begin{aligned}
\liminf_{m\to\infty} a_1(m)a_2(m)-a_3(m) &> ((\lambda_1-\mu_1)+(r+v)+\mu_2)((\lambda_1-\mu_1)\mu_2)-(\lambda_1-\mu_1)(r+v)\mu_2 \\ & > (r+v)(\lambda_1-\mu_1)\mu_2-(\lambda_1-\mu_1)(r+v)\mu_2 =0. 
\end{aligned}\]
Thus, for $r>v$, for all $m>m^*$ sufficiently large, $(n_{1a}^*,n_{1i}^*,n_{2}^*)$ is asymptotically stable, as claimed.
\end{proof}
Next, we prove Proposition~\ref{prop-firststable}.
\begin{proof}[Proof of Proposition~\ref{prop-firststable}]
In this proof, we employ arguments similar to those in~\cite[Section 3]{BK98}. In order to verify the proposition, we have to show that for $m>m^*$ sufficiently close to $m$, all eigenvalues of $A(n_{1a}^*,n_{1d}^*,n_{1i}^*,n_2^*)$ have a strictly negative real part. To this aim, let us first study the extreme case when $m=m^*$, in other words, $(\barn1a,0,0,0)=(n_{1a}^*,n_{1d}^*,n_{1i}^*,n_2^*)$. Then, the Jacobi matrix $A(n_{1a}^*,n_{1d}^*,n_{1i}^*,n_2^*)=A(\barn1a,0,0,0)$ is given according to \eqref{Abarn1a}. Again, $-(\lambda_1-\mu)$ is a strictly negative eigenvalue of this matrix thanks to our assumptions. Using Laplace's expansion theorem, it follows that the remaining three eigenvalues of the matrix are $-\kappa\mu_1-\sigma<0$ and the two eigenvalues of the last $2\times 2$ block of the matrix. This block has a zero eigenvalue because $m=m^*$ is equivalent to the assertion that \eqref{viruscoexcond} holds with an equality instead of `$>$', in other words, the determinant of the $2\times 2$ block is zero. However, even in this case, the trace of the $2\times 2$ block is negative, and hence 0 is only a single eigenvalue of the block and hence also of the Jacobi matrix. The other eigenvalue equals the trace $-(r+v)-(1-q)D\barn1a-\mu_2$ of the block, and hence all eigenvalues are real.

Now, in the limit $m \downarrow m^*$ the trace of the Jacobi matrix $A(n_{1a}^*,n_{1d}^*,n_{1i}^*,n_2^*)$ remains negative and bounded away from zero, and the determinant tends to zero from above (since $n_2^*$ tends to 0 and $n_{1a}^*$ to $\barn1a$; cf.~\eqref{detcomp}) by continuity. In this limit, each eigenvalue of $A(n_{1a}^*,n_{1d}^*,n_{1i}^*,n_2^*)$ converges to the corresponding one of $A(\barn1a,0,0,0)$ corresponding to $m=m^*$. In particular, precisely one eigenvalue tends to 0, and hence this eigenvalue must be real. Hence, one eigenvalue tends to zero and three remain bounded away from zero, while the imaginary part of any of the eigenvalues has to tend to zero. This implies that for $m>m^*$ sufficiently close to $m^*$, there can be at most one eigenvalue with positive real part, namely the real one that converges to zero. But if there was precisely one such eigenvalue, this would contradict the assertion of the lemma that for all $m>m^*$, there are either 2 or 4 eigenvalues with positive real parts. Hence, we conclude that for all $m>m^*$ sufficiently small, all eigenvalues of $A(n_{1a}^*,n_{1d}^*,n_{1i}^*,n_2^*)$ have negative real parts.
\end{proof}

\section{Poissonian construction for the couplings involving branching processes}\label{sec-coupling}
We consider a family of  independent Poisson point processes with uniform intensity on $[0,\infty)^2$ as follows:
\begin{itemize}
\item $P_{1a \to 2\times 1a}(\d s,\d \theta)$ driving the birth of type 1a individuals,
\item $P_{1a \to \emptyset} (\d s,\d \theta)$ driving the death of type 1a individuals,
\item $P_{1a+2 \to 1d+2}(\d s, \d \theta )$ driving the unsuccessful virus attacks,
\item $P_{1a+2 \to 1i}(\d s, \d \theta )$ driving the successful virus attacks,
\item $P_{1d \to \emptyset}(\d s, \d \theta )$ driving the death of type 1d individuals, 
\item $P_{1d \to 1a}(\d s, \d \theta )$ driving the resuscitation of type 1d individuals, 
\item $P_{1i \to m\times 2}(\d s,\d \theta)$ driving the death of type 1i individuals by lysis,
\item $P_{1i\to 1a}(\d s,\d \theta)$ driving the recovery of type 1i individuals,
\item $P_{2\to\emptyset}(\d s,\d \theta)$ driving the death of type 2 individuals.
\end{itemize} 
Using these Poisson point processes, our process $((N_{1a,t},N_{1d,t},N_{1i,t},N_{2,t}))_{t\geq 0}$ is constructed as follows (writing $N_{1,t}=N_{1a,t}+N_{1d,t}+N_{1i,t}$ as before):
\begin{align*}
(N_{1a,t},N_{1d,t},N_{1i,t},N_{2,t}) & = (N_{1a,0},N_{1d,0},N_{1i,0},N_{2,0}) \\ &\quad  + (1,0,0,0)\int_0^t \int_0^{\infty}  \mathds 1_{\{ \theta \leq \lambda_1 N_{1a,s-} \}}(s,\theta) P_{1a \to 2\times 1a}(\d s,\d \theta) 
\\ & \quad + (-1,0,0,0) \int_0^t \int_0^{\infty} \mathds 1_{\{ \theta \leq N_{1a,s-}(\mu+\frac{C}{K} N_{1,s-} ) \}}(s,\theta) P_{1a \to \emptyset}(\d s,\d \theta) 
\\ & \quad + (-1,1,0,0) \int_0^t \int_0^{\infty}  \mathds 1_{\{ \theta \leq q\frac{D}{K}N_{1a,s-}( N_{1a,s-}N_{2,s-} ) \}}(s,\theta) P_{1a+2 \to 1d+2}(\d s,\d \theta) 
\\ & \quad + (-1,0,1,-1) \int_0^t \int_0^{\infty}  \mathds 1_{\{ \theta \leq(1-q)\frac{D}{K}N_{1a,s-}( N_{1a,s-}N_{2,s-} ) \}}(s,\theta) P_{1a+2 \to 1i}(\d s,\d \theta) 
\\ & \quad  + (0,-1,0,0) \int_0^t \int_0^{\infty} \mathds 1_{\{ \theta \leq \kappa\mu_1 N_{1d,s-} \}}(s,\theta) P_{1d\to\emptyset}(\d s,\d \theta) 
\\ & \quad  + (1,-1,0,0) \int_0^t \int_0^{\infty}  \mathds 1_{\{ \theta \leq  \sigma N_{1d,s-} \}}(s,\theta) P_{1d\to 1a}(\d s,\d \theta)
\\ & \quad  +(0,0,-1,m) \int_0^t \int_0^{\infty} \mathds 1_{\{ \theta \leq v N_{1i,s-} \}}(s,\theta) P_{1i\to m\times 2}(\d s,\d \theta)
\\ &   \quad  + (1,0,-1,0) \int_0^t \int_0^{\infty} \mathds 1_{\{ \theta \leq  r N_{1i,s-} \}}(s,\theta) P_{1i\to 1a}(\d s,\d \theta)
\\ &  \quad  + (0,0,0,-1) \int_0^t \int_0^{\infty} \mathds 1_{\{ \theta \leq  \mu_2 N_{2,s-} \}}(s,\theta) P_{2\to\emptyset}(\d s,\d \theta).
\end{align*}

In order to make the coupling equations~\eqref{branchingcouplingvirus} and~\eqref{originalcouplingvirus} hold on the event $A_\eps$ defined in~\eqref{Aepsdefvirus} and on the time interval $[0,t_\eps)$ where $t_\eps=T_0^2 \wedge T_\eps^2 \wedge Q_{b\eps}$, we construct the branching processes $N^{\eps,-}$, $N^{\eps,+}$, $((\widehat N_{1d,t},\widehat N_{1i,t},\widehat N_{2,t}))_{t\geq 0}$ using the same Poisson point processes as follows (in accordance with the transition rates appearing in the definition of these branching processes): we define $N^{\eps,-}=((N_{1i,t}^{\eps,-},N_{2,t}^{\eps,-}))_{t\geq 0}$ as
\begin{align*}
(N_{1i,t}^{\eps,-},N_{2,t}^{\eps,-}) & = (N_{1i,0}^{\eps,-},N_{2,0}^{\eps,-}) + (1,-1) \int_0^t \int_0^{\infty}  \mathds 1_{\{ \theta \leq(1-q)D N_{1i,s-}^{\eps,-} (\barn1a-b\eps) \}}(s,\theta) P_{1a+2 \to 1i}(\d s,\d \theta)
\\ & \qquad + (0,-1) \int_0^t \int_0^{\infty}  \mathds 1_{\{ \theta \leq 2b\eps(1-q)D N_{1i,s-}^{\eps,-}  \}}(s,\theta) P_{1a+2 \to 1i}(\d s,\d \theta)  
\\ & \qquad  + (-1,m) \int_0^t \int_0^{\infty}  \mathds 1_{\{ \theta \leq v N_{1i,s-}^{\eps,-} \}}(s,\theta) P_{1i\to m\times 2}(\d s,\d \theta)
\\ &   \qquad  + (-1,0) \int_0^t \int_0^{\infty} \mathds 1_{\{ \theta \leq  r N_{1i,s-}^{\eps,-} \}}(s,\theta) P_{1i\to 1a}(\d s,\d \theta)
\\ &  \qquad  + (0,-1) \int_0^t \int_0^{\infty}  \mathds 1_{\{ \theta \leq  \mu_2 N_{2,s-}^{\eps,-} \}}(s,\theta) P_{2\to\emptyset}(\d s,\d \theta),
\end{align*}
for $N^{\eps,+}=((N_{1i,t}^{\eps,+},N_{2,t}^{\eps,+}))_{t\geq 0}$ we put
\begin{align*}
(N_{1i,t}^{\eps,+},N_{2,t}^{\eps,+}) & = (N_{1i,0}^{\eps,+},N_{2,0}^{\eps,+}) + (1,-1) \int_0^t \int_0^{\infty} \mathds 1_{\{ \theta \leq(1-q)D N_{1i,s-}^{\eps,+} (\barn1a-b\eps) \}}(s,\theta) P_{1a+2 \to 1i}(\d s,\d \theta) 
\\ & \qquad + (1,0) \int_0^t \int_0^{\infty}  \mathds 1_{\{ \theta \leq  2 b\eps (1-q)D N_{1i,s-}^{\eps,+} \}}(s,\theta) P_{1a+2 \to 1i}(\d s,\d \theta) 
\\ & \qquad  + (-1,m)\int_0^t \int_0^{\infty}  \mathds 1_{\{ \theta \leq v N_{1i,s-}^{\eps,+} \}}(s,\theta) P_{1i\to m\times 2}(\d s,\d \theta)
\\ &   \qquad  + (-1,0)\int_0^t \int_0^{\infty}  \mathds 1_{\{ \theta \leq  r N_{1i,s-}^{\eps,+} \}}(s,\theta) P_{1i\to 1a}(\d s,\d \theta)
\\ &  \qquad  + (0,-1) \int_0^t \int_0^{\infty}  \mathds 1_{\{ \theta \leq  \mu_2 N_{2,s-}^{\eps,+} \}}(s,\theta) P_{2\to\emptyset}(\d s,\d \theta),
\end{align*}
and finally we define $\widehat{ \mathbf N}(t)=(\widehat N_{1d}(t),\widehat N_{1i}(t),\widehat N_{2}(t))$ as
\begin{align*}
(\widehat N_{1d}(t),\widehat N_{1i}(t),\widehat N_{2}(t)) & =  (1,0,0)\int_0^t \int_0^{\infty}  \mathds 1_{\{ \theta \leq qD\bar n_{1a}\widehat N_{2}(s-)  \}}(s,\theta) P_{1a+2 \to 1d+2}(\d s,\d \theta) 
\\ & \qquad + (0,1,-1) \int_0^t \int_0^{\infty}  \mathds 1_{\{ \theta \leq(1-q)D\bar n_{1a} \widehat N_{2}(s-) \}}(s,\theta) P_{1a+2 \to 1i}(\d s,\d \theta) 
\\ & \qquad  + (-1,0,0) \int_0^t \int_0^{\infty} \mathds 1_{\{ \theta \leq \kappa\mu_1 \widehat N_{1d}(s-) \}}(s,\theta) P_{1d\to\emptyset}(\d s,\d \theta) 
\\ & \qquad  + (-1,0,0) \int_0^t \int_0^{\infty}  \mathds 1_{\{ \theta \leq  \sigma \widehat N_{1d}(s-) \}}(s,\theta) P_{1d\to 1a}(\d s,\d \theta)
\\ & \qquad  + (0,-1,m) \int_0^t \int_0^{\infty} \mathds 1_{\{ \theta \leq v \widehat N_{1i}(s-) \}}(s,\theta) P_{1i\to m\times 2}(\d s,\d \theta)
\\ &   \qquad  + (0,-1,0)\int_0^t \int_0^{\infty} \mathds 1_{\{ \theta \leq  r \widehat N_{1i}(s-) \}}(s,\theta) P_{1i\to 1a}(\d s,\d \theta)
\\ &  \qquad  + (0,0,-1) \int_0^t \int_0^{\infty} \mathds 1_{\{ \theta \leq  \mu_2 \widehat N_{2}(s-) \}}(s,\theta) P_{2\to\emptyset}(\d s,\d \theta).
\end{align*}
Considering that the initial conditions for the virus (resp.\ infected host) coordinates of all the four processes are equal and that on the event $A_\eps$ for any $t \in [0,t_\eps)$ we have
\[ \bar n_{1a}-b\eps \leq N_{1a,t}/K \leq \bar n_{1a}+b\eps, \]
we conclude that \eqref{originalcouplingvirus} holds for all $t\in [0,t_\eps)$, while~\eqref{branchingcouplingvirus} actually holds for all $t\geq 0$.\color{black}

\subsection*{Acknowledgements} The authors thank two anonymous reviewers for insightful comments and F.~Gillich for interesting remarks that inspired Section~\ref{sec-R0}.

\end{document}